\documentclass[10pt]{article}

\usepackage{amsmath, amsfonts, amssymb, amsthm}
\usepackage[title]{appendix}
\usepackage[affil-it]{authblk}
\usepackage[font=footnotesize,labelfont=bf]{caption}
\usepackage[nocompress]{cite}
\usepackage{fancyhdr}
\usepackage{float}
\usepackage[margin=1in]{geometry}
\usepackage{graphicx}
\usepackage{hyperref}
\usepackage{mathtools}
\usepackage{subcaption}
\usepackage[nameinlink]{cleveref}

\title{\Large\textbf{Interpretable and flexible non-intrusive reduced-order models using reproducing kernel Hilbert spaces}}
\author[1]{Alejandro~N.~Diaz\thanks{Corresponding author. E-mail: \href{mailto:andiaz@sandia.gov}{andiaz@sandia.gov}.}}
\affil[1]{\normalsize Sandia National Laboratories}
\author[1,2]{Shane~A.~McQuarrie}
\affil[2]{\normalsize Department of Mathematics, Brigham Young University}
\author[1]{John~T.~Tencer}
\author[1]{Patrick~J.~Blonigan}

\date{}

\hypersetup{
    colorlinks=true,
    linkcolor=black,
    urlcolor=black,
    citecolor=black,
    pdftitle={Interpretable and flexible non-intrusive reduced-order models using reproducing kernel Hilbert spaces}
    pdfauthor={Alejandro N. Diaz, Shane A. McQuarrie, John T. Tencer, Patrick J. Blonigan},
}

\crefformat{equation}{#2(#1)#3}
\numberwithin{equation}{section}

\newtheorem{definition}{Definition}[section]
\newtheorem{proposition}{Proposition}[section]
\newtheorem{lemma}{Lemma}[section]
\newtheorem{theorem}{Theorem}[section]
\newtheorem{corollary}{Corollary}[section]

\theoremstyle{definition}
\newtheorem{remark}{Remark}[section]
\newtheorem{example}[theorem]{Example}

\newcommand{\apdxref}[1]{\hyperref[#1]{Appendix~\ref{#1}}}

\newcommand{\trp}{^{\mathsf{T}}}
\newcommand{\fun}{v}
\newcommand{\bfun}{\mathbf{v}}
\newcommand{\set}[1]{\left\{ #1 \right\} }
\newcommand{\norm}[1]{\left\| #1 \right\|}
\newcommand{\innerprod}[1]{\left\langle #1 \right\rangle}

\newcommand{\real}{\mathbb{R}}
\newcommand{\cD}{{\cal D}}
\newcommand{\cH}{{\cal H}}

\newcommand{\bzero}{\mathbf{0}}

\newcommand{\bb}{\mathbf{b}}
\newcommand{\bd}{\mathbf{d}}
\newcommand{\be}{\mathbf{e}}
\newcommand{\bff}{\mathbf{f}}
\newcommand{\bg}{\mathbf{g}}
\newcommand{\bh}{\mathbf{h}}
\newcommand{\bp}{\mathbf{p}}
\newcommand{\bq}{\mathbf{q}}
\newcommand{\bs}{\mathbf{s}}
\newcommand{\bu}{\mathbf{u}}
\newcommand{\bv}{\mathbf{v}}
\newcommand{\bw}{\mathbf{w}}
\newcommand{\bx}{\mathbf{x}}
\newcommand{\by}{\mathbf{y}}
\newcommand{\bz}{\mathbf{z}}
\newcommand{\bA}{\mathbf{A}}
\newcommand{\bB}{\mathbf{B}}
\newcommand{\bC}{\mathbf{C}}
\newcommand{\bD}{\mathbf{D}}
\newcommand{\bG}{\mathbf{G}}
\newcommand{\bH}{\mathbf{H}}
\newcommand{\bI}{\mathbf{I}}
\newcommand{\bL}{\mathbf{L}}
\newcommand{\bM}{\mathbf{M}}
\newcommand{\bQ}{\mathbf{Q}}

\newcommand{\bU}{\mathbf{U}}
\newcommand{\bV}{\mathbf{V}}
\newcommand{\bW}{\mathbf{W}}
\newcommand{\bX}{\mathbf{X}}
\newcommand{\bY}{\mathbf{Y}}
\newcommand{\bZ}{\mathbf{Z}}

\newcommand{\tbb}{\tilde{\mathbf{b}}}
\newcommand{\tbc}{\tilde{\mathbf{c}}}
\newcommand{\tbf}{\tilde{\mathbf{f}}}
\newcommand{\tbq}{\tilde{\mathbf{q}}}
\newcommand{\tbA}{\tilde{\mathbf{A}}}
\newcommand{\tbB}{\tilde{\mathbf{B}}}
\newcommand{\tbH}{\tilde{\mathbf{H}}}
\newcommand{\tbQ}{\tilde{\mathbf{Q}}}
\newcommand{\tbW}{\tilde{\mathbf{W}}}

\newcommand{\hbc}{\hat{\mathbf{c}}}
\newcommand{\hbe}{\hat{\mathbf{e}}}
\newcommand{\hbf}{\hat{\mathbf{f}}}
\newcommand{\hbq}{\hat{\mathbf{q}}}
\newcommand{\hbA}{\hat{\mathbf{A}}}
\newcommand{\hbH}{\hat{\mathbf{H}}}
\newcommand{\hbQ}{\hat{\mathbf{Q}}}
\newcommand{\hbL}{\hat{\mathbf{L}}}
\newcommand{\hbM}{\hat{\mathbf{M}}}
\newcommand{\hbO}{\hat{\mathbf{O}}}
\newcommand{\hbZ}{\hat{\mathbf{Z}}}

\newcommand{\obq}{\bar{\mathbf{q}}}
\newcommand{\obx}{\bar{\mathbf{x}}}

\newcommand{\bdelta}{{\boldsymbol{\delta}}}
\newcommand{\bGamma}{{\boldsymbol{\Gamma}}}
\newcommand{\bmu}{{\boldsymbol{\mu}}}
\newcommand{\bnu}{{\boldsymbol{\nu}}}
\newcommand{\bsigma}{{\boldsymbol{\sigma}}}
\newcommand{\bphi}{{\boldsymbol{\phi}}}
\newcommand{\bpsi}{{\boldsymbol{\psi}}}
\newcommand{\bSigma}{{\boldsymbol{\Sigma}}}
\newcommand{\bomega}{{\boldsymbol{\omega}}}
\newcommand{\bOmega}{{\boldsymbol{\Omega}}}

\newcommand{\dt}{\frac{\textrm{d}}{\textrm{d}t}}
\newcommand{\pdt}{\frac{\partial}{\partial t}}
\newcommand{\pdx}{\frac{\partial}{\partial x}}
\newcommand{\pdxx}{\frac{\partial^2}{\partial x^2}}

\begin{document}

\thispagestyle{empty}
\maketitle

\begin{abstract}
\noindent
This paper develops an interpretable, non-intrusive reduced-order modeling technique using regularized kernel interpolation.
Existing non-intrusive approaches approximate the dynamics of a reduced-order model (ROM) by solving a data-driven least-squares regression problem for low-dimensional matrix operators.
Our approach instead leverages regularized kernel interpolation, which yields an optimal approximation of the ROM dynamics from a user-defined reproducing kernel Hilbert space.
We show that our kernel-based approach can produce interpretable ROMs whose structure mirrors full-order model structure by embedding judiciously chosen feature maps into the kernel.
The approach is flexible and allows a combination of informed structure through feature maps and closure terms via more general nonlinear terms in the kernel.
We also derive a computable \emph{a posteriori} error bound that combines standard error estimates for intrusive projection-based ROMs and kernel interpolants.
The approach is demonstrated in several numerical experiments that include comparisons to operator inference using both proper orthogonal decomposition and quadratic manifold dimension reduction.
\end{abstract}

\vspace{2mm}
\noindent \emph{Keywords:} data-driven model reduction, kernel interpolation, feature maps, interpretable reduced-order model, error bounds, quadratic manifolds

\section{Introduction}
\label{sec:introduction}

Large-scale numerical simulations are a crucial component of the engineering design process.
For many applications, the complexity of the underlying physics and the required fidelity make such simulations highly computationally expensive, which renders many-query simulation tasks such as uncertainty quantification and design optimization infeasible.
Model reduction techniques seek to mitigate high computation costs in numerical simulations by systematically extracting the relevant dynamics of a large-scale system, called the full-order model (FOM), and constructing a low-dimensional, computationally efficient reduced-order model (ROM), which can be used as a substitute for the FOM in many-query design tasks.
Two appealing features of ROMs over other surrogate modeling techniques are that they aim to incorporate underlying physics from the FOM and often come equipped with rigorous error bounds.
In this paper, we propose a novel model reduction framework that uses regularized kernel interpolation to compute data-driven ROMs that are interpretable, flexible, and have rigorous error estimates.

Classical projection-based model reduction techniques construct ROMs by identifying a low-dimensional linear subspace that best represents the FOM dynamics in some sense, then projecting the governing equations onto the subspace. Examples
of projection-based approaches include balanced truncation~\cite{ACAntoulas_2005a,PBenner_TBreiten_2017a};
interpolatory projections~\cite{ACAntoulas_CABeattie_SGugercin_2020a,ANDiaz_IVGosea_MHeinkenschloss_ACAntoulas_2023a};
moment-matching~\cite{PBenner_TBreiten_2015a,CGu_2011a};
and proper orthogonal decomposition (POD)~\cite{MGubisch_SVolkwein_2017a,MHinze_SVolkwein_2005a}, in which the optimal low-dimensional subspace is defined as the span of the leading left singular vectors of a representative set of state data.
In recent years, several dimension reduction approaches have been proposed that aim to overcome approximation limitations of linear subspaces, including
nonlinear manifolds (NMs) using autoencoders~\cite{JCocola_JTencer_FRizzi_EParish_PBlonigan_2023a,ANDiaz_YChoi_MHeinkenschloss_2024a,YKim_YChoi_DWidemann_TZohdi_2022a,KLee_KTCarlberg_2020a,FRomor_GStabile_GRozza_2023a},
quadratic manifolds (QMs)~\cite{JBarnett_CFarhat_2022a,RGeelen_LBalzano_SWright_KWillcox_2024a,RGeelen_SWright_KWillcox_2023a,PSchwerdtner_BPeherstorfer_2024a},
and the projection-based ROM + artificial neural network (PROM-ANN) approach~\cite{JBarnett_CFarhat_YMaday_2023a}.
These strategies are especially beneficial when applied to problems with slowly decaying Kolmogorov $n$-width, such as transport-dominated problems or problems with sharp gradients \cite{MOhlberger_SRave_2016a,peherstorfer2022kolmogorov}.
In many cases, these nonlinear dimension reduction approaches can still be used to produce projection-based ROMs by inserting the state approximation into the governing FOM and projecting the residual by a test basis.

Projection-based methods have enjoyed success in a number of applications. However, a common disadvantage is that they require intrusive access to code of a given FOM.
This is often an infeasible request when the FOM is defined through legacy or commercial code, and hence an intrusive projection-based ROM is unobtainable.
Several so-called non-intrusive model reduction approaches have been developed recently to overcome this difficulty.
These methods apply a dimension reduction technique, such as POD or an autoencoder, to project pre-computed snapshot data onto a low-dimensional latent space and learn a function that models the system dynamics within the latent space.
For example, dynamic mode decomposition (DMD) \cite{JNKutz_SLBrunton_BWBrunton_JLProctor_2016a,CWRowley_IMezic_SBagheri_PSchlatter_DSHenningson_2009a,PJSchmid_2010a,PJSchmid_2022a,JHTu_CWRowley_DMLuchtenburg_SLBrunton_JNKutz_2014a} approximates a dynamical system by fitting a least-squares optimal linear operator to time series data. This approach has been extended to approximate nonlinear dynamical systems using Koopman operator theory, but selecting observables that yield approximately linear dynamics can be challenging \cite{SLBrunton_BWBrunton_JLProctor_JNKutz_2016a,IMezic_2005a,rosenfeld2023singulardmd,MOWilliams_IGKevrekidis_CWRowley_2015a}.
Operator inference (OpInf) \cite{ghattas2021acta,kramer2024survey,BPeherstorfer_KWillcox_2016a} is a related method that constrains the learnable dynamics to have the same structure (e.g, polynomial) as a projection-based ROM, thereby producing interpretable nonlinear ROMs.
In this method, reduced-order operators are computed by solving a linear least-squares regression that minimizes the residual of the desired reduced dynamics.
Non-polynomial nonlinearities can often be incorporated by first applying a lifting transformation to the training data, then learning a polynomial ROM \cite{EQian_BKramer_BPeherstorfer_KWillcox_2020a}.
By contrast, neural network (NN)-based approaches \cite{KBhattacharya_BHosseini_NBKovachki_AMStuart_2021a,RMaulik_BLusch_PBalaprakash_2021a,OSan_RMaulik_2018a,OSan_RMaulik_MAhmed_2019a} typically use autoencoders for the dimension reduction and model the reduced dynamics using a NN.
While these methods are very flexible in that they can model dynamics with arbitrary structure, the resulting ROMs are not interpretable.
Another method, latent space dynamics identification (LaSDI)
\cite{CBonneville_YChoi_DGhosh_JLBelof_2024a,CBonneville_etal_2024a,WDFries_XHe_YChoi_2022a,XHe_YChoi_WDFries_JLBelof_JSChen_2023a,JSRPark_SWCheung_YChoi_YShin_2024a}, can be viewed as a hybrid of OpInf and NN-based approaches that typically uses autoencoder-based dimension reduction and learns reduced-order dynamics by solving a least-squares regression problem for coefficient matrices corresponding to a library of nonlinear candidates functions.
This approach is related to the SINDy algorithm \cite{EKaiser_JNKutz_SLBrunton_2018a} but does not enforce a sparsity requirement.
The library of candidate functions for LaSDI is typically chosen to be polynomial, which results in solving a similar least-squares regression problem to OpInf when learning the latent dynamics.
Unlike OpInf, while the resulting ROM structure is interpretable, a natural structure for the ROM dynamics cannot be deduced \emph{a priori} since autoencoder-based dimension reduction does not preserve structure from the FOM.
In each of these approaches, error estimates for the resulting ROMs are limited, with the exception of the recent thermodynamics-based LaSDI approach \cite{JSRPark_SWCheung_YChoi_YShin_2024a}.

Our proposed kernel-based non-intrusive ROMs, which we call ``Kernel ROMs'',
share similarities with the aforementioned approaches while overcoming some noticeable drawbacks.
Like other approaches, we begin by applying POD or QM dimension reduction to a set of training snapshots and learn a function that approximates the system dynamics within a latent space.
However, instead of modeling the ROM dynamics as a polynomial and learning the polynomial coefficients through least-squares regression as in OpInf and LaSDI, we use regularized kernel interpolation
\cite{CAMicchelli_MPontil_2005a,GSantin_2018a,GSantin_BHaasdonk_2021a} to model the reduced dynamics with a function belonging to a user-defined reproducing kernel Hilbert space (RKHS).
The structure of the learned function depends on the positive-definite kernel that defines the RKHS.
For example, if the governing FOM has a polynomial structure, we can use a kernel induced by a feature map to compute ROM dynamics that share the same polynomial structure.
On the other hand, if the FOM dynamics have unknown or only partially known structure, a more generic nonlinear kernel can be used to model the unknown part of the ROM dynamics.
In this sense, our proposed approach has a natural way of incorporating closure terms into the ROM dynamics.
While kernel methods have been used to emulate reduced dynamics in previous work \cite{JPhillips_JAfonso_AOliveira_LMSilveira_2003a, DWirtz_BHaasdonk_2012a}, these approaches are intrusive in that they assume that the FOM dynamics can be sampled explicitly, and they do not demonstrate a way to inject explicit structure into the learned ROM.
The authors in \cite{PJBaddoo_BHermann_BJMcKeon_SLBrunton_2022a} propose a generalization of DMD using kernel methods to learn a high-dimensional dynamical system consisting of a linear term and a nonlinear term, implicitly injecting structure by modeling nonlinear terms using polynomial kernels.
The structure is implicit because a post-processing step is required to extract the linear term.
While this approach bears some similarities to ours, we focus on constructing low-dimensional ROMs by applying a dimension reduction technique, projecting the training snapshots onto a low-dimensional latent space, and learning a low-dimensional dynamical system in the latent space whose structure \emph{explicitly} mirrors the FOM structure. The key to preserving the FOM structure is the use of feature map kernels, which enables identifying the ROM coefficient matrices without requiring a post-processing step.
In summary, our proposed approach is entirely data-driven, can produce interpretable and flexible ROMs, and yields computable \emph{a posteriori} error bounds between the non-intrusive ROM and FOM solutions.

The outline of this paper is as follows.
We first review essential aspects of regularized kernel interpolation in \Cref{sec:kernel_interpolation}.
We then review intrusive projection-based model reduction in \Cref{sec:intrusive_rom}, with a focus on quadratic dimension reduction and the resulting model structure.
\Cref{sec:nonintrusive_rom} details the application of regularization kernel interpolation in the non-intrusive model reduction setting, and a corresponding \emph{a posteriori} error analysis is provided in \Cref{sec:error_estimates}.
We demonstrate our proposed approach numerically on several examples in \Cref{sec:numerics}, including comparisons to OpInf and intrusive ROMs when possible.
The results show that our proposed approach can accommodate either POD or QM dimension reduction and produces comparable results to OpInf while also yielding a computable error bound.
Finally, \Cref{sec:conclusion} provides a few concluding remarks and identifies potential avenues for future development.

\section{Regularized kernel interpolation}\label{sec:kernel_interpolation}

This section reviews the essentials of regularized kernel interpolation, the key ingredient for our non-intrusive model reduction approach.
\Cref{sec:scalar_valued_interpolation} reviews scalar-valued interpolation, which is extended to vector-valued interpolation in \Cref{sec:vector_valued_interpolation}.
Scenario-specific kernel design is then discussed in \Cref{sec:kernel_selection}.

\subsection{Scalar-valued kernel interpolation}\label{sec:scalar_valued_interpolation}
We begin with a review of regularized kernel interpolation for scalar-valued functions.
By the Moore--Aronszajn Theorem (see, e.g., \cite[Theorem 3.10]{GSantin_2018a}),
a positive-definite kernel function defines a unique Hilbert space with desirable properties.
This result leads to the Representer Theorem --- the key result used for computing an optimal interpolant in an RKHS --- as well as a pointwise error bound on the interpolant.

\begin{definition}[Positive-definite kernels]
\label{def:pd_kernels}
A function $K:\real^{n_x} \times \real^{n_x}\to \real$ is a (real-valued) \emph{kernel function} if it is symmetric, i.e., $K(\bx,\bx') = K(\bx',\bx)$ for all $\bx,\bx'\in\real^{n_x}$. A kernel function $K$ is said to be \emph{positive definite} if for any matrix $\bX=[~\bx_1~~\cdots~~\bx_m~]\in \real^{n_x\times m}$ with pairwise distinct columns, the kernel matrix $K(\bX, \bX)\in \real^{m\times m}$ with entries $K(\bX, \bX)_{ij}=K(\bx_i, \bx_j)$ is positive semi-definite. If $K(\bX, \bX)$ is strictly positive definite, then $K$ is said to be strictly positive definite.
\end{definition}

\begin{definition}[RKHS]
Let $K:\real^{n_x}\times\real^{n_x}\to\real$ be a positive-definite kernel function.
Consider the pre-Hilbert space of functions
\begin{align*}\cH_K^0(\real^{n_x})
    = \set{
        \fun:\real^{n_x}\to\real
        ~\bigg|~
        \exists\,m\in\mathbb{N},\,
        \bomega\in\real^{m},\,
        \set{\bx_j}_{j=1}^{m}\subset\real^{n_x}
        ~\textup{such that}~
        \fun(\bx) = \sum_{j=1}^m \omega_j K(\bx_j, \bx)
    }.
\end{align*}
The \emph{reproducing kernel Hilbert space (RKHS)} $\cH_K(\real^{n_x})$ induced by the kernel $K$ is the (unique) completion of $\cH_K^0(\real^{n_x})$ with respect to the norm $\norm{\cdot}_{\cH_K(\real^{n_x})} \coloneqq \innerprod{\cdot,\cdot}_{\cH_K(\real^{n_x})}^{1/2}$ induced by the inner product
\begin{align*}
\innerprod{\fun, \fun'}_{\cH_K(\real^d)}
    \coloneqq \sum_{j=1}^{m} \sum_{k=1}^{m'} \omega_j \omega_k' K(\bx_j, \bx_k'),
\end{align*}
in which $\fun(\bx) = \sum_{j=1}^m \omega_j K(\bx_j, \bx)$ and $\fun'(\bx) = \sum_{j=1}^{m'} \omega_k' K(\bx_k', \bx)$.
\end{definition}

For an ordered collection of pairwise-distinct vectors $\{\bx_j\}_{j=1}^m \subset \real^{n_x}$, we use $K(\bX,\bX)$ to denote $m \times m$ kernel matrix of \Cref{def:pd_kernels} and define the vector $K(\bX,\bx) = [~K(\bx_1,\bx)~~\cdots~~K(\bx_m,\bx)~]\trp\in\real^m$.
To simplify notation, we will write $\cH_K$ for $\cH_K(\real^{n_x})$ when it is understood that $K$ is defined over $\real^{n_x}\times \real^{n_x}$.
Importantly, for $\fun(\bx) = \sum_{j=1}^m \omega_j K(\bx_j, \bx)$, the induced RKHS norm $\norm{\fun}_{\cH_K}$ can be computed efficiently via the corresponding kernel matrix,
\begin{align}\label{eq:kernel_norm}
    \norm{\fun}_{\cH_K}^2
    = \sum_{j=1}^{m} \sum_{k=1}^{m} \omega_j \omega_k K(\bx_j, \bx_k)
    = \bomega\trp K(\bX, \bX)\bomega.
\end{align}
We now state a main result from RKHS theory that is fundamental for our method.

\begin{definition}[Regularized kernel interpolant]
    \label{def:kernel_interpolant}
    Let $\fun:\real^{n_x}\to \real$, $\set{\bx_j}_{j=1}^m\subset \real^{n_x}$ be pairwise distinct, and denote $y_j=\fun(\bx_j)$.
    For a given RKHS $\cH_K$ and a regularization parameter $\gamma \geq 0$, a \emph{regularized interpolant} $s_{\fun}^\gamma \in \cH_K$ of $\fun$ is a solution to the minimization problem
\begin{align}\label{eq:kernel_optimization_problem}
    \min_{s\in \cH_K} \; \sum_{j=1}^m(y_j - s(\bx_j))^2 + \gamma \norm{s}_{\cH_K}^2.
\end{align}
\end{definition}

\begin{theorem}[Representer Theorem]\label{thm:representer}
    If $\gamma >0$, then the minimization problem \cref{eq:kernel_optimization_problem} has a solution of the form
    \begin{subequations}
    \begin{align}\label{eq:kernel_interpolant}
        s_{\fun}^\gamma(\bx)
        = \sum_{j=1}^m \omega_j K(\bx_j, \bx)
        = \bomega\trp K(\bX, \bx),
    \end{align}
    where the coefficient vector $\bomega = [~\omega_1~~\cdots~~\omega_m~]\trp\in \real^m$ solves the $m\times m$ linear system
    \begin{align}\label{eq:coefficient_equation}
        \big(K(\bX, \bX) + \gamma \bI\big)
        \bomega
= \left[\begin{array}{c}
            y_1 \\ \vdots \\ y_m
        \end{array}\right].
    \end{align}
    \end{subequations}
Moreover, if $K$ is strictly positive definite, then $s_{\fun}^\gamma$ is the unique minimizer of \cref{eq:kernel_optimization_problem}.
\end{theorem}
See, e.g., \cite[Theorem 9.3]{GSantin_BHaasdonk_2021a} for a proof of \Cref{thm:representer}.
A key observation from \Cref{thm:representer} is that a solution to the infinite-dimensional minimization problem \cref{eq:kernel_optimization_problem} can be obtained by solving the finite-dimensional linear system \cref{eq:coefficient_equation}.

Without regularization ($\gamma=0$) and assuming that $K(\bX, \bX)$ is invertible, the function $s_{\fun}^0\in\cH_K$ exactly interpolates the data, i.e., $s_{\fun}(\bx_j)=y_j$ for each $j=1, \dots, m$. Moreover, $s_{\fun}^0$ satisfies the following error bound.

\begin{theorem}[Power function error bound]\label{thm:power_function_bound}
    If $\fun\in \cH_K$ and $s_{\fun}^0\in \cH_K$ is an (unregularized) interpolant of $\fun$ corresponding to the pairwise distinct data $\set{\bx_j}_{j=1}^m\subset \real^{n_x}$ and $y_j = \fun(\bx_j)\in \real$, then
    \begin{subequations}
    \begin{align}\label{eq:power_function_bound}
        |\fun(\bx)-s_{\fun}^0(\bx)|\leq P_{K, \bX}(\bx)\norm{\fun}_{\cH_K} \qquad \forall \; \bx \in \real^d,
    \end{align}
    where $P_{K, \bX}:\real^{n_x}\to \real$ is the so-called \emph{power function} defined by
    \begin{align}\label{eq:power_function}
        P_{K, \bX}(\bx)
        = \sqrt{K(\bx, \bx)-K(\bX, \bx)\trp K(\bX, \bX)^{-1}K(\bX, \bx)}.
    \end{align}
    \end{subequations}
\end{theorem}
See, e.g., \cite[Thm 4.9]{GSantin_2018a} for a proof of the bound \cref{eq:power_function_bound} and \cite[Prop. 4.11, Prop. 4.12]{GSantin_2018a} for the characterization~\cref{eq:power_function} of the power function.
While this error bound is for the unregularized, fully interpolatory case, it is still useful in practice for the regularized case when $\gamma > 0$ is small.

\subsection{Vector-valued kernel interpolation}\label{sec:vector_valued_interpolation}

Kernel interpolation can be readily extended to vector-valued functions.
The simple extension presented here, which is sufficient for our use case, is a special case of a more general extension relying on matrix-valued kernels
(see, e.g., \cite{CAMicchelli_MPontil_2005a,GSantin_BHaasdonk_2021a}).

Consider the vector-valued function $\bfun:\real^{n_x}\to\real^{n_y}$, $n_y>1$.
As before, let $\set{\bx_j}_{j=1}^m\subset \real^{n_x}$ be pairwise distinct and suppose $\by_j=\bfun(\bx_j)\in \real^{n_y}$ for $j=1,\ldots,m$.
Also let $\fun_i:\real^{n_x}\to \real$ denote the $i$-th component of $\bfun$,
$y_{j, i}$ be the $i$-th component of $\by_j$, and define the input and output data matrices
\begin{align}
    \bX &= [~\bx_1~~\cdots~~\bx_m~]\in\real^{n_x\times m},
    &
    \bY &= [~\by_1~~\cdots~~\by_m~]\in\real^{n_y\times m}.
\end{align}
We construct a vector-valued regularized kernel interpolant $\bs_{\bfun}^\gamma$ by fitting scalar-valued kernel interpolants to each component $\fun_i$ of $\bfun$.
Consequently, $\bs_{\bfun}^\gamma$ is an element of the $n_y$-fold Cartesian product $\cH_K^{n_y} \coloneqq \cH_K\times \dots \times\cH_K$, which has the inner product
$\innerprod{\bu, \bw}_{\cH_K^{n_y}} = \sum_{i=1}^{n_y}\innerprod{u_i, w_i}_{\cH_K}$ for all $\bu = (u_1,\ldots,u_{n_y})\in \cH_K^{n_y}$ and $\bw = (w_1,\ldots,w_{n_y}) \in \cH_K^{n_y}$.
The regularized kernel interpolant constructed in this manner solves the optimization problem
\begin{align}\label{eq:vv_kernel_optimization_problem}
    \min_{\bs\in \cH_K^{n_y}} \; \sum_{j=1}^m \norm{\by_j - \bs(\bx_j)}_2^2 + \gamma\norm{\bs}_{\cH_K^{n_y}}^2,
\end{align}
where $\norm{\cdot}_2$ denotes the Euclidean $2$-norm and $\norm{\cdot}_{\cH_K^{n_y}}^2 = \innerprod{\cdot, \cdot}_{\cH_K^{n_y}}$.
To see this, note that the objective function in \cref{eq:vv_kernel_optimization_problem} can be rewritten as
\begin{align*}
    \sum_{j=1}^m \norm{\by_j - \bs(\bx_j)}_2^2 + \gamma\norm{\bs}_{\cH_K^{n_y}}^2
    =
    \sum_{i=1}^{n_y} \left(\sum_{j=1}^m \left(y_{j, i} - s_i(\bx_j)\right)^2 + \gamma\norm{s_i}_{\cH_K}^2\right),
    \quad
    \bs(\bx) = \left[\begin{array}{c}s_1(\bx)\\ \vdots \\ s_{n_y}(\bx)\end{array}\right],
\end{align*}
and therefore \cref{eq:vv_kernel_optimization_problem} decouples into ${n_y}$ independent scalar-valued regularized interpolation problems:
\begin{align}
    \min_{\bs_i\in \cH_K} \; \sum_{j=1}^m(y_{j, i} - s_i(\bx_j))^2 + \gamma \norm{s_i}_{\cH_K}^2,
    \qquad
    i=1, \dots, n_y.
\end{align}
\Cref{thm:representer} can then be applied to each subproblem to yield scalar-valued interpolants $s_{\fun_i}^\gamma$ of the form
\begin{subequations}
\begin{align}
    s_{\fun_i}^\gamma(\bx)
    = \sum_{j=1}^m\omega_{i, j}K(\bx_j, \bx)
    = \bomega_i\trp K(\bX, \bx),
\end{align}
where each coefficient vector $\bomega_1,\ldots,\bomega_{n_y}\in \real^{m}$ solves an $m \times m$ linear system,
\begin{align}
    \big(K(\bX, \bX) + \gamma \bI\big)
    \bomega_i
= \left[\begin{array}{c}
        y_{i, 1} \\ \vdots \\ y_{i, m}
    \end{array}\right],
    \quad
    i = 1,\ldots,p.
\end{align}
\end{subequations}
As before, $\gamma \ge 0$ is a given regularization parameter.
An interpolant of $\bfun$ can then be defined by
\begin{align*}
    \bs_{\bfun}^\gamma(\bx)
    = [~s_{\fun_1}^\gamma(\bx)~~\cdots~~s_{\fun_{n_y}}^\gamma(\bx)~]\trp.
\end{align*}
We summarize with the following corollary of \Cref{thm:representer} and a straightforward extension of \Cref{thm:power_function_bound}.

\begin{corollary}[Vector Representer Theorem]\label{cor:vv_representer}
    If $\gamma>0$, then the minimization problem \cref{eq:vv_kernel_optimization_problem} has a solution of the form
    \begin{subequations}\label{eq:vv_kernel_interpolant-both}
    \begin{align}\label{eq:vv_kernel_interpolant}
        \bs_{\bfun}^\gamma(\bx)
        = \bOmega\trp K(\bX, \bx),
    \end{align}
    where the coefficient matrix $\bOmega\in \real^{m\times {n_y}}$ solves the linear system
    \begin{align}\label{eq:vv_coefficient_equation}
        \big(K(\bX, \bX) + \gamma \bI\big)
        \bOmega
        =
        \bY\trp.
    \end{align}
    \end{subequations}
    Moreover, if $K$ is strictly positive definite, $s_{\bfun}^\gamma$ is the unique minimizer.
\end{corollary}

\begin{corollary}\label{cor:vector_power_function_bound}
Let $\bfun\in \cH_K^{n_y}$ and $\bM\in \real^{{n_y}\times {n_y}}$ be a symmetric positive definite weighting matrix with Cholesky factorization $\bM = \bL\bL\trp$.
If $\bs_{\bfun}^0\in \cH_K^{n_y}$ is an (unregularized) vector-valued interpolant of $\bfun$ of the form \cref{eq:vv_kernel_interpolant-both} corresponding to the pairwise distinct data $\set{\bx_i}_{i=1}^m\subset \real^{n_x}$ and $\by_i = \bfun(\bx_i)\in \real^{n_y}$, then
\begin{align}\label{eq:vector_power_function_bound}
    \norm{\bfun(\bx)-\bs_{\bfun}^0(\bx)}_\bM \leq P_{K, \bX}(\bx)\norm{\bL}_2\norm{\bfun}_{\cH_K^{n_y}} \qquad \forall \; \bx \in \real^{n_x}.
\end{align}
\begin{proof}
Since $\bs_{\bfun}^0$ interpolates $\bfun$ component-wise using the same kernel $K$ and interpolation points $\set{\bx_i}_{i=1}^{m}$, applying \Cref{thm:power_function_bound} yields
\begin{align*}
    \norm{\bfun(\bx)-\bs_{\bfun}^0(\bx)}_\bM^2
    = \norm{\bL\trp(\bfun(\bx)-\bs_{\bfun}^0(\bx))}_2^2
    &\leq \norm{\bL}_2^2\norm{\bfun(\bx)-\bs_{\bfun}^0(\bx)}_2^2
    = \norm{\bL}_2^2\sum_{i=1}^{n_y} |\fun_i(\bx) - s_{\fun_i}^0(\bx)|^2
    \\
    &\leq \norm{\bL}_2^2\sum_{i=1}^{n_y}  P_{K, \bX}(\bx)^2\norm{\fun_i}_{\cH_K}^2
= P_{K, \bX}(\bx)^2\norm{\bL}_2^2\norm{\bfun}_{\cH_K^p}^2.
\end{align*}
\end{proof}
\end{corollary}

In \Cref{sec:nonintrusive_rom}, we use \Cref{cor:vv_representer} to develop a strategy for constructing reduced-order models (ROMs) from data; \Cref{cor:vector_power_function_bound} is used in \Cref{sec:error_estimates} to derive \emph{a posteriori} error estimates for these ROMs.

\subsection{Kernel selection}\label{sec:kernel_selection}

Since a positive-definite kernel $K$ uniquely defines the RKHS $\cH_K$, the choice of kernel determines what form an interpolant can take as well as the approximation power of the optimal interpolant.
We argue for the use of different types of kernels depending on how much information is available about the function $\bfun:\real^{n_x}\to \real^{n_y}$ being interpolated.

\subsubsection{Unknown structure: radial basis function kernels}

If the structure of $\bfun$ is unknown, one effective choice is to generate the kernel using a radial basis function~(RBF).
These general-purpose kernels have the form
\begin{subequations}
\label{eq:rbf-psi-def}
\begin{align}\label{eq:rbf_kernel}
    K(\bx, \bx')=\psi(\epsilon \norm{\bx-\bx'}_2),
\end{align}
where $\psi:\real_{\geq 0}\to\real$ and $\epsilon>0$.
Hence, RBF kernel interpolants are given by
\begin{align}\label{eq:rbf_interpolant}
    \bs_{\bfun}^\gamma(\bx) = \bOmega\trp\bpsi_{\!\epsilon}(\bx),
    \qquad
    \bpsi_{\!\epsilon}(\bx)
    = \begin{bmatrix}
        \psi(\epsilon\norm{\bx_1-\bx}_2) \\
        \vdots \\
        \psi(\epsilon\norm{\bx_m-\bx}_2)
    \end{bmatrix}
    \in\real^{m}.
\end{align}
\end{subequations}
The so-called shape parameter $\epsilon$ is a hyperparameter that should be tuned to achieve optimal performance.
\Cref{tbl:rbf_kernels} provides examples of commonly used RBF generator functions $\psi$.
Note that the cost of evaluating an RBF kernel interpolant is $\mathcal{O}(m(n_x + n_y))$.
A thorough discussion of the use of RBFs in kernel interpolation can be found in, e.g., \cite{GBWright_2003a}.

\begin{table}[t]
    \centering
    \begin{tabular}{r|l}
        Name & $\psi(x)$ \\ \hline
        Gaussian & $\exp(-x^2)$ \\ Basic Mat\'{e}rn & $\exp(-x)$ \\ Inverse Quadratic & $(1 + x^2)^{-1}$ \\ Inverse Multiquadric & $(1 + x^2)^{-1/2}$ \\ Thin Plate Spline & $x^2 \log(x)$ \\
    \end{tabular}
    \caption{Examples of common RBF kernel-generating functions.}
    \label{tbl:rbf_kernels}
\end{table}

\subsubsection{Known structure: feature map kernels}

If the structure of $\bfun$ is known, kernels induced by \emph{feature maps} can often be used to endow the interpolant with matching structure, which can result in more accurate and interpretable approximations than when using general-purpose kernels.
A feature map kernel can be written as
\begin{subequations}
\begin{align}\label{eq:feature_map_kernel}
    K(\bx, \bx') = \bphi(\bx)\trp\bG\bphi(\bx'),
\end{align}
where $\bphi:\real^{n_x}\to \real^{n_\phi}$ is called the feature map and $\bG\in \real^{n_\phi\times n_\phi}$ is a symmetric positive definite weighting matrix.
It can be easily verified that feature map kernels are positive definite kernels (see, e.g., \cite{GSantin_2018a}).
A feature map kernel results in a kernel interpolant of the form
\begin{align}\label{eq:feature_map_interpolant}
    \bs_{\bfun}^\gamma(\bx) &= \bOmega\trp K(\bX, \bx)
    = \underbrace{\bOmega\trp \bphi(\bX)\trp\bG}_{\bC} \bphi(\bx)
    = \bC\bphi(\bx),
\end{align}
\end{subequations}
where $\bphi(\bX) \coloneqq [~\bphi(\bx_1)~~\cdots~~\bphi(\bx_m)~]\in\real^{n_\phi\times m}$.
Importantly, the matrix $\bC\in\real^{n_y\times n_\phi}$ can be computed once and reused repeatedly for online kernel evaluations.
After constructing $\bC$, the cost of evaluating a feature map kernel interpolant is therefore $\mathcal{O}(n_\phi n_y)$, plus the expense of evaluating $\bphi$ once.

The advantage of feature map kernels is that one can imbue $\bs_{\bfun}^\gamma$ with specific structure by designing the feature map $\bphi$ accordingly.
For example, if
\begin{align}\label{eq:quadratic_fm}
\bphi(\bx) = \begin{bmatrix}
    \bx \\ \bx\otimes \bx
\end{bmatrix} \in \real^{n_x+n_x^2},
\end{align}
where $\otimes$ denotes the Kronecker product \cite{vanLoan2000kronecker},
then the associated kernel interpolant can be written as
\begin{align}
    \bs_{\bfun}^\gamma(\bx)
    = \bC_1 \bx + \bC_2 [\bx \otimes \bx],
    \qquad
    \bC = [~\bC_1~~\bC_2~] \in \real^{n_y\times (n_x+n_x^2)}.
\end{align}
Therefore, if it is known that $\bfun$ has linear-quadratic structure, then using a kernel induced by the feature map \cref{eq:quadratic_fm} results in a kernel interpolant that has the same linear-quadratic structure.

\subsubsection{Hybrid approach}

For the purposes of model reduction, it is critical to keep the cost of evaluating the kernel interpolant low.
The cost of evaluating an RBF kernel interpolant \cref{eq:rbf_interpolant} scales with the number of training samples $m$; by contrast, the cost of evaluating a feature map kernel interpolant \cref{eq:feature_map_interpolant} is independent of $m$, but depends on the feature dimension $n_\phi$.
If a feature map that fully specifies the desired structure requires a large $n_\phi$, one alternative is to define a new kernel that sums a less aggressive feature map kernel with an RBF kernel:
\begin{subequations}
\begin{align}\label{eq:hybrid_kernel}
    K(\bx, \bx')
    = c_\phi \bphi(\bx)\trp\bG\bphi(\bx') + c_\psi \psi(\epsilon \norm{\bx-\bx'}_2),
\end{align}
where $c_\phi, c_\psi \in \real_{>0}$ are positive weighting coefficients and $\bphi$ is chosen to keep $n_\phi$ from being too large.
The resulting kernel interpolant then has the form
\begin{align}\label{eq:hybrid_kernel_interpolant}
    \bs_{\bfun}^\gamma(\bx) &= \bOmega\trp K(\bX, \bx)
    = \bOmega\trp\big(c_\phi\bphi(\bX)\trp\bG\bphi(\bx) + c_\psi\bpsi_{\!\epsilon}(\bx)\big)
    = \bC\bphi(\bx)+ c_\psi \bOmega\trp\bpsi_{\!\epsilon}(\bx),
\end{align}
\end{subequations}
where $\bC = c_\phi \bOmega\trp\bphi(\bX)\trp\bG$ now incorporates the weighting coefficient $c_\phi$.
The idea is to use the feature map to incorporate dominant structure while relying on the RBF to approximate additional, potentially expensive terms.
Note that this framework also applies to scenarios where the structure of $\bv$ is only partially known.

As an example, consider the case where $\bv$ is a quartic polynomial, i.e.,
\begin{align}
    \label{eq:quartic-target-function}
    \bv(\bx)
    = \bA_1\bx
    + \bA_2[\bx\otimes \bx]
    + \bA_3[\bx\otimes \bx \otimes \bx]
    + \bA_4[\bx\otimes \bx \otimes \bx\otimes \bx],
\end{align}
where $\bA_1\in\real^{n_x\times n_x}$, $\bA_2\in\real^{n_x\times n_x^2}$, $\bA_3\in\real^{n_x\times n_x^3}$, and $\bA_4\in\real^{n_x\times n_x^4}$.
One option is to fully capture the structure using a quartic feature map,
\begin{align}
    \label{eq:quartic_feature_map}
    \bphi(\bx)
    = \begin{bmatrix}
        \bx \\
        \bx \otimes \bx \\
        \bx \otimes \bx \otimes \bx \\
        \bx \otimes \bx \otimes \bx \otimes \bx
    \end{bmatrix}
    \in\real^{n_x + n_x^2 + n_x^3 + n_x^4}.
\end{align}
However, evaluating the associated kernel interpolant costs $\mathcal{O}(n_x^4 n_y)$ operations, which is quite large for moderate $n_x$.
Using the linear-quadratic feature map \cref{eq:quadratic_fm} decreases $n_\phi$ from $n_x^4$ to $n_x^2$, and supplementing with an RBF kernel results in a kernel interpolant of the form
\begin{align}\label{eq:hybrid_interpolant_example}
    \bs_{\bfun}^\gamma(\bx)
    = \bC_1\bx + \bC_2[\bx\otimes \bx] + c_\psi\bOmega\trp\bpsi_{\!\epsilon}(\bx).
\end{align}
This interpolant does not fully represent the quartic structure of \cref{eq:quartic-target-function}, but it can be evaluated with only $\mathcal{O}((n_x^2 + m)n_y)$ operations.
In this case, the RBF term acts as a type of closure term for structure that is not accounted for by the feature map.

\begin{remark}[Input normalization]\label{remark:kernel_normalization}
In some cases, in particular when using high-order polynomial feature maps, the kernel matrix $K(\bX, \bX)$ used for determining $\bOmega$ may be poorly conditioned.
Increasing the regularization constant $\gamma$ can improve the conditioning of the system \cref{eq:vv_coefficient_equation}, but this can also degrade the accuracy of the resulting kernel interpolant. Applying a normalization to the inputs can help remedy the situation:
for any injective $\bnu:\real^{n_x}\to\real^{n_x}$, if $K$ is positive definite, then the function $K_\bnu:\real^{n_x}\times\real^{n_x}\to \real$ defined by
\begin{align}\label{eq:normalized_kernel}
    K_\bnu(\bx, \bx') = K(\bnu(\bx), \bnu(\bx'))
\end{align}
is also a positive-definite kernel function \cite{GSantin_2018a},
and choosing $\bnu$ judiciously can improve the conditioning of $K_\bnu(\bX,\bX)$ compared to $K(\bX,\bX)$.
A common choice is $\bnu(\bx) = \bSigma^{-1}(\bx - \bar{\bx})$, where $\bSigma = \operatorname{diag}(\bsigma)\in\real^{n_x \times n_x}$ and $\bar{\bx}\in\real^{n_x}$ with components
\begin{align}\label{eq:kernel_normalization_choice}
    \sigma_i = \max_{j}(\bX_{ij}) - \min_{j}(\bX_{ij}),
    \qquad
    \bar{x}_i = \min_{j}(\bX_{ij}),
    \qquad
    i = 1, \dots, d,
\end{align}
which maps the entries of each row of inputs to the interval $[0, 1]$.
In this case, an effective choice for the weighting matrix $\bG$ in feature map kernels is $\bG = (1/n_\phi)\bI$, where $n_\phi$ is the feature map dimension.
\end{remark}

\section{Intrusive projection-based model reduction}\label{sec:intrusive_rom}

We now return to the model reduction setting and give a brief overview of intrusive projection-based ROMs, which inherit certain structure from the systems they emulate.
\Cref{sec:nonintrusive_rom} presents a non-intrusive alternative to intrusive model reduction for which kernel interpolation is the key ingredient and which can be designed to mimic the structure inheritance enjoyed by projection-based ROMs.

\subsection{Generic projection-based reduced-order models}

We consider high-dimensional systems of ordinary differential equations (ODEs) of the form
\begin{align}\label{eq:generic_fom}
    \dt\bq(t) = \bff(\bq(t); \bu(t)),
    \qquad
    \bq(0) = \bq_0(\bmu),
\end{align}
where $\bq: [0, T]\to \real^{n_q}$ is the state,
$\bu:[0, T]\to \real^{n_u}$ is an input function,
$\bff:\real^{n_q} \times \real^{n_u} \to \real^{n_q}$ governs the state evolution,
$\bq_0(\bmu)\in\real^{n_q}$ is the initial condition parameterized by $\bmu\in \cD \subset \real^{n_\mu}$, and $T > 0$ is the final desired simulation time.
Models of this form often arise from semi-discretizations of time-dependent partial differential equations (PDEs), in which case the large state dimension $n_q$ corresponds to the fidelity of the underlying mesh.
We call \cref{eq:generic_fom} the full-order model (FOM).
For ease of exposition, we focus on autonomous systems without an input term $\bu$; see \Cref{rmk:input_terms} for discussion on incorporating input terms.

A ROM for \cref{eq:generic_fom} is a low-dimensional system of ODEs whose solution can be used to approximate the FOM state $\bq(t)$.
To that end, we consider a low-dimensional state approximation,
\begin{align}\label{eq:generic_decoder}
    \bq(t) \approx \bg(\tbq(t)),
\end{align}
where $\bg:\real^{r}\to\real^{n_q}$ and $\tbq:[0,T]\to\real^{r}$ is the reduced-order state, with $r \ll n_q$.
The function $\bg$ represents a decompression operation, mapping from reduced coordinates to the original high-dimensional space.
We assume the existence of a corresponding compression map $\bh:\real^{n_q}\to\real^{r}$, mapping high-dimensional states to reduced coordinates, such that $\bh \circ \bg :\real^r\to \real^r$ is the identity.
Importantly, $(\bg\circ\bh)^2 = \bg \circ (\bh \circ \bg) \circ \bh = \bg \circ \bh$, i.e., $\bg \circ \bh :\mathbb{R}^{n_q}\to\mathbb{R}^{n_q}$ is a projection.
The evolution for the reduced state $\tbq(t)$ is then given by
\begin{align}
    \dt\tbq(t)
    = \dt\bh(\bg(\tbq(t)))
    = \bh'(\bg(\tbq(t)))\dt\bg(\tbq(t))
    \approx \bh'(\bg(\tbq(t)))\bff(\bg(\tbq(t))),
\end{align}
in which $\bh':\real^{n_q}\to\real^{n_q\times n_q}$ is the Jacobian of $\bh$ and where the final step comes from inserting the approximation \cref{eq:generic_decoder} into the FOM \cref{eq:generic_fom}.
The resulting system
\begin{align}
    \label{eq:generic_rom-nostructure}
    \dt\tbq(t)
    = \bh'(\bg(\tbq(t)))\bff(\bg(\tbq(t))),
    \qquad
    \tbq(0) = \bh(\bq_0(\bmu))
\end{align}
is the projection-based ROM for \cref{eq:generic_fom} corresponding to $\bg$ and $\bh$.

As written, \cref{eq:generic_rom-nostructure} is not highly practical because it involves mapping up to the high-dimensional state space, performing computations in that space, then compressing the results.
However, for many common choices of $\bff$, $\bg$, and $\bh$, \cref{eq:generic_rom-nostructure} simplifies in such a way that all computations can be performed in the reduced space, as we will demonstrate shortly.

\subsection{Linear and quadratic dimension reduction}

Classical model reduction methods typically define $\bg$ and $\bh$ as affine functions.
In this work, we consider a slightly generalized approximation introduced in \cite{jain2017quadraticmanifolds} and leveraged in \cite{JBarnett_CFarhat_2022a,RGeelen_LBalzano_SWright_KWillcox_2024a,RGeelen_SWright_KWillcox_2023a,PSchwerdtner_BPeherstorfer_2024a}: let
\begin{align}
    \label{eq:quadratic_decoder}
    \bg(\tbq) = \obq + \bV\tbq + \bW[\tbq \otimes \tbq],
\end{align}
where $\obq \in \real^{n_q}$ is a fixed reference vector,
$\bV\in\real^{n_q\times r}$ has orthonormal columns, and $\bW\in\real^{n_q\times r^2}$ satisfies $\bV\trp\bW = \bf0$.
This approximation defines an $r$-dimensional quadratic manifold embedded in $\real^{n_q}$.
An appropriate compression map corresponding to \cref{eq:quadratic_decoder} is given by
\begin{align}
    \bh(\bq) &= \bV\trp(\bq - \obq),
\end{align}
which has Jacobian $\bh'(\bq) = \bV\trp$ and satisfies
\begin{align}
    \begin{aligned}
    (\bh \circ \bg)(\tbq)
&= \bV\trp\big((\obq + \bV\tbq + \bW[\tbq \otimes \tbq]) - \obq\big)
    = \bV\trp\obq + \tbq - \bV\trp\obq
    = \tbq,
    \end{aligned}
\end{align}
since $\bV\trp\bV$ is the identity and $\bV\trp$ annihilates $\bW$.
With $\bg$ and $\bh$ thus defined, the ROM \cref{eq:generic_rom-nostructure} becomes
\begin{align}
    \label{eq:generic_rom-withgandh}
    \dt\tbq(t)
    = \tbf(\tbq(t))
    \coloneqq \bV\trp\bff\big(\obq + \bV\tbq(t) + \bW [\tbq(t) \otimes \tbq(t)]\big),
    \qquad
    \tbq(0) = \bV\trp(\bq_0(\bmu) - \obq),
\end{align}
a system of $r \ll n_q$ ODEs defined by the function $\tbf:\real^{r}\to\real^{r}$.

The choices of $\obq$, $\bV$, and $\bW$ dictate the quality of the approximation \cref{eq:quadratic_decoder} and of the resulting ROM~\cref{eq:generic_rom-withgandh}.
To make an informed selection, we assume access to a limited set of training data: given a set of training parameters $\bmu_1,\ldots,\bmu_M \subset \cD$ and observation times $t_0,t_1,\ldots,t_{n_t}$, let
\begin{align}
    \label{eq:snapshots}
    \bq_k^{(\ell)}
    \coloneqq \bq(t_k; \bmu_\ell)
    \in \real^{n_q},
    \qquad
    \begin{aligned}
    \ell &= 1,\ldots,M,
    \\
    k &= 0, 1, \ldots, n_t,
    \end{aligned}
\end{align}
which are snapshots of the full-order state solution to the FOM \cref{eq:generic_fom}.
The reference vector $\obq$ is usually set to zero, the initial condition at a fixed training parameter value, or the average snapshot, i.e.,
\begin{align}
    \obq = \frac{1}{M n_t}\sum_{\ell=1}^{M}\sum_{k=0}^{n_t}\bq_{k}^{(\ell)}.
\end{align}
The model reduction framework developed in \Cref{sec:nonintrusive_rom} applies for any $\obq$, $\bV$ and $\bW$ such that $\bV\trp\bV = \bI$ and $\bV\trp\bW = \bf0$, but we focus on two best-practice cases.

First, if $\bW = \bf0$, the manifold defined by $\bg$ has no curvature and reduces to an affine subspace (or a linear subspace if $\obq = \bf0$) of $\real^{n_q}$.
In this case, we select $\bV$ using proper orthogonal decomposition (POD)~\cite{berkooz1993pod,graham1999optimal,sirovich1987turbulence}.
Define
\begin{align}
    \bQ
    \coloneqq \begin{bmatrix}
        (\bq_0^{(1)} - \obq)
        & \cdots &
        (\bq_{n_t}^{(1)} - \obq)
        &
        (\bq_0^{(2)} - \obq)
        & \cdots &
        (\bq_{n_t}^{(M)} - \obq)
    \end{bmatrix}
    \in \real^{n_q \times M(n_t + 1)},
\end{align}
the matrix of snapshots stacked column-wise and shifted by the reference snapshot.
The rank-$r$ POD basis matrix $\bV$ is given by the first $r$ left singular vectors of $\bQ$.
With this choice, $\bg\circ\bh$ is the optimal $r$-dimensional approximator for the (shifted) training snapshots in an $L^2$ sense.

Second, to construct a nonzero $\bW$, we use the greedy-optimal quadratic manifold (QM) approach of \cite{PSchwerdtner_BPeherstorfer_2024a}.
This method iteratively selects the columns of $\bV$ from the left singular vectors of $\bQ$ and solves a least-squares problem to determine $\bW$,
\begin{align}
    \label{eq:qm_optimization_problem}
    \min_{\bv_i,\bW}\;\norm{(\bI-\bV\bV\trp)\bQ - \bW[\bV\trp\bQ \odot \bV\trp\bQ]}_F^2 + \rho \norm{\bW}_F^2,
\end{align}
where $\bv_i$ is the final column of $\bV$ and all other columns are fixed from previous iterations.
Here, $\odot$ indicates the Khatri--Rao (column-wise Kronecker) product, and $\rho \geq 0$ is a scalar regularization parameter.
Traditional POD always sets $\bv_i$ to the $i$-th left singular vector, but here each $\bv_i$ can be chosen from among any of the left singular vectors that have not yet been selected, which can lead to substantial accuracy gains.

\begin{remark}[Kronecker redundancy]\label{remark:compressed-kronecker}
The product $\tbq\otimes \tbq$ contains redundant terms, i.e., $\tilde{q}_{i}\tilde{q}_{j}$ appears twice for each $i\neq j$, which means two columns of $\bW\in\real^{r\times r^{2}}$ act on the same quadratic state interaction in the product $\bW[\tbq\otimes\tbq]$.
As a consequence, the learning problem \cref{eq:qm_optimization_problem} has infinitely many solutions.
In practice, this issue is avoided by replacing $\otimes$ in \cref{eq:quadratic_decoder} with a compressed Kronecker product $\tilde{\otimes}$, defined by
\begin{align}
    \tbq\,\tilde{\otimes}\,\tbq
    \coloneqq \left[
        ~\tilde{q}_1^2~
        ~\tilde{q}_1 \tilde{q}_2~
        ~\tilde{q}_2^2~
        ~\dots~
        ~\tilde{q}_{r-1} \tilde{q}_r~
        ~\tilde{q}_r^2~
    \right]\trp \in \real^{r(r+1)/2},
\end{align}
which leads to a matrix $\tbW\in\real^{r\times r(r+1)/2}$ such that $\bW[\tbq\otimes\tbq] = \tbW[\tbq\,\tilde{\otimes}\,\tbq]$ for all $\tbq\in\real^{r}$.
Then, if $\odot$ applies $\tilde{\otimes}$ column-wise, the optimization \cref{eq:qm_optimization_problem} has a unique solution.
Similar adjustments can be made for higher-order Kronecker products.
\end{remark}

\subsection{Intrusive reduced-order models for quadratic systems}

The key observation in projection-based model reduction is that projection preserves certain structure.
Suppose that the function $\bff$ defining the dynamics of the FOM \cref{eq:generic_fom} has linear-quadratic structure, i.e.,
\begin{align}
    \label{eq:linearquadratic_fom}
    \dt\bq(t)
    = \bff(\bq(t))
    \coloneqq \bA\bq(t) + \bH[\bq(t) \otimes \bq(t)],
\end{align}
where $\bA\in \real^{n_q\times n_q}$, $\bH\in\real^{n_q\times n_q^2}$.
It is assumed that $\bH$ is symmetric in the sense that $\bH[\bq\otimes\bp] = \bH[\bp\otimes \bq]$ for all $\bq, \bp\in \real^{n_q}$.
Models with quadratic structure arise from quadratic PDEs, but can also result from applying lifting transformations to models with other structure \cite{kramer2019lifting,EQian_BKramer_BPeherstorfer_KWillcox_2020a}.
With a linear state approximation ($\obq = \bf0$ and $\bW = \bf0$), the ROM \cref{eq:generic_rom-withgandh} can be written as
\begin{align}
    \label{eq:linear_qm_rom}
    \dt\tbq(t)
    = \tbA\tbq(t) + \tbH[\tbq(t) \otimes \tbq(t)],
    \qquad
    \tbq(0) = \bV\trp\bq_0(\bmu),
\end{align}
in which $\tbA = \bV\trp\bA\bV \in \real^{r \times r}$ and $\tbH = \bV\trp\bH[\bV\otimes\bV]\in\real^{r \times r^2}$.
Constructing \cref{eq:linear_qm_rom} is an intrusive process because $\tbA$ and $\tbH$ depend explicitly on $\bA$ and $\bH$; however, we need not have access to $\bA$ and $\bH$ to observe that the quadratic structure is preserved.

In the QM case ($\bW\neq\bf0$, but still with $\obq = \bf0$ for convenience), the ROM \cref{eq:generic_rom-withgandh} has quartic dynamics,
\begin{align}
    \label{eq:quadratic_qm_rom}
    \begin{aligned}
    \dt \tbq(t)
    &= \tbA\tbq(t)
    + \tbH_{2}[\tbq(t)\otimes\tbq(t)]
    + \tbH_{3}[\tbq(t)\otimes\tbq(t)\otimes\tbq(t)]
    + \tbH_{4}[\tbq(t)\otimes\tbq(t)\otimes\tbq(t)\otimes\tbq(t)],
    \\
    \tbq(0) &= \bV\trp\bq_0(\bmu),
    \end{aligned}
\end{align}
where
$\tbA = \bV\trp \bA \bV$,
$\tbH_{2} = \bV\trp\bA\bW + \bV\trp\bH(\bV\otimes \bV)$,
$\tbH_{3} = \bV\trp\bH(\bV\otimes\bW + \bW\otimes\bV)$, and
$\tbH_{4} = \bV\trp\bH(\bW\otimes\bW)$.
Again, this process is intrusive, but the key result is that if one knows the structure of the FOM dynamics, one can also deduce the structure of the projection-based ROM.
See \apdxref{appendix:fullquadstructure} for the case when $\obq\neq\bf0$, in which a constant term appears in the reduced dynamics.

\section{Non-intrusive model reduction via kernel interpolation}\label{sec:nonintrusive_rom}

This section leverages regularized kernel interpolation to construct ROMs akin to \cref{eq:generic_rom-withgandh}, denoted
\begin{align}
    \label{eq:nonintrusive_rom}
    \dt\hbq(t) = \hbf(\hbq(t)),
    \qquad
    \hbq(0) = \bV\trp(\bq_0(\bmu) - \obq),
\end{align}
where $\hbq:[0,T]\to\real^{r}$ and $\hbf:\real^{r}\to\real^{r}$.
The structure of $\hbf$ can be informed by intrusive projection, but unlike projection, defining $\hbf$ through kernel interpolation does not require access to FOM operators such as $\bA$ or $\bH$ in \cref{eq:linearquadratic_fom}.
We use the notation $\hat{\cdot}$ to mark non-intrusive objects and differentiate from intrusive objects, which are marked with $\tilde{\cdot}$.

\subsection{Kernel reduced-order models}

We pose the problem of learning an appropriate $\hbf$ for the ROM \cref{eq:nonintrusive_rom} as a regression, which requires data for the state $\hbq(t)$ and its time derivative.
For the former, we reduce the FOM snapshots \cref{eq:snapshots} using the compression map $\bh$, that is,
\begin{align}
    \label{eq:reduced_state_data}
    \hbq_{k}^{(\ell)}
    \coloneqq \bh(\bq_{k}^{(\ell)})
    = \bV\trp (\bq_{k}^{(\ell)}-\obq)
    \in \real^{r}.
\end{align}
If the time step between observations is sufficiently small, an accurate approximation for the time derivatives of the state can be computed from finite differences of the reduced states, for example,
\begin{align}
    \dot{\hbq}_{k}^{(\ell)}
    \coloneqq
    \frac{\hbq_{k}^{\ell} - \hbq_{k-1}^{\ell}}{t_{k} - t_{k-1}}
    \approx \dt\hbq(t)\big|_{t=t_k}.
\end{align}
The ROM function $\hbf$ can then be defined as the solution to a minimization problem,
\begin{align}
    \label{eq:generic_nonintrusive_regression}
    \hbf
    = \underset{\bs \in S}{\arg \min} \; \sum_{\ell=1}^M\sum_{k=0}^{n_t}
    \norm{\dot{\hbq}_{k}^{(\ell)} - \bs(\hbq_{k}^{(\ell)})}_2^2
    + R(\bs),
\end{align}
where $S$ is some set of functions and $R:S\to \real_{\ge 0}$ is a regularization function.

The generic minimization \cref{eq:generic_nonintrusive_regression} encompasses several data-driven approaches which each use different choices for the space $S$ and the regularizer $R$.
By defining a kernel $K$ and an associated RKHS $S = \mathcal{H}_{K}^{r}$, and setting $R(\bs) = \gamma\|\bs\|_{\mathcal{H}_{K}^{r}}$, we obtain a vector regularized kernel interpolation problem,
\begin{align}
    \label{eq:kernel_nonintrusive_regression}
    \hbf
    = \underset{\bs \in \cH_K^r}{\arg\min} \; \sum_{\ell=1}^M\sum_{k=0}^{n_t}
    \norm{\dot{\hbq}_{k}^{(\ell)} - \bs(\hbq_{k}^{(\ell)})}_2^2
    + \gamma \norm{\bs}_{\cH_K^r}^2,
\end{align}
which is \cref{eq:vv_kernel_optimization_problem} with $\bx_j = \hbq_{k}^{(\ell)}$, $\by_j = \dot{\hbq}_{k}^{(\ell)}$ and $m = M (n_t + 1)$ after some minor reindexing for $k$ and $\ell$.
\Cref{cor:vv_representer} gives an explicit representation for $\hbf$, resulting in the ROM
\begin{subequations}\label{eq:kernel-rom}
\begin{align}
    \label{eq:kernel-rom-dynamics}
    \dt\hbq(t)
    = \hbf(\hbq(t))
    \coloneqq \bOmega\trp K(\hbQ, \hbq(t)),
    \qquad
    \hbq(0) = \bV\trp(\bq_0(\bmu) - \obq),
\end{align}
where $\bOmega\in\real^{M(n_t+1) \times r}$ solves the linear system
\begin{align}
    \label{eq:kernelrom-solution-system}
    \big(K(\hbQ,\hbQ) + \gamma\bI\big)\bOmega
    = \hbZ\trp,
\end{align}
with interpolation input and output matrices
\begin{align}
    \label{eq:training-data-matrices}
    \begin{aligned}
    \hbQ &= \begin{bmatrix}
        \hbq_{0}^{(1)} & \cdots & \hbq_{n_t}^{(1)} &
        \hbq_{0}^{(2)} & \cdots & \hbq_{n_t}^{(M)}
    \end{bmatrix} \in \real^{r\times M(n_t+1)},
    \\
    \hbZ &= \begin{bmatrix}
        \dot{\hbq}_{0}^{(1)} & \cdots & \dot{\hbq}_{n_t}^{(1)} &
        \dot{\hbq}_{0}^{(2)} & \cdots & \dot{\hbq}_{n_t}^{(M)}
    \end{bmatrix} \in \real^{r\times M(n_t+1)}.
    \end{aligned}
\end{align}
\end{subequations}
Note that the cost of evaluating $\hbf$ is $\mathcal{O}(rMn_t)$, plus the cost of evaluating the kernel term $K(\hbQ,\hbq(t))$.

\begin{remark}
    If the time derivatives of the FOM snapshots $\dt \bq_k^{(\ell)}$ are available, the time derivatives of the reduced state can instead be computed as
    \begin{align}
        \dot{\hbq}_{k}^{(\ell)} = \dt \hbq_k^{(\ell)}
        &\approx \dt (\bh(\bq_k^{(\ell)}))
        = \bh'(\bq_k^{(\ell)})\dt \bq_k^{(\ell)}
        = \bV\trp \dt \bq_k^{(\ell)}.
    \end{align}
\end{remark}
\subsection{Specifying structure through kernel design}

We now employ the observations of \Cref{sec:kernel_selection} to endow Kernel ROMs with structure.
If the structure of the FOM function $\bff$ is unknown, an RBF kernel is a reasonable general-purpose choice for $K$.
However, if the structure of $\bff$ is known, a feature map kernel can be employed so that the resulting $\hbf$ has the same structure of the intrusive projection-based ROM function $\tbf$.
This is best shown by example.

\begin{example}
Consider again the quartic QM ROM \cref{eq:quadratic_qm_rom}.
Using the quartic feature map $\bphi$ of \cref{eq:quartic_feature_map} (with $\bx = \hbq$ and $n_x = r$) to define a feature map kernel $K(\hbq, \hbq')=\bphi(\hbq)\trp\bG\bphi(\hbq')$, the Kernel ROM \cref{eq:kernel-rom} takes the form
\begin{align}
    \label{eq:quartic-kernel-rom}
    \begin{aligned}
    \dt\hbq(t)
    &= \bC\bphi(\hbq(t))
    \\
    &= \hbA \hbq(t)
    + \hbH_{2}[\hbq(t)\otimes\hbq(t)]
    + \hbH_{3}[\hbq(t)\otimes\hbq(t)\otimes\hbq(t)]
    + \hbH_{4}[\hbq(t)\otimes\hbq(t)\otimes\hbq(t)\otimes\hbq(t)],
    \end{aligned}
\end{align}
in which $\bC = \bOmega\trp\bphi(\hbQ)\trp\bG = [~\hbA~~\hbH_{2}~~\hbH_{3}~~\hbH_{4}~]$.
This ROM has the same dynamical structure as \cref{eq:quadratic_qm_rom} but can be constructed non-intrusively.
The structure can be tailored by adjusting the feature map: if the FOM~\cref{eq:linearquadratic_fom} is linear ($\bH = \bf0$), then the intrusive QM ROM \cref{eq:quadratic_qm_rom} simplifies to a quadratic form,
\begin{align}
    \dt\tbq(t)
    = \tbA\tbq(t)
    + \tbH_{1}[\tbq(t)\otimes\tbq(t)],
    \qquad
    \tbH_{1} = \bV\trp\bA\bW,
\end{align}
which can be mimicked by a Kernel ROM by employing a linear-quadratic feature map as in \cref{eq:quadratic_fm}.
\end{example}

\begin{remark}[Input terms]\label{rmk:input_terms}
Kernel ROMs can be designed to account for known input terms by including them in the feature map.
Suppose we wish to construct a ROM with the structure
\begin{align}
    \label{eq:rom-with-inputs}
    \dt\tbq(t)
    = \tbf(\tbq(t)) + \tbb(\bu(t)),
\end{align}
where $\tbb:\real^{n_u}\to\real^{r}$ and $\bu:[0, T]\to \real^{n_u}$ model, for example, time-varying boundary conditions or forcing terms.
In this case, we can construct feature maps $\bphi_q$ and $\bphi_u$ which aim to emulate the structures of $\tbf$ and $\tbb$, respectively, and define a concatenated feature map
\begin{align}
    \bphi(\hbq; \bu)
    = \begin{bmatrix}
        \bphi_q(\hbq) \\ \bphi_u(\bu)
    \end{bmatrix}.
\end{align}
The resulting Kernel ROM has the form
\begin{align}
    \dt\hbq(t)
= \bC_q\bphi_q(\hbq(t)) + \bC_u\bphi_u(\bu(t)),
\end{align}
whose structure can be tailored to that of \cref{eq:rom-with-inputs} by designing $\bphi_q$ and $\bphi_u$ appropriately.
\end{remark}

As discussed in \Cref{sec:kernel_selection}, feature map kernels can lead to cost savings over generic kernels.
Let $n_\phi$ be the dimension of the feature map, i.e., $\bphi(\hbq)\in\real^{n_\phi}$.
Because the matrix $\bC \in \real^{r \times n_\phi}$ can be computed once and reused, the cost of evaluating the ROM function $\hbf$ online is $\mathcal{O}(r n_\phi)$.
Hence, if $n_\phi < Mn_t$, a feature map kernel is less expensive to evaluate than a generic kernel.
If $n_\phi > Mn_t$ (e.g., due to a moderate reduced state dimension $r$), it can be beneficial to reduce $n_\phi$ and add a more generic element to the kernel to compensate. For instance, in place of the quartic ROM \cref{eq:quartic-kernel-rom}, we may choose a quadratic feature map and add an RBF term to account for the cubic and quartic nonlinearities, resulting in a ROM of the form
\begin{align}
    \label{eq:linquad-rbf-hybrid}
    \dt\hbq(t)
    = \hbA\hbq(t) + \hbH[\hbq(t) \otimes \hbq(t)]
    + c_\psi \bOmega\trp\bpsi_{\!\epsilon}(\hbq(t)),
\end{align}
where $\bpsi_{\!\epsilon}:\real^{r}\to\real^{M(n_t+1)}$ is as in \cref{eq:rbf-psi-def} and $c_\psi > 0$ is a weighting coefficient as in \cref{eq:hybrid_kernel_interpolant}.
We test ROMs with this hybrid structure in \Cref{sec:numerics}.
Note that this strategy can also apply to cases where the desired ROM structure is only partially known or representable by a feature map kernel.

\subsection{Comparison to operator inference}

Our kernel-based method is philosophically similar to the operator inference (OpInf) framework pioneered in \cite{BPeherstorfer_KWillcox_2016a}, with a few key differences.
Like our method, OpInf stipulates the form of a ROM based on structure that arises from intrusive projection, and the objects defining the ROM are learned from a regression problem of reduced states and corresponding time derivatives.
However, the learning problems in each approach use different candidate function spaces and regularizers, resulting in different ROMs even when the same training data and model structure are used for both procedures.

Generally speaking, OpInf constructs ROMs of the form
\begin{subequations}\label{eq:opinf-rom-general}
\begin{align}
    \dt\hbq(t)
    = \hbO\bphi(\hbq(t))
\end{align}
for a specified feature map $\bphi:\real^{r}\to\real^{n_\phi}$ by solving the regularized residual minimization problem
\begin{align}\label{eq:opinf_regression-general}
    \min_{\hbO\in \real^{r\times n_\phi}} \;
    \sum_{\ell=1}^{M}\sum_{k=0}^{n_t}
    \norm{\dot{\hbq}_{k}^{(\ell)} - \hbO\bphi(\hbq_{k}^{(\ell)})}_2^2
    + \|\bGamma\hbO\trp\|_F^2,
\end{align}
\end{subequations}
where $\bGamma\in\real^{n_\phi\times n_\phi}$.
This is the generic learning problem \cref{eq:generic_nonintrusive_regression} with the function space $S$ given by
\begin{align}
    S_{\bphi} = \set{
        \bs:\real^r\to \real^r \; : \;
        \bs(\hbq) = \hbO\bphi(\hbq)
        \quad\textup{for some}\quad
        \hbO\in\real^{r\times n_\phi}
    },
\end{align}
and where $R$ is a Tikhonov regularizer.
The so-called operator matrix $\hbO$ satisfies the linear system
\begin{align}
    \label{eq:opinf-solution-system}
    (\bD\trp\bD + \bGamma\trp\bGamma)\hbO\trp
    = \bD\trp\hbZ\trp,
\end{align}
where $\bD = \bphi(\hbQ)\trp$ and $\hbQ$ and $\hbZ$ are the training data matrices in \cref{eq:training-data-matrices}.
As with our kernel-based approach, the feature map is chosen to emulate the structure of a projection-based ROM.
For example, the OpInf regression to learn a linear-quadratic ROM of the form \cref{eq:linear_qm_rom} is given by
\begin{align}\label{eq:opinf_regression}
    \min_{\hbA\in \real^{r\times r}, \hbH\in \real^{r\times r^2}} \;
    \sum_{\ell=1}^{M}\sum_{k=0}^{n_t}
    \norm{\dot{\hbq}_{k}^{(\ell)} - \left(
        \hbA \hbq_{k}^{(\ell)}
        + \hbH[\hbq_{k}^{(\ell)}\otimes \hbq_{k}^{(\ell)}]
    \right)}_2^2 + \norm{\bGamma[~\hbA~~\hbH~]\trp}_F^2,
\end{align}
and the solution $\hbO = [~\hbA~~\hbH~]$ satisfies \cref{eq:opinf-solution-system} with $\bD = [~\hbQ\trp~~(\hbQ \odot \hbQ)\trp~]\in\real^{M(n_t + 1)\times (r + r^{2})}$.
The underlying feature map in this case is the linear-quadratic map \cref{eq:quadratic_fm}.
In practice, the compressed Kronecker product of \Cref{remark:compressed-kronecker} is used so that \cref{eq:opinf_regression} has a unique solution.

For a kernel ROM with the kernel specified entirely by a feature map, the resulting ROM can be expressed in terms of the training data and the feature map as
\begin{align}
    \begin{aligned}
    \dt\hbq(t)
    &= \hbZ\big(K(\hbQ,\hbQ) + \gamma\bI\big)^{-1}K(\hbQ,\hbq(t))
    \\
    &= \underbrace{\hbZ\big(\bphi(\hbQ)\trp\bG\bphi(\hbQ) + \gamma\bI\big)^{-1}\bphi(\hbQ)\trp\bG}_{\bC}\bphi(\hbq(t)),
    \end{aligned}
\end{align}
whereas the OpInf ROM with the same training data and feature map is given by
\begin{align}
    \dt\hbq(t)
    &= \underbrace{\hbZ\bphi(\hbQ)\trp\big(\bphi(\hbQ)\bphi(\hbQ)\trp + \bGamma\trp\bGamma\big)^{-1}}_{\hbO}\bphi(\hbq(t)).
\end{align}
These models share the same nonlinear structure due to the final term $\bphi(\hbq(t))$, but the coefficients on the feature map differ: the Kernel ROM coefficient matrix $\bC$ solves the $M(n_t+1)\times M(n_t+1)$ linear system~\cref{eq:kernelrom-solution-system}, while the solution $\hbO$ to the OpInf regression satisfies an $n_\phi\times n_\phi$ linear system~\cref{eq:opinf-solution-system}.
Furthermore, OpInf is in general restricted to the feature map formulation \cref{eq:opinf-rom-general}, though it has in some cases been augmented with additional nonlinear terms through, e.g., the discrete empirical interpolation method \cite{benner2020opinfdeim}; by contrast, Kernel ROMs can be designed to have general nonlinear (RBF) structure or hybrid structure such as in \cref{eq:linquad-rbf-hybrid}, depending on the choice of kernel.
Finally, establishing error bounds is an open problem for OpInf ROMs, whereas Kernel ROMs inherit properties from the underlying RKHS which lead to error estimates.

\section{Error estimates}\label{sec:error_estimates}

We now derive several \emph{a posteriori} error estimates for Kernel ROMs that relate the FOM solution $\bq(t)$, the intrusive ROM solution $\tbq(t)$, and the Kernel ROM solution $\hbq(t)$.
Mathematical error estimates for non-intrusive ROMs are uncommon in the literature, typically restricted to special cases, and often difficult to obtain. Even in the intrusive case, error estimates for QM ROMs are not well-studied. One advantage of our kernel methods-based approach is that certain error estimates are straightforward to derive. This section presents an error bound for a general, non-intrusive kernel ROM with a generic nonlinear decoder $\bg$. One immediate consequence is an intrusive error bound (\Cref{cor:galerkin_error}), which we include to fill the gap in the literature for intrusive QM ROMs.

These results require three main ingredients: the so-called local logarithmic Lipschitz constant, a Gr{\"o}nwall-type inequality, and standard error results for kernel interpolants.
In this section, $\bM\in \real^{n_x\times n_x}$ denotes a symmetric positive definite weighting matrix with Cholesky factorization $\bM=\bL\bL\trp$.
The $\bM$-weighted inner product and norm are denoted with $\innerprod{\bx,\by}_\bM \coloneqq \bx\trp\bM\by = \innerprod{\bL\trp\bx,\bL\trp\by}$ and $\norm{\bx}_\bM \coloneqq \innerprod{\bx, \bx}_\bM^{1/2} = \|\bL\trp\bx\|_2$, respectively.

\subsection{Preliminaries}\label{sec:error_preliminaries}

We begin with the definition of the local logarithmic Lipschitz constant.
The reader is directed to, e.g.,~\cite{GSoderlind_2006a, DWirtz_DCSorensen_BHaasdonk_2014a} for a more complete overview.

\begin{definition}\label{def:local_log_lipschitz}
    For a function $\bb:\real^{n_x}\to \real^{n_x}$, the local logarithmic Lipschitz constant at $\bx \in \real^{n_x}$ with respect to $\bM$ is defined as
    \begin{align}\label{eq:local_log_lipschitz}
        \Lambda_{\bM}[\bb](\bx) \coloneqq \sup_{\bz \in \real^d} \frac{\innerprod{\bz-\bx, \bb(\bz)-\bb(\bx)}_\bM}{\norm{\bz-\bx}_\bM^2}.
    \end{align}
\end{definition}

The local logarithmic Lipschitz constant can be seen as a nonlinear generalization of the logarithmic norm of a matrix.

\begin{definition}[Logarithmic norm]\label{def:log_norm}
    The logarithmic norm of a matrix $\bB \in \real^{n_x\times n_x}$ with respect to $\bM$ is defined as
    \begin{align}\label{eq:log-norm}
        \lambda_{\bM}(\bB) \coloneqq \sup_{\bx \in \real^{n_x}} \frac{ \innerprod{\bx, \bB \bx}_\bM}{\norm{\bx}_\bM^2} = \max \sigma \left( \frac{1}{2} \left( \tbB +  \tbB\trp\right) \right),
    \end{align}
    where $\sigma(\frac{1}{2} ( \tbB +  \tbB\trp))$ is the spectrum of $\frac{1}{2} ( \tbB +  \tbB\trp)$ and $\tbB=\bL\trp\bB\bL^{-\mathsf{T}}.$
\end{definition}

If $\bb$ is an affine function, i.e., $\bb(\bx) = \bB \bx + \bd$ for some $\bB\in\real^{n_x\times n_x}$ and $\bd\in\real^{n_x}$, then $\Lambda_\bM[\bb](\bx) =  \lambda_\bM(\bB)$.
Note that the local logarithmic Lipschitz constant and the logarithmic norm can  be negative, unlike a standard Lipschitz constant.
We also note that if $\bb$ is differentiable, then $\Lambda_{\bM}[\bb](\bx)$ can be approximated by the logarithmic norm of the Jacobian $\bb'(\bx)$:

\begin{align*}
    \frac{\innerprod{\bz-\bx, \bb(\bz)-\bb(\bx)}_\bM}{\norm{\bz-\bx}_\bM^2}
    &= \frac{\innerprod{\bz-\bx, \bb'(\bx)(\bz-\bx)}_\bM}{\norm{\bz-\bx}_\bM^2}+ \mathcal{O}(\norm{\bz-\bx}_\bM)
    \approx \frac{\innerprod{\bz-\bx, \bb'(\bx)(\bz-\bx)}_\bM}{\norm{\bz-\bx}_\bM^2}
    \leq \lambda_\bM(\bb'(\bx)).
\end{align*}

We also need the following Gr{\"o}nwall-type inequality.

\begin{lemma}[Gr{\"o}nwall inequality]\label{lem:gronwall}
    Let $T>0$ and $\alpha, \beta:[0, T]\to \real$ be integrable functions.
    If $u:[0, T]\to \real$ is differentiable and satisfies
    $
        u'(t)\leq \beta(t)u(t)+\alpha(t)
    $
    for all $t \in (0, T)$, then
    \begin{align*} u(t)\leq \int_0^t \alpha(s) e^{\int_s^t \beta(\tau)d\tau}ds + e^{\int_0^t \beta(\tau)d\tau} u(0)
    \end{align*}
    for any $0 \le t \le T$.
\end{lemma}

See, e.g., \cite[Lemma 2.6]{DWirtz_DCSorensen_BHaasdonk_2014a} for a proof.

\subsection{Error bounds}

We now present an \emph{a posteriori} error analysis for Kernel ROMs, which follows the approach detailed in \cite{DWirtz_DCSorensen_BHaasdonk_2014a}.
The strategy is to view the Kernel ROM function $\hbf$ as a regularized kernel interpolant of the intrusive projection-based ROM function $\tbf$, plus a discrepancy term $\bdelta$ that accounts for the approximation error between $\tbf(\hbq_{k}^{(\ell)})$ and the time derivative estimates $\dot{\hbq}_{k}^{(\ell)}$ used to train the interpolant.

First, define the Kernel ROM reconstruction error
\begin{align}\label{eq:reconstruction_error}
    \be(t)
    \coloneqq \bq(t) - \bg(\hbq(t)),
\end{align}
where $\bq(t)$ is the solution to the FOM \cref{eq:generic_fom}, $\hbq(t)$ is the solution to the Kernel ROM \cref{eq:nonintrusive_rom}, and $\bg$ is the decompression map \cref{eq:quadratic_decoder}.
The reconstruction error evolves according to the system
\begin{align}
    \label{eq:error_ode_non-intrusive}
    \dt\be(t) = \bff(\bq(t)) - \bg'(\hbq(t))\hbf(\hbq(t)),
    \qquad
    \be(0) = \bq_0 - \bg(\bV\trp(\bq_0-\obq)),
\end{align}
where $\bg'$ is the Jacobian of $\bg$.
Although we use a QM to define the reconstruction mapping $\bg$, the following error analysis holds for any reconstruction mapping.

\begin{theorem}[\emph{A posteriori} error]\label{thm:a_posteriori_error}
If $\hbf$ is an unregularized kernel interpolant of $\tbf+\bdelta\in \cH_K^r$ where $\norm{\bdelta(\hbq(s))}_\bM<\delta(s)$,
    then
    \begin{align}\label{eq:error_estimate}
        \norm{\be(t)}_\bM &\leq \int_0^t (\alpha_P(s)+\alpha_K(s))e^{\int_s^t \beta(\tau)d\tau}ds + e^{\int_0^t \beta(\tau)d\tau}\norm{\be_N(0)}_\bM,
        & \forall & \; t \in [0, T],
    \end{align}
    where
    \begin{subequations}\label{eq:error_functions}
    \begin{align}
        \label{eq:error_functions-aP}
        \alpha_P(s) &=
        \norm{(\bI-\bg'(\hbq(t))\bV\trp)\bff(\bg(\hbq(s)))}_\bM,
        \\
        \label{eq:error_functions-aK}
        \alpha_K(s) &=\norm{\bg'(\hbq(s))}_\bM\left(P_{K, \tbQ}(\hbq(s))\norm{\bL}_2\|\tbf+\bdelta\|_{\cH_K^r} + \delta(s)\right),
        \quad\text{and}
        \\
        \label{eq:error_functions-b}
        \beta(s) &= \Lambda_{\bM}[\bff](\bg(\hbq(s))).
    \end{align}
    \end{subequations}
\begin{proof}
Notice that the evolution equations in \cref{eq:error_ode_non-intrusive} can be rewritten as
\begin{align*}
    \dt\be(t) &= \bff(\bq(t)) -
    \bg'(\hbq(t))\hbf(\hbq(t))
    - \bff(\bg(\hbq(t))) + \bff(\bg(\hbq(t)))
    \\ &\quad
    - \bg'(\hbq(t))\bV\trp\bff(\bg(\hbq(t))) + \bg'(\hbq(t))\bV\trp\bff(\bg(\hbq(t)))
-\bg'(\hbq(t))\bdelta(\hbq(t)) + \bg'(\hbq(t))\bdelta(\hbq(t))
    \\
    &= \bff(\bq(t)) - \bff(\bg(\hbq(t))) + (\bI-\bg'(\hbq(t))\bV\trp)\bff(\bg(\hbq(t)))
    \\ &\quad
    +\bg'(\hbq(t))\left(\bV\trp\bff(\bg(\hbq)) +\bdelta(\hbq(t)) - \hbf(\hbq(t)) \right)
- \bg'(\hbq(t))\bdelta(\hbq(t)).
\end{align*}
Taking the $\bM$-weighted inner product with $\be(t)$ and using the definition of the logarithmic Lipschitz constant and \Cref{cor:vector_power_function_bound} yields
\begin{align*}
    &\innerprod{\be(t), \dt\be(t)}_\bM \\
    &=
    \innerprod{\be(t), \bff(\bq(t)) - \bff(\bg(\hbq(t)))}_\bM
+ \innerprod{\be(t),(\bI-\bg'(\hbq(t))\bV\trp)\bff(\bg(\hbq(t)))}_\bM
    \\ &\quad
    + \innerprod{\be(t), \bg'(\hbq(t))\left(\bV\trp\bff(\bg(\hbq)) +\bdelta(\hbq(t)) - \hbf(\hbq(t)) \right)}_\bM
- \innerprod{\be(t), \bg'(\hbq(t))\bdelta(\hbq(t))}_{\bM} \\
    & \leq \Lambda_{\bM}[\bff](\bg(\hbq(t)))\norm{\be(t)}_\bM^2 + \norm{(\bI-\bg'(\hbq(t))\bV\trp)\bff(\bg(\hbq(t)))}_\bM\norm{\be(t)}_\bM
    \\ &\quad
    + \norm{\bg'(\hbq(t))}_\bM \|\underbrace{\bV\trp\bff(\bg(\hbq(t)))}_{\tbf(\hbq(t))} +\bdelta(\hbq(t)) - \hbf(\hbq(t))\|_\bM\norm{\be(t)}_\bM
+ \norm{\bg'(\hbq(t))}_\bM \underbrace{\norm{\bdelta(\hbq(t))}_\bM}_{\leq \delta(t)} \norm{\be(t)}_\bM\\
    & \leq \Lambda_{\bM}[\bff](\bg(\hbq(t)))\norm{\be(t)}_\bM^2 + \norm{(\bI-\bg'(\hbq(t))\bV\trp)\bff(\bg(\hbq(t)))}_\bM\norm{\be(t)}_\bM \\
    &\quad + \norm{\bg'(\hbq(t))}_\bM P_{K, \tbQ}(\hbq(t))\norm{\bL}_2\|\tbf+\bdelta\|_{\cH_K^r}\norm{\be(t)}_\bM
+ \delta(t)\norm{\bg'(\hbq(t))}_\bM\norm{\be(t)}_\bM.
\end{align*}
Therefore,
\begin{align*}
    \dt \norm{\be(t)}_\bM = \frac{\innerprod{\be(t), \dt\be(t)}_\bM}{\norm{\be(t)}_\bM}
    &\leq \Lambda_\bM[\bff](\bg(\hbq(t)))\norm{\be(t)}_\bM + \norm{(\bI-\bg'(\hbq(t))\bV\trp)\bff(\bg(\hbq(t)))}_\bM \\
    &\quad + \norm{\bg'(\hbq(t))}_\bM \left(P_{K, \hbQ}(\hbq(t))\norm{\bL}_2\|\tbf+\bdelta\|_{\cH_K^r} + \delta(t)\right).
\end{align*}
Applying \Cref{lem:gronwall} yields the result.
\end{proof}
\end{theorem}

A caveat to the result in \Cref{thm:a_posteriori_error} is that it relies on \Cref{cor:vector_power_function_bound}, which requires zero regularization.
However, as we demonstrate empirically in \Cref{sec:numerics}, the error bound \cref{eq:error_estimate} still holds when the regularization hyperparameter $\gamma$ is small.
Secondly, computing the local logarithmic Lipschitz constant $\lambda_\bM[\bff]$ is difficult to do in general.
In practice, we instead approximate it using the logarithmic norm of $\bff'(\bg(\hbq(t)))$.
Similarly, since $\tbf$ is unavailable, we estimate $\|\tbf + \bdelta\|_{\cH_K^r}$ with $\|\hbf\|_{\cH_K^r}$, which is easily computable.
Lastly, we note that the estimate \cref{eq:error_estimate} requires evaluating the FOM right-hand side $\bff$
and its Jacobian. While these evaluations are possible with some black-box PDE solvers, they are not possible in general, which limits the utility of \Cref{thm:a_posteriori_error} for the \emph{a posteriori} certification of a Kernel ROM. Nonetheless, the error estimate \cref{eq:error_estimate} provides mathematical justification for the usage of Kernel ROMs, and our numerical experiments empirically verify the estimate's validity. Using \Cref{thm:a_posteriori_error} for ROM certification would require the non-intrusive estimation of \cref{eq:error_estimate}, i.e., assuming that $\bff$ and $\bff'$ cannot be queried, which is non-trivial and thus left to future work.

We conclude with an \emph{a posteriori} error result for intrusive projection-based ROMs by examining the special case $\hbf=\tbf$.
As mentioned above, error estimates for intrusive QM ROMs are not widely available in the literature, so we provide the following result, which includes intrusive QM ROMs as the special case when $\bg$ is of the form \cref{eq:quadratic_decoder}.
\begin{corollary}\label{cor:galerkin_error}
    If $\hbf=\tbf$, then the following error estimate holds for all $t\in [0,T]$:
    \begin{align}\label{eq:galerkin_error_estimate}
        \norm{\bq(t)-\bg(\tbq(t))}_\bM &\leq
        \int_0^t \norm{(\bI-\bg'(\tbq(s))\bV\trp)\bff(\bg(\tbq(s)))}_\bM e^{\int_s^t \Lambda_\bM[\bff](\bg(\hbq_I(\tau)))d\tau}ds \\
        &\quad + e^{\int_0^t \Lambda_\bM[\bff](\bg(\tbq(\tau)))d\tau}\norm{\bq(0)-\bg(\hbq(0))}_\bM. \nonumber
    \end{align}
\end{corollary}

\section{Numerical results}\label{sec:numerics}

In this section, we test Kernel ROMs on several numerical examples using both POD and QM for dimension reduction.
In each experiment, we construct Kernel ROMs with three kernel designs: 1) a feature map kernel encoding the full structure of the projection-based ROM, abbreviated ``FM''; 2) an RBF kernel, marked ``RBF''; and 3) a feature map-RBF hybrid kernel, labeled ``Hybrid''.
We also compare to the performance of intrusive projection-based ROMs in the first two examples and to OpInf in all three examples.

\subsection{1D Advection-diffusion equation}\label{sec:numerics_adv_diff}

We first consider a linear PDE, the advection-diffusion equation in one spatial dimension with periodic boundary conditions:
\begin{subequations}\label{eq:advection_diffusion_pde}
\begin{align}
    \pdt q(x, t) - \kappa \pdxx q(x, t) + \beta \pdx q(x, t) = 0,
    &\qquad
    x \in (0, 1), \quad t \in (0, T),
    \\
    q(0, t) = q(1, t),
    \quad
    \pdx q(0, t) = \pdx q(1, t),
    &\qquad
    t \in (0, T),
    \\ \label{eq:advection_diffusion_initial_condition}
    q(x, 0) = q_0(x;\bmu) \coloneqq e^{-({x-\mu_1})^2/\mu_2^2},
    &\qquad
    x \in (0, 1).
\end{align}
\end{subequations}
Here, $\kappa>0$ is the diffusion parameter, $\beta\geq 0$ is the advection parameter, $T>0$ is the final time, and $\bmu=(\mu_1,\mu_2)$ parameterizes the initial condition.
For this experiment, we set $\kappa=10^{-2}$, $\beta = 1$, and $T = 1$.
The initial condition is a Gaussian pulse with center $\mu_1\in [0.25, 0.35]$ and width $\mu_2\in [0.05, 0.15]$.
The dynamics of \cref{eq:advection_diffusion_pde} are linear, but advective phenomena can be difficult to capture with linear dimension reduction methods such as POD.

\begin{figure}[t]
    \centering
    \includegraphics[width=\textwidth]{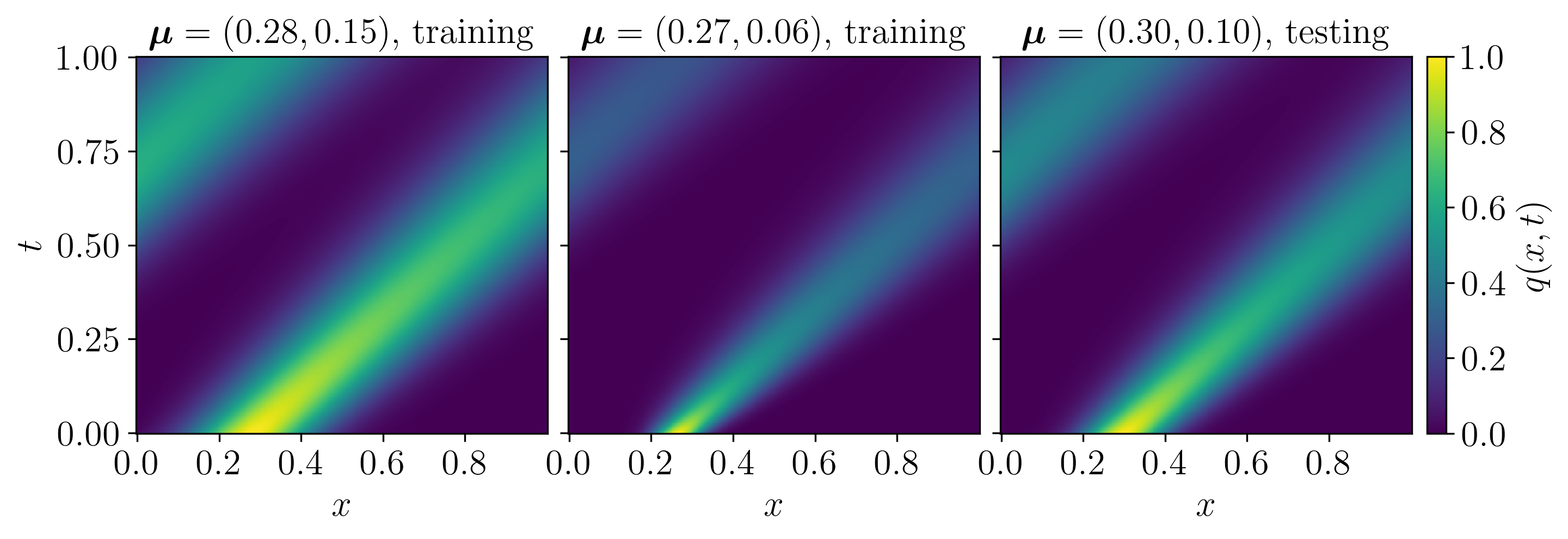}
    \caption{Solutions of the full-order advection diffusion model \cref{eq:advection_diffusion_fom} with initial conditions \cref{eq:advection_diffusion_initial_condition} for various choices of $\bmu$.}
    \label{fig:advdiff_fom_state}
\end{figure}

Spatially discretizing \cref{eq:advection_diffusion_pde} with an upwind finite difference scheme over a grid of $n_q + 1$ uniformly spaced points in the spatial domain $[0, 1]$ results in a linear FOM of the form
\begin{align}
    \label{eq:advection_diffusion_fom}
    \dt \bq(t) &= \bA \bq(t),
    \qquad
    \bq(0) = \bq_0(\bmu),
\end{align}
where $\bq(t),\bq_0(\bmu)\in\real^{n_q}$ and $\bA\in\real^{n_q\times n_q}$.
We use $n_q = 256$ spatial degrees of freedom in this experiment.
To collect training data, we sample $M=10$ initial conditions corresponding to $10$ Latin hypercube samples from the parameter domain $\cD = [0.25, 0.35]\times [0.05, 0.15]$ and integrate the FOM \cref{eq:advection_diffusion_fom} using a fully implicit variable-order backwards difference formula (BDF) time stepper with quasi-constant step size, executed with \texttt{scipy.interpolate.solve\_ivp()} in Python \cite{virtanen2020scipy,LFShampine_MWReichelt_1997a}.
The solution is recorded at $n_t=256$ equally spaced time instances after the initial condition, resulting in $M(n_t+1) = 2570$ total training snapshots.
We also solve the FOM at the testing parameter value $\bar{\bmu} = (0.3,0.1)$, which is not included in the training set.
\Cref{fig:advdiff_fom_state} plots the FOM states for two training parameter values and the testing parameter value.

\begin{table}[t]
\centering
\begin{tabular}{c|cc}
    POD
    & $\bphi(\hbq) = \begin{bmatrix}
        1 \\ \hbq
    \end{bmatrix},
    \qquad
    \bG = \frac{1}{1+r}\bI_{1+r}$
    \\[.5cm] \hline \\[-.25cm]
    QM & $\bphi(\hbq) = \begin{bmatrix}
        1 \\ \hbq \\ \hbq \otimes \hbq
    \end{bmatrix},
    \qquad
    \bG = \begin{bmatrix}
        \bI_{1+r} &\bzero \\ \bzero & \norm{\bW}_F \bI_{r^2}
    \end{bmatrix}$
\end{tabular}
\caption{Feature maps and weighting matrices in POD and QM Kernel ROMs for the 1D advection-diffusion example.}
\label{tbl:advdiff_kernel_table}
\end{table}

The training snapshots are used to compute POD and QM state approximations with the reference vector $\obq$ set to the average training snapshot.
Since the FOM \cref{eq:advection_diffusion_fom} is linear and $\obq\neq\bf0$, the intrusive projection-based POD ROM of dimension $r$ has affine structure,
\begin{align}\label{eq:advdiff_galerkin_rhs_pod}
    \dt\tbq(t)
= \tbc + \tbA \tbq(t),
\end{align}
where $\tbc\in\real^{r}$ and $\tbA\in\real^{r\times r}$, whereas the intrusive QM ROM has the form
\begin{align}\label{eq:advdiff_galerkin_rhs_qm}
    \dt\tbq(t)
= \tbc + \tbA \tbq(t) + \tbH [\tbq(t) \otimes \tbq(t)],
\end{align}
with $\tbc\in\real^{r}$, $\tbA\in\real^{r\times r}$, and $\tbH\in\real^{r\times r^2}$.
For both POD and QM, we construct feature map Kernel ROMs and OpInf ROMs with the corresponding intrusive ROM structure.
The underlying feature maps $\bphi$ and weighting matrices $\bG$ are listed in \Cref{tbl:advdiff_kernel_table}.
Note that the second diagonal block in the weight $\bG$ for QM is scaled by $\norm{\bW}_F$ to account for the fact that $\tbH$ in the intrusive QM ROM \cref{eq:advdiff_galerkin_rhs_qm} also depends on $\bW$.
We also construct an RBF Kernel ROM using a Gaussian kernel-generating RBF $\psi$ (see \Cref{tbl:rbf_kernels}) with fixed shape parameter $\epsilon=10^{-1}$.
This ROM has the same evolution equations in the POD and QM cases, since the compression map $\bh$ is the same in both instances, but we report results for both POD and QM decompression maps $\bg$.
Finally, we construct hybrid Kernel ROMs using the POD feature map from \Cref{tbl:advdiff_kernel_table} with weighting coefficient $c_\phi=1$ and a Gaussian RBF kernel with $\epsilon=0.1$ and weighting coefficient $c_\psi=10^{-3}$, yielding ROMs with the following structure:
\begin{align}\label{eq:advdiff_hybrid_rhs}
    \dt\hbq(t)
= \hbc + \hbA\hbq(t) + 10^{-3}\bOmega\trp\bpsi_{\!\epsilon}(\hbq(t)).
\end{align}
For QM, the RBF term takes the place of the quadratic nonlinearity $\tbH[\tbq(t) \otimes \tbq(t)]$, but for POD, the RBF term is purely supplementary.
Kernel input normalization as in \Cref{remark:kernel_normalization} is not needed in this problem.
Performance is measured with a relative $L^\infty$-$L^2$ error between the FOM and reconstructed ROM states,
\begin{align}\label{eq:relative_linf_l2_error}
    \be(\bq, \hbq) = \frac{\max_{k} \; \norm{\bq(t_k) - \bg(\hbq(t_k))}_2}{\max_{k} \; \norm{\bq(t_k)}_2},
\end{align}
where the ROMs are integrated with the same BDF time stepper as the FOM and the maxima are taken over time indices $k \in \{0, 1, \ldots, n_t\}$.
The ROM error is bounded from below by the projection error $\be(\bq, \bh(\bq))$.

\begin{figure}[t]
    \centering
    \includegraphics[width=\textwidth]{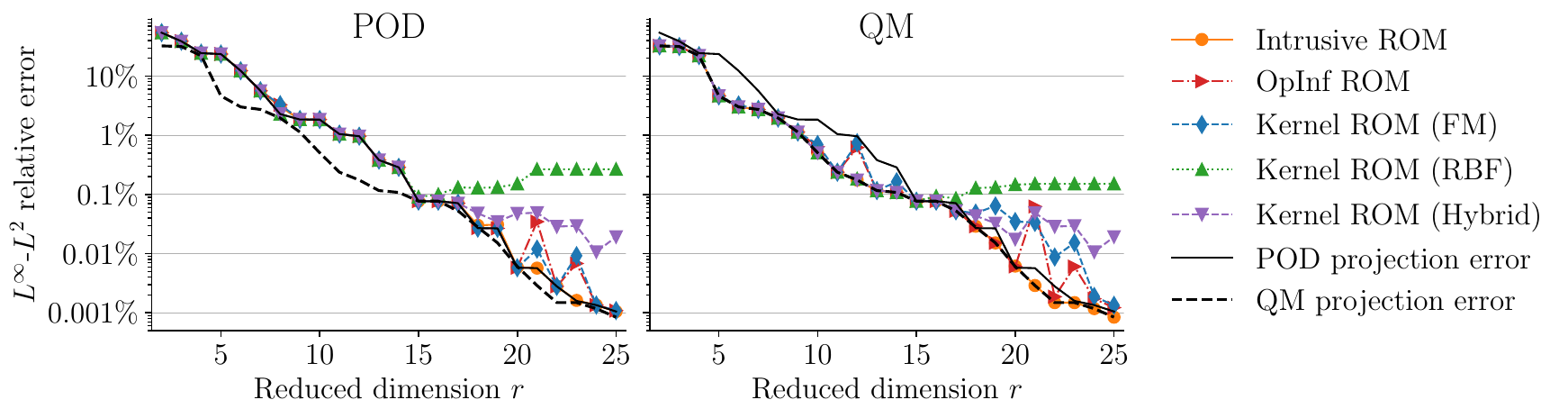}
    \caption{Relative ROM error at the test parameter $\bar{\bmu} = (0.3,0.1)$ as a function of number of basis vectors in linear POD (left) and quadratic manifold (right) reduced state approximations for the advection-diffusion problem \cref{eq:advection_diffusion_pde}.
    }
    \label{fig:advdiff_romsize_vs_error}
\end{figure}

Results are reported in \Cref{fig:advdiff_romsize_vs_error}, which compares ROM and projection errors at the testing parameter value $\bar{\bmu}$ for both POD and QM as a function of the reduced dimension $r$.
For each Kernel ROM, the regularization hyperparameter $\gamma$ for the learning problem \cref{eq:kernel_nonintrusive_regression} is selected to minimize the ROM error over the training data, i.e.,
\begin{align}
    \gamma
    = \underset{\gamma}{\arg\min}
    \sum_{\ell=1}^{M}\sum_{k=0}^{n_t}\big\|\hbq_{k}^{(\ell)} - \hbq(t_k;\bmu_{\ell},\gamma)\big\|_2,
\end{align}
where $\hbq_{k}^{(\ell)}$ are the training snapshots \cref{eq:reduced_state_data} and $\hbq(t;\bmu_{\ell},\gamma)$ denotes the solution to the Kernel ROM with regularization $\gamma$ evaluated for training parameter $\bmu_{\ell}$.
In this experiment, we do this via a grid search over $\gamma\in\{10^{-14},10^{-13},\ldots,10^{2}\}$ for each Kernel ROM.
This procedure is adapted from best practices for OpInf~\cite{SAMcQuarrie_CHuang_KEWillcox_2021a,EQian_IGFarcas_KWillcox_2022a}; a similar selection is carried out for OpInf ROMs with the regularization matrix $\bGamma$ parameterized so that
\begin{align}
    \big\|\bGamma\hbO\trp\big\|_F^2
    = \gamma_1^2(\|\hbc\|_2^2 + \|\hbA\|_{F}^{2}) + \gamma_2^2\|\hbH\|_{F}^{2},
\end{align}
where $\gamma_1,\gamma_2\ge 0$.
This is the state-of-the-art procedure for OpInf and results in accurate ROMs.
Indeed, \Cref{fig:advdiff_romsize_vs_error} shows that each of the POD-based ROMs yield errors that are nearly identical to the POD projection error for $r \leq 15$.
The POD RBF Kernel ROM error plateaus for $r>15$, possibly due to the RBF shape parameter being fixed independent of $r$.
The POD hybrid Kernel ROM error begins to plateau for $r>17$, again possibly due to the fixed RBF shape parameter and fixed weighting coefficients $c_\phi$ and $c_\psi$.
The OpInf ROMs and feature map Kernel ROMs match the projection error for $r \le 20$, but deviate slightly from the projection and intrusive ROM errors for some values of $r>20$.

\begin{figure}[t]
    \centering
    \includegraphics[width=\textwidth]{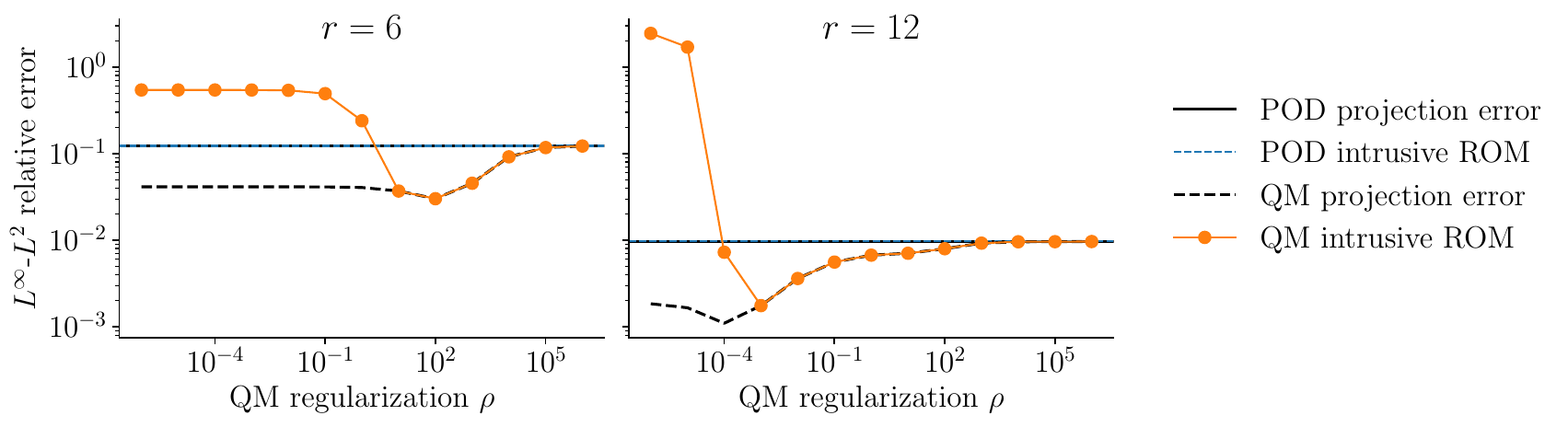}
    \caption{Effect of the quadratic manifold regularization parameter $\rho$ on the projection error and intrusive ROM error for two different reduced state dimensions $r$, using data from the advection-diffusion problem \cref{eq:advection_diffusion_pde} at the test parameter $\bar{\bmu} = (0.3, 0.1)$.}
    \label{fig:advdiff_qm_regularization}
\end{figure}

The QM regularization parameter $\rho \ge 0$ in \cref{eq:qm_optimization_problem} plays an important role in the stability and accuracy of QM ROMs, see \apdxref{appendix:lti_stability_error} for a stability analysis of the intrusive QM ROM for a linear FOM.
\Cref{fig:advdiff_qm_regularization} plots the value of $\rho$ versus the projection error and the intrusive ROM error for two choices of the reduced dimension $r$.
As is evident from \cref{eq:qm_optimization_problem}, $\norm{\bW}_F \to 0$ as $\rho$ increases, which is why the QM projection and QM ROM errors approach their POD counterparts for large enough $\rho$.
Note that the optimal $\rho$ varies with the reduced state dimension $r$. Furthermore, at least for $r = 12$, the best $\rho$ for the reconstruction error is not necessarily the best $\rho$ for the intrusive QM ROM error.
To account for this, the QM results in \Cref{fig:advdiff_romsize_vs_error} report only the best results for each ROM after testing each of the QM regularization values $\rho \in \{10^{-3},10^{-2},\ldots,10^{8}\}$.
In other words, \Cref{fig:advdiff_romsize_vs_error} shows a best-case scenario comparison.
The QM OpInf ROMs and QM feature map Kernel ROMs again show highly similar performance, while the QM RBF and QM hybrid Kernel ROM errors plateau for $r>17$.
Note that the POD and QM projection errors are close for $r>15$, indicating that in this particular problem QM results in diminishing returns over POD for large enough $r$.

\begin{figure}[t]
    \centering
    \begin{subfigure}{\columnwidth}
        \centering
        \includegraphics[width=\textwidth]{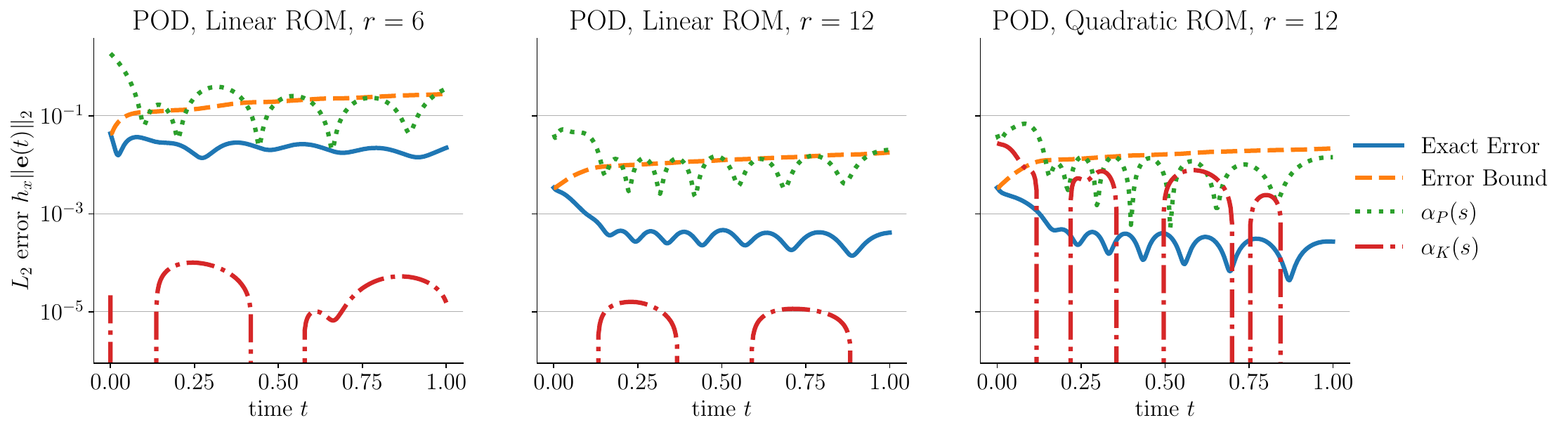}
    \end{subfigure}
    \begin{subfigure}{\columnwidth}
        \centering
        \includegraphics[width=\textwidth]{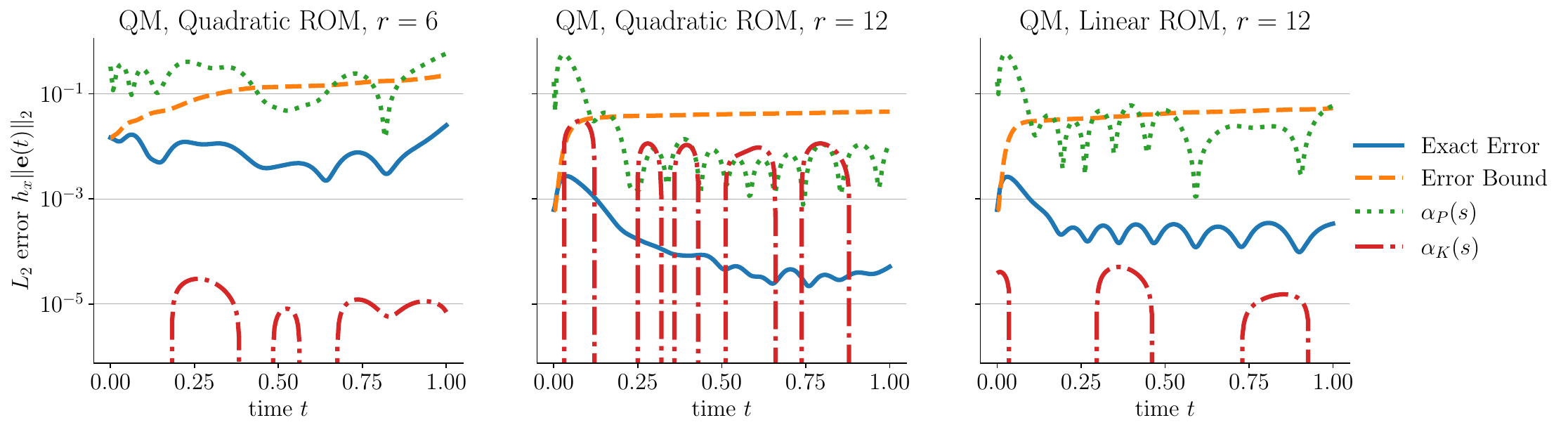}
    \end{subfigure}
\caption{Approximate error bounds for POD and QM Kernel ROMs for the advection-diffusion problem \cref{eq:advection_diffusion_pde}.}
    \label{fig:adv_diff_error_bound}
\end{figure}

Next, we compute the error bound from \Cref{thm:a_posteriori_error} for the feature map Kernel ROMs for $r\in\{6, 12\}$.
Although the computed Kernel ROMs use a nonzero regularization $\gamma\neq 0$, the computed error bounds still hold.
We estimate the norm $\|\tbf+\bdelta\|_{\cH_K^r}$ with the norm of the interpolant $\|\hbf\|_{\cH_K^r}$, which can be computed quickly and explicitly using equation \cref{eq:kernel_norm}.
The local logarithmic Lipschitz constant $\Lambda_\bM[\bff](\bg(\hbq(s)))$ is estimated using the logarithmic norm
$\lambda_\bM(\bff(\hbq(s)))$, and the weighting matrix $\bM$ is taken to be $\frac{1}{r}\bI_{r}$.
We also examine feature map Kernel ROMs where the chosen feature map does not match the true projection-based ROM form, i.e. POD with a quadratic feature map and QM with a linear feature map.
The results are displayed in \Cref{fig:adv_diff_error_bound}, which shows that the computed error estimates indeed bound the true error without dramatically overestimating it.
In the POD cases with linear ROMs, the $\alpha_P$ term, which is related to the POD projection errors, is what dominates the error bound computation, while the $\alpha_K$ term, which corresponds to the pointwise kernel error bound from \Cref{cor:vector_power_function_bound}, is negligible.
For the QM Quadratic ROM with $r=6$, $\alpha_P$ again dominates the error bound and the $\alpha_K$ is negligible.
However, for the QM Quadratic ROM with $r=12$, the $\alpha_K$ term is much larger.
This may indicate that the chosen quadratic feature map may yield a non-optimal model form for the Kernel ROM.
Indeed, since POD with $r=12$ already yields small ROM errors, one may expect that a QM is unnecessary for $r=12$, and thus the quadratic term in the Kernel ROM may be extraneous.
To test this, we remove the quadratic term, which comes from the quadratic component of $\bg$, and compute the error bound for a linear QM Kernel ROM with $r=12$.
We observe that the $\alpha_K$ term is once again negligible in this case.
On the other hand, adding a quadratic term to the POD ROM with $r=12$ also substantially increases $\alpha_K$.
Therefore, we can infer that a larger $\alpha_K$ contribution may indicate that a non-optimal model form (i.e.,\ feature map) was used for the Kernel ROM.

\subsection{1D Burgers' equation}\label{sec:numerics_burgers}

\begin{figure}[t]
    \centering
    \includegraphics[width=\textwidth]{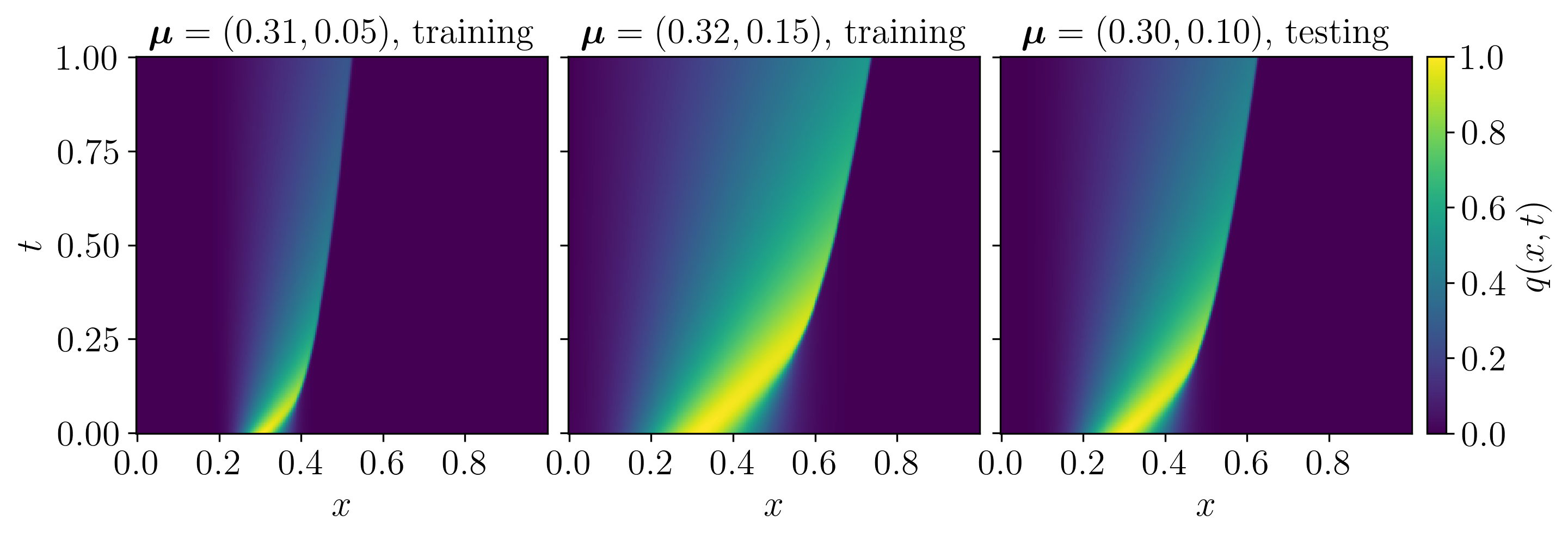}
    \caption{Solutions of the full-order Burgers' model \cref{eq:burgers_fom} with initial conditions \cref{eq:burgers_initial_condition} for various choices of $\bmu$.}
    \label{fig:burgers_fom_state}
\end{figure}

We now consider the 1D viscous Burgers' equation with homogeneous Dirichlet boundary conditions, which is nonlinear with respect to the state:
\begin{subequations}\label{eq:burgers_pde}
    \begin{align}
        \pdt q(x, t) - \nu \pdxx q(x, t) + q(x, t) \pdx q(x, t) = 0,
        &\qquad
        x \in (0, 1), \quad t \in (0, T), \\
        q(0, t) = 0, \quad q(1, t) = 0,
        &\qquad t \in (0, T), \\
        \label{eq:burgers_initial_condition}
        q(x, 0) = q_0(x;\bmu) \coloneqq e^{-(x - \mu_1)^{2}/\mu_2^{2}},
        &\qquad x \in (0, 1).
    \end{align}
\end{subequations}
Here, $\nu>0$ is the viscosity, which we set to $\nu = 10^{-4}$ for our experiments.
Solutions to this system are characterized by sharp gradients along an advection front.
Just as in the previous problem, we consider parameterized Gaussian initial conditions with center $\mu_1\in[0.25,0.35]$ and $\mu_2\in[0.05,0.15]$, set the final time to $T=1$, use $n_q = n_t = 256$ spatial degrees of freedom and temporal observations, and draw $M = 10$ latin hypercube samples of the parameters $\bmu = [~\mu_1~~\mu_2~]\trp$ to use for generating training data.
The spatial discretization uses a uniform centered finite difference for the second derivative term and an upwind finite difference for the first derivative term, yielding a quadratic FOM of the form
\begin{align}
    \label{eq:burgers_fom}
    \dt \bq(t)
    = \bA \bq(t) + \bH[\bq(t)\otimes\bq(t)],
    \qquad
    \bq(0) = \bq_0(\bmu),
\end{align}
where $\bq(t),\bq_0(\bmu)\in\real^{n_q}$, $\bA\in\real^{n_q\times n_q}$, and $\bH\in\real^{n_q\times n_q^2}$.
We again use a BDF time integrator to solve the FOM (and constructed ROMs) at the parameter samples, resulting in $M = 10$ trajectories of $n_t + 1 = 257$ snapshots each.
The FOM states for a few parameter values are displayed in \Cref{fig:burgers_fom_state}.

\begin{table}[t]
\centering
\begin{tabular}{c|c}
    POD & $\bphi(\hbq) = \begin{bmatrix}
        \hbq \\ \hbq \otimes \hbq
    \end{bmatrix},
    \qquad
    \bG = \frac{1}{r+r^2}
    \bI_{r+r^2}$
    \\[.65cm] \hline \\[-.25cm]
    QM & $\bphi(\hbq) = \begin{bmatrix}
        \hbq \\ \hbq \otimes \hbq \\ \hbq \otimes \hbq \otimes \hbq \\ \hbq \otimes \hbq \otimes \hbq \otimes \hbq
    \end{bmatrix},
    \qquad
    \bG = \begin{bmatrix}
            \bI_{r+r^2} &\bzero &\bzero \\
            \bzero & \norm{\bW}_F \bI_{r^3} &\bzero \\
            \bzero & \bzero & \norm{\bW}_F^2 \bI_{r^4} \\
    \end{bmatrix}$
\end{tabular}
\caption{Feature maps and weighting matrices for POD and QM Kernel ROMs for the 1D Burgers' example.}
\label{tbl:burgers_kernel_table}
\end{table}

For both POD and QM, we use $\obq = \bf0$, hence the intrusive POD ROM takes the quadratic form \cref{eq:linear_qm_rom}, whereas the intrusive QM ROM has the quartic form \cref{eq:quadratic_qm_rom}.
We therefore construct feature map Kernel ROMs to mirror the structure of the intrusive models by using the feature maps and weighting matrices listed in \Cref{tbl:burgers_kernel_table}.
While the intrusive ROM informs the choice of feature map, we omit a comparison to the intrusive ROM since Galerkin ROMs are known to struggle in advection-dominated settings---often requiring additional stabilization or the usage of a least-squares Petrov-Galerkin projection (see, e.g., \cite{carlberg2017galerkin})---and focus instead on comparisons between the Kernel and OpInf ROMs and the projection error.
Similar to before, we learn Gaussian RBF Kernel ROMs with fixed shape parameter $\epsilon=10^{-1}$ and hybrid Kernel ROMs using the POD feature map from \Cref{tbl:burgers_kernel_table} with weighting coefficient $c_\phi=1$ and a Gaussian RBF kernel with $\epsilon=0.1$ and weighting coefficient $c_\psi=10^{-3}$, which result in ROMs of the form
\begin{align}\label{eq:burgers_hybrid_rhs}
    \dt\hbq(t)
= \hbc + \hbA\hbq(t) + \tbH [\tbq(t) \otimes \tbq(t)]
    + 10^{-3}\bOmega\trp\bpsi_{\!\epsilon}(\tbq(t)).
\end{align}
We also learn OpInf ROMs with the intrusive ROM structure, with the regularization designed so
\begin{align}
    \big\|\bGamma\hbO\trp\big\|_F^2
    = \gamma_1^2\|\hbA\|_{F}^{2} + \gamma_2^2\|\hbH\|_{F}^{2}
\end{align}
for the POD OpInf ROM, and
\begin{align}
    \big\|\bGamma\hbO\trp\big\|_F^2
    = \gamma_3^2(\|\hbH_{2}\|_{F}^{2} + \|\hbH_{3}\|_{F}^{2}) + \gamma_4^2\|\hbH_{4}\|_{F}^{2}
\end{align}
for the QM OpInf ROM, performing a grid search for $\gamma_1,\ldots,\gamma_4>0$.
The relative $L^\infty$-$L^2$ error \cref{eq:relative_linf_l2_error} is used to evaluate ROM performance at the testing parameter value $\bar{\bmu} = (0.3, 0.1)$.

\begin{figure}[t]
    \centering
    \includegraphics[width=\textwidth]{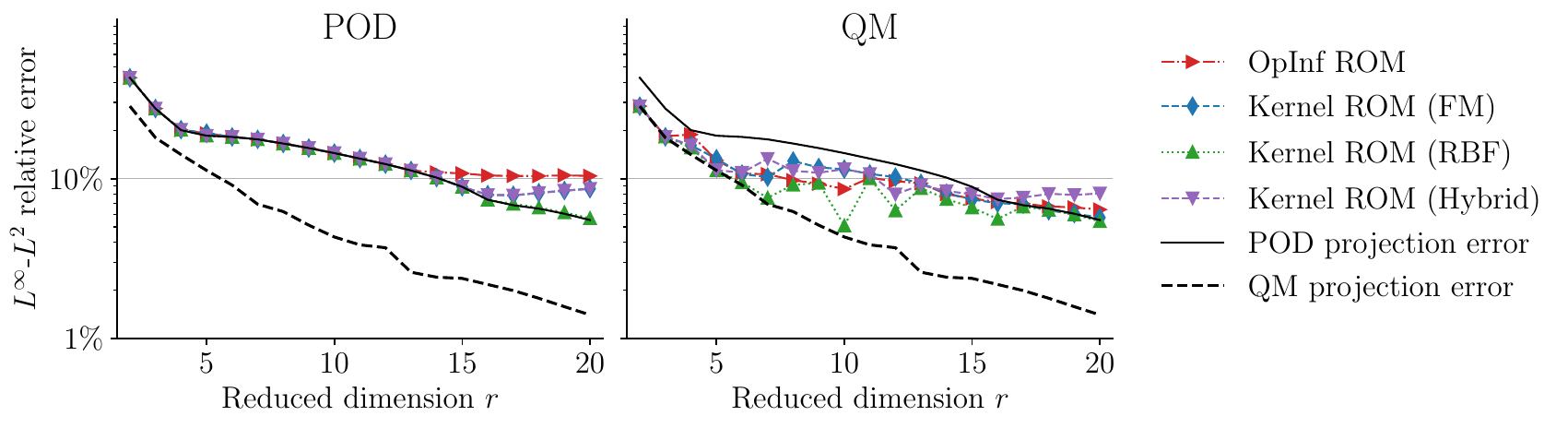}
    \caption{Relative ROM error at the test parameter $\bar{\bmu} = (0.3,0.1)$ as a function of number of basis vectors in linear POD (left) and quadratic manifold (right) state approximations.
    }
    \label{fig:romsize_vs_error-burgers}
\end{figure}

\Cref{fig:romsize_vs_error-burgers} reports results for various reduced dimensions $r$.
The POD ROM errors and the POD projection error are nearly identical for $r\leq 14$. The OpInf POD ROM plateaus in error for $r\geq15$, while the FM and hybrid Kernel ROMs plateau for $r\geq 16$ with slightly improved performance compared to OpInf. The RBF Kernel ROM errors are nearly identical to the projection errors for each value of $r$.
In the QM case, the ROMs yield similar errors to the QM projection error for $r\leq 6$.
For $r>6$, each of the QM ROMs have similar performance, and roughly plateau in error with nearly an order of magnitude difference in comparison to the QM projection error.
The QM OpInf and FM Kernel ROM errors slightly decrease monotonically as $r$ increases, except in the range $8\leq r \leq 12$.
The Hybrid Kernel ROM has similar performance to the FM Kernel and OpInf ROMs, but has slightly worse error for $r\geq 17$.
The RBF Kernel ROM errors have non-monotonic behavior as $r$ increases, but are the lowest for each value of $r$.
The QM ROMs overall obtain smaller errors than the POD projection error for $r\leq 16$, but slightly larger errors for $r>16$.
These results indicate that while QM is effective at decreasing the projection error relative to POD, its benefits are limited when integrated into OpInf or Kernel ROMs in this particular experiment.

\begin{figure}[t]
    \centering
    \begin{subfigure}{\columnwidth}
        \centering
        \includegraphics[width=\textwidth]{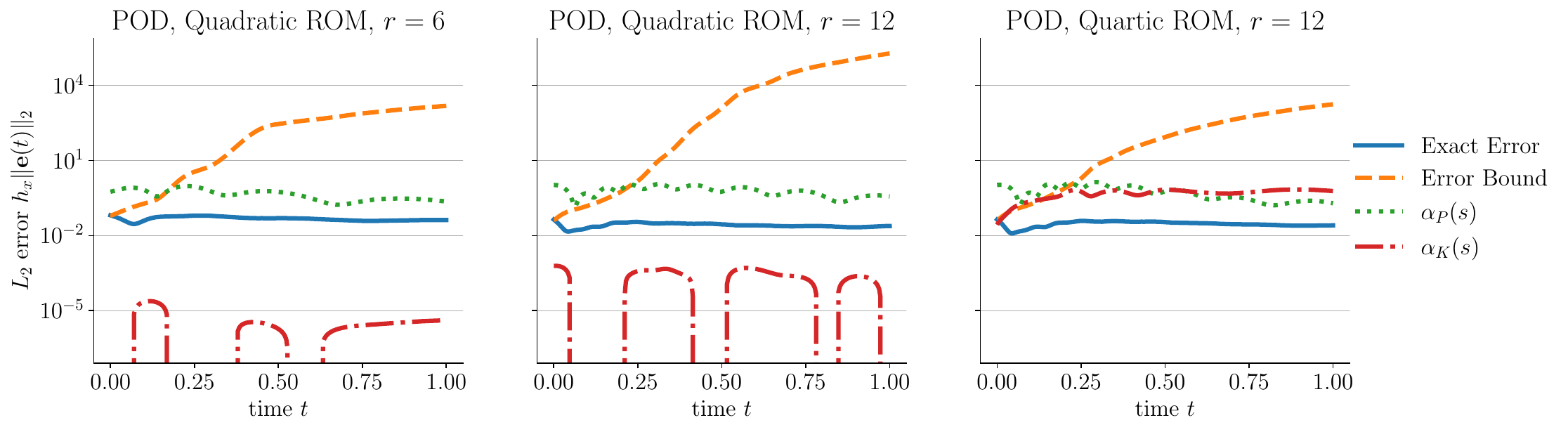}
    \end{subfigure}
    \begin{subfigure}{\columnwidth}
        \centering
        \includegraphics[width=\textwidth]{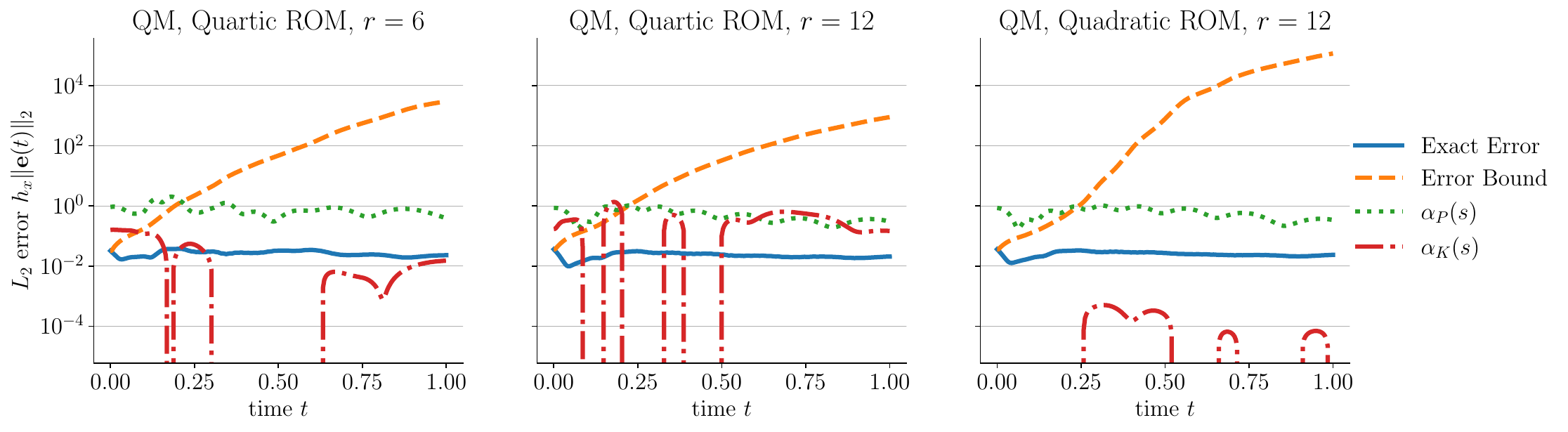}
    \end{subfigure}
\caption{Approximate error bound for POD and QM kernel ROMs for different values of $r$.}
    \label{fig:burgers_error_bound}
\end{figure}

We next compute the error bound from \Cref{thm:a_posteriori_error} for the FM Kernel ROMs for $r=6, 12$.
The quantities $\|\tbf+\bdelta\|_{\cH_K^r}$, $\delta(s)$, $\Lambda_\bM[\bff](\bg(\hbq(s)))$ are estimated in the same way as in the advection-diffusion case.
We use feature map Kernel ROMs corresponding to the quadratic and quartic feature maps in \Cref{tbl:burgers_kernel_table} and examine the cases when the chosen feature map does not match the true projection-based ROM form, i.e. POD with a quartic feature map and QM with a quadratic feature map.
\Cref{fig:burgers_error_bound} displays the results and shows that the computed error estimate again bounds the true error, but is much larger than in the advection-diffusion case. We hypothesize that this difference in magnitude is related to the viscosity parameter $\nu=10^{-4}$.
In the POD cases with quadratic ROMs, the $\alpha_P$ term dominates the error bound contribution, while the $\alpha_K$ term is negligible in the $r=6$ case, but less negligible in the $r=12$ case.
For the QM quartic ROMs, in the $r=6$ case, the $\alpha_P$ term dominates the error bound evaluation, while for $r=12$, the $\alpha_P$ and $\alpha_K$ terms contribute similarly.
We again compute the error bound for a QM ROM with the cubic and quartic terms removed, which come from the quadratic part of $\bg$, resulting in a QM quadratic ROM.
As in the advection-diffusion example, the $\alpha_K$ term decreases significantly, which may indicate that a quadratic model form may be the better choice for a QM Burgers ROM.
To again test if an incorrect model form significantly increases $\alpha_K$, we compute a POD quartic ROM and observe that $\alpha_K$ is much larger than for the POD quadratic ROM, as expected.
This further evidences that a larger $\alpha_K$ contribution may indicate that a non-optimal model form is being used for the Kernel ROM.

\subsection{2D Euler--Riemann problem}\label{sec:numerics_euler_riemann}

Our last numerical example uses the 2D conservative Euler equations
\begin{align}\label{eq:euler_pde}
    \pdt
    \begin{bmatrix}
        \rho \\ \rho u \\ \rho v \\ \rho E
    \end{bmatrix}
    + \pdx
    \begin{bmatrix}
        \rho u \\ \rho u^2+p \\ \rho u v \\ (E+p) u
    \end{bmatrix}
    + \pdx
    \begin{bmatrix}
        \rho v  \\ \rho u v \\ \rho v^2+p \\ (E+p) v
    \end{bmatrix}
    = 0,
\end{align}
where $u$ is the $x$-velocity, $v$ is the $y$-velocity, $\rho$ is the fluid density, $p$ is the pressure, and $E$ is the energy.
The system is closed by the state equation
\begin{align}\label{eq:state_equation}
    p = (\gamma - 1)\left(\rho E - \frac{1}{2} \rho (u^2 + v^2)\right),
\end{align}
where $\gamma=1.4$ is the specific heat ratio.
The spatial domain is the unit square $\Omega = (0,1) \times (0,1)$ with homogeneous Neumann boundary conditions on each side, and the time domain is $(0, 0.8)$.

The initial condition is given by a classical Riemann problem as follows.
The spatial domain is divided into four quadrants with a vertical dividing line at $x=0.8$ and a horizontal dividing line at $y=0.8$.
The initial pressure is set to $p_{BL} = 0.029$ in the bottom left quadrant; in the top right quadrant, the initial velocities are fixed at $u_{TR} = v_{TR} = 0$, and the initial density is $\rho_{TR} = 1.5$.
We parameterize the initial condition by setting the upper-right quadrant pressure to $p_{TR} \in \set{0.5, 0.75, 1.0, 1.25, 1.5}$ and compute remaining quantities following the relations in \cite[Configuration 3]{CWSchulz-Rinne_1993a}.
For testing, we consider the initial upper-right quadrant pressure to $\bar{p}_{TR} = 1.125$.
In every case, the discontinuities of the initial condition propagate through the domain, a highly challenging scenario for projection-based model reduction.

\begin{figure}[t]
    \centering
    \begin{subfigure}{0.32\columnwidth}
        \centering
        \includegraphics[width=\textwidth]{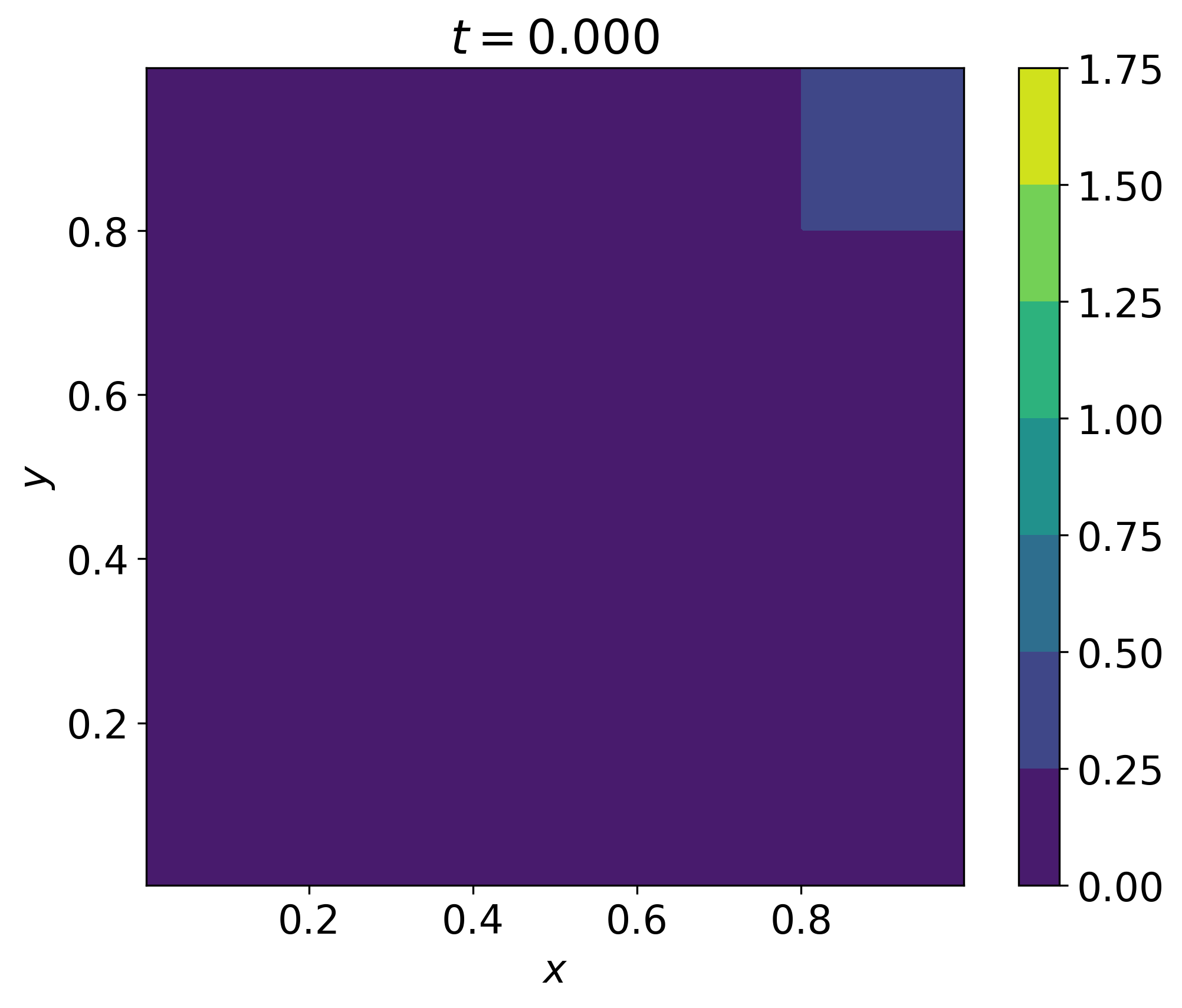}
        \includegraphics[width=\textwidth]{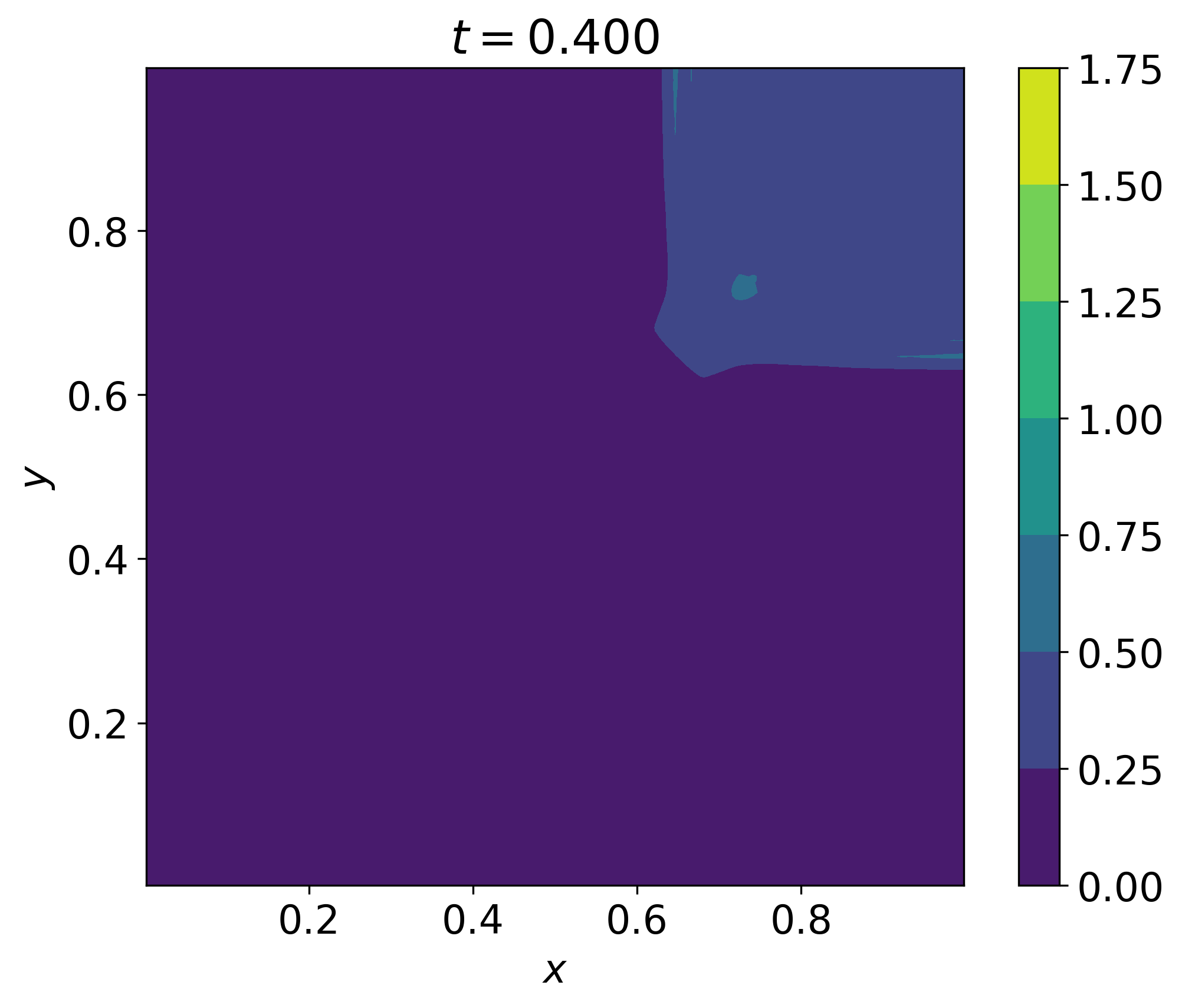}
        \includegraphics[width=\textwidth]{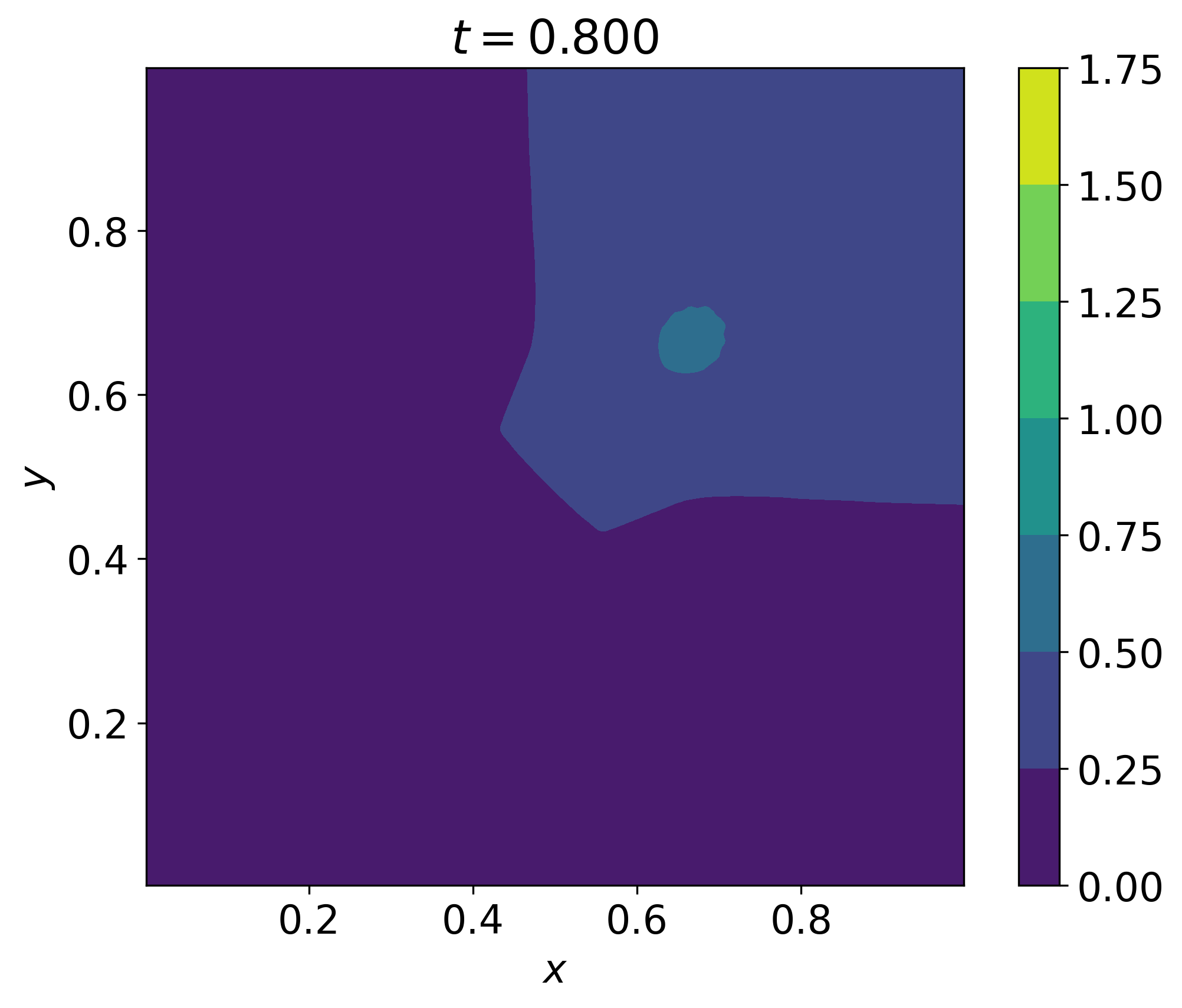}
        \caption{Training, $p_{TR} = 0.5$}
    \end{subfigure}
    \begin{subfigure}{0.32\columnwidth}
        \centering
        \includegraphics[width=\textwidth]{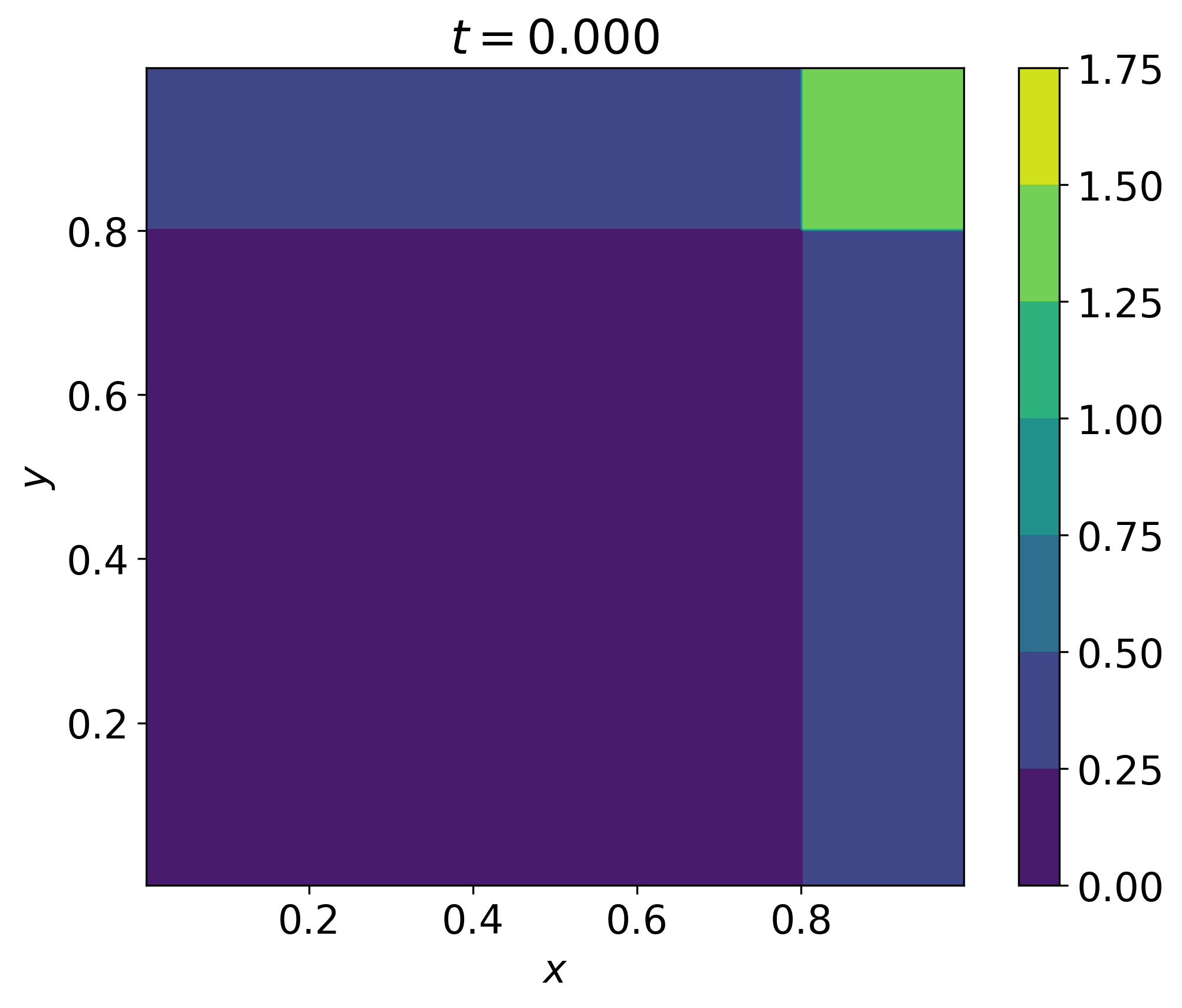}
        \includegraphics[width=\textwidth]{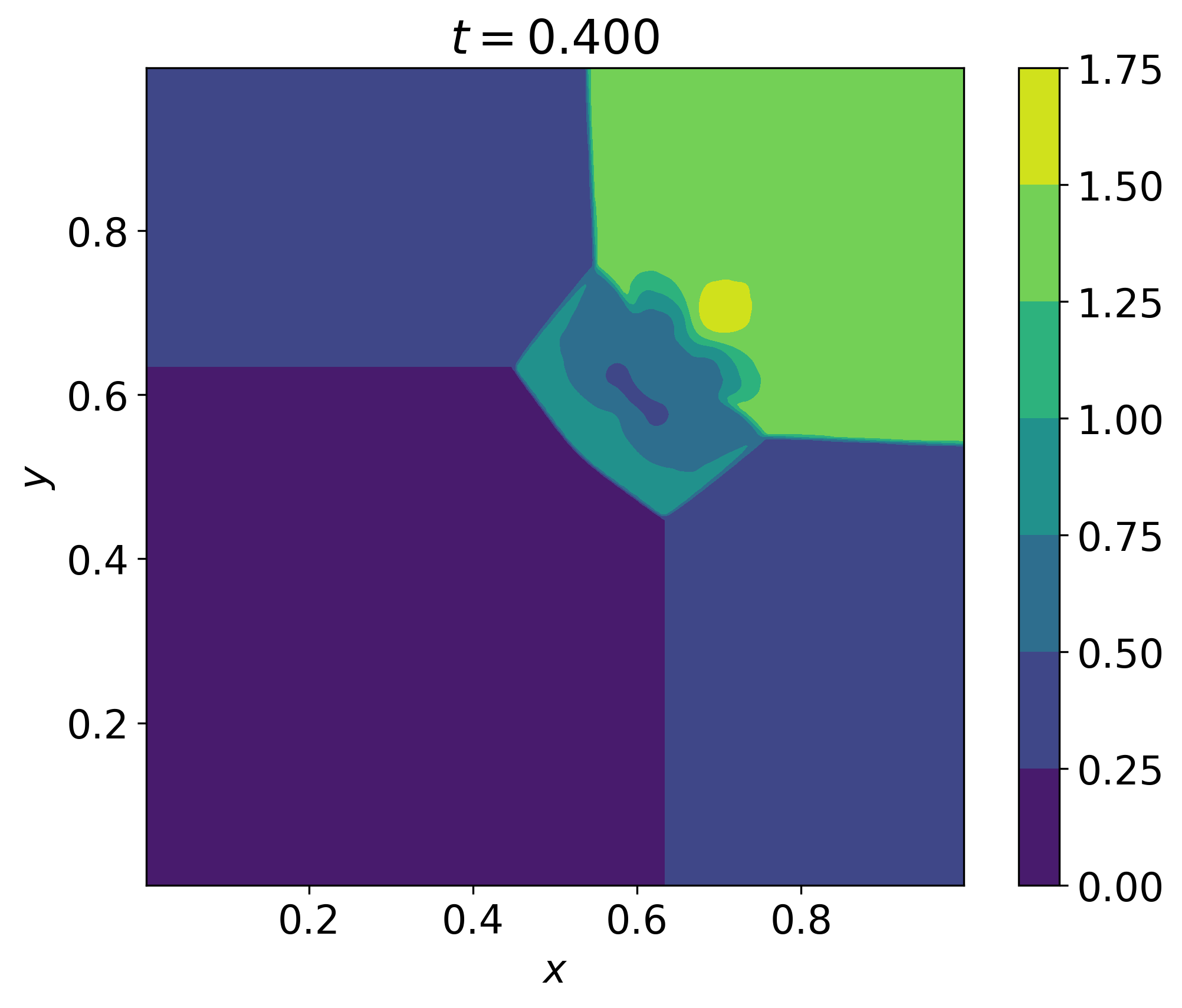}
        \includegraphics[width=\textwidth]{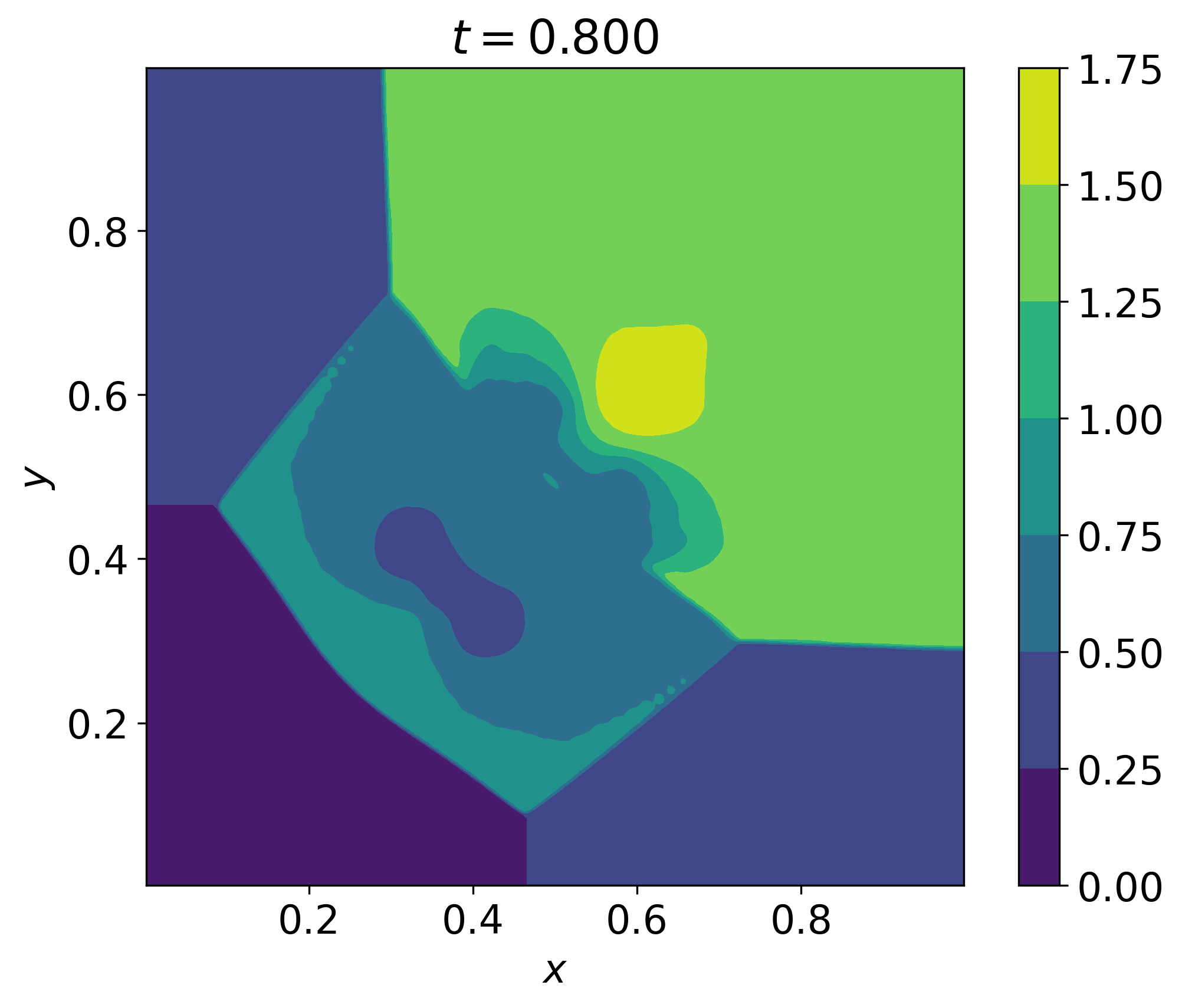}
        \caption{Training, $p_{TR} = 1.5$}
    \end{subfigure}
    \begin{subfigure}{0.32\columnwidth}
        \centering
        \includegraphics[width=\textwidth]{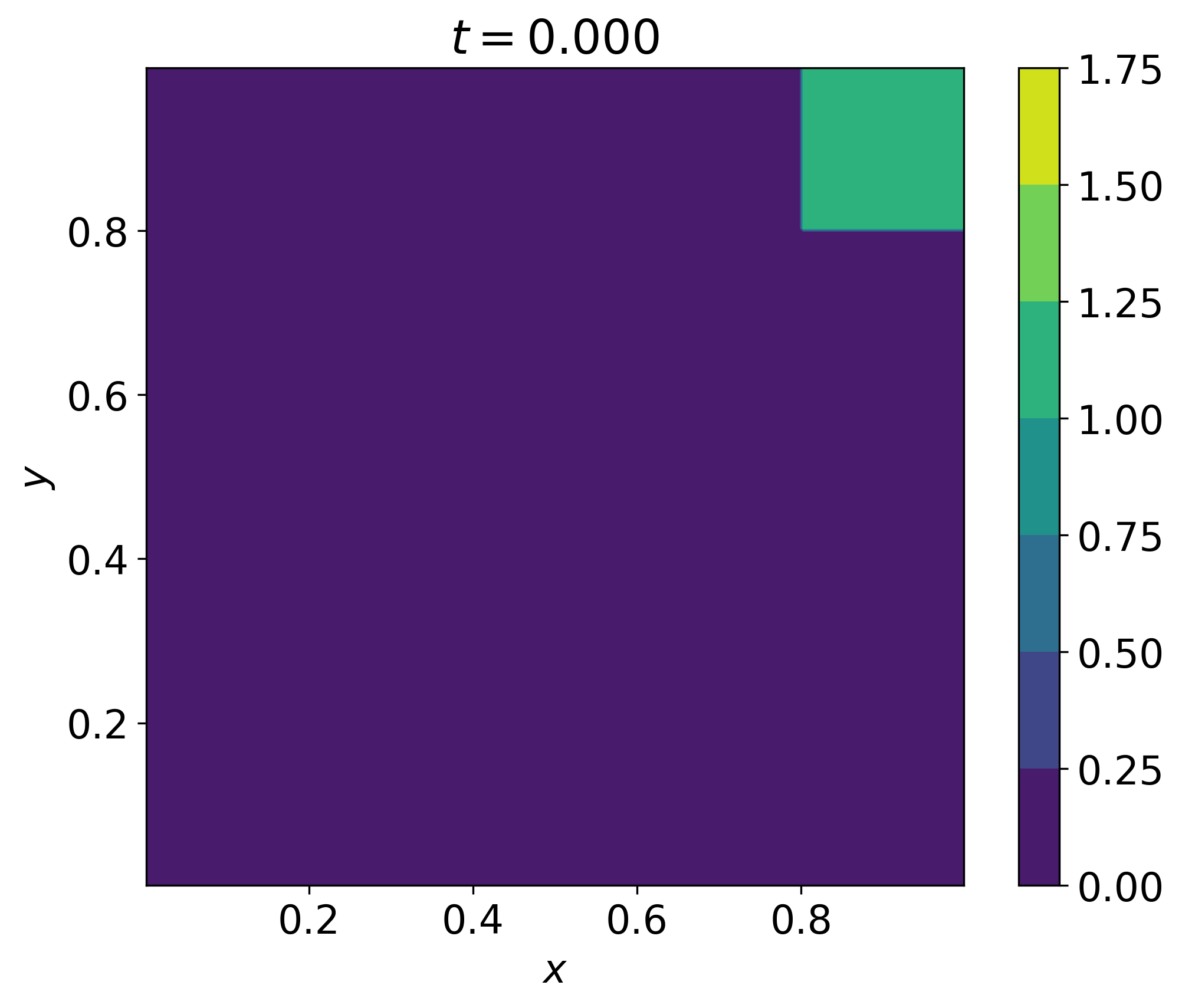}
        \includegraphics[width=\textwidth]{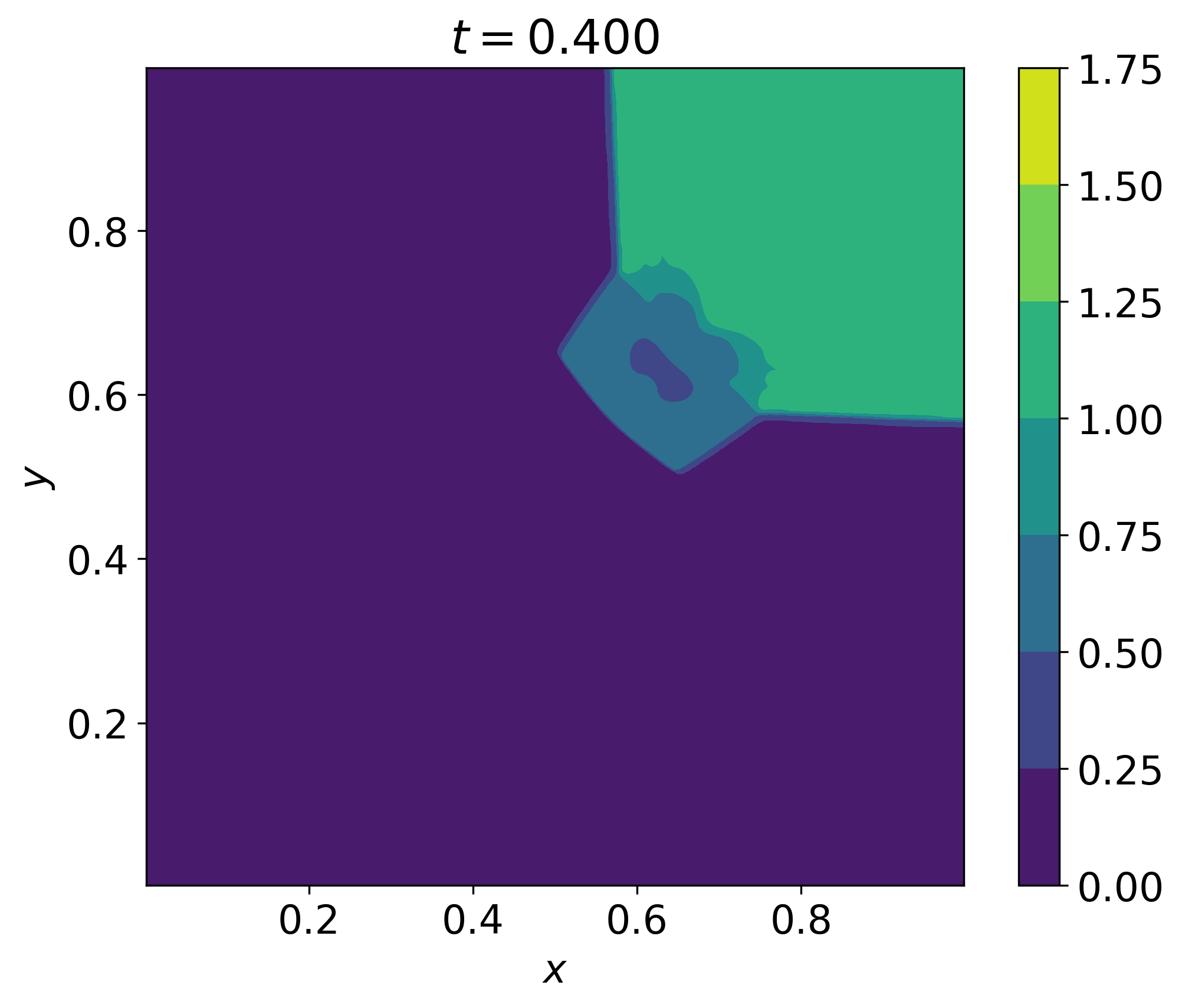}
        \includegraphics[width=\textwidth]{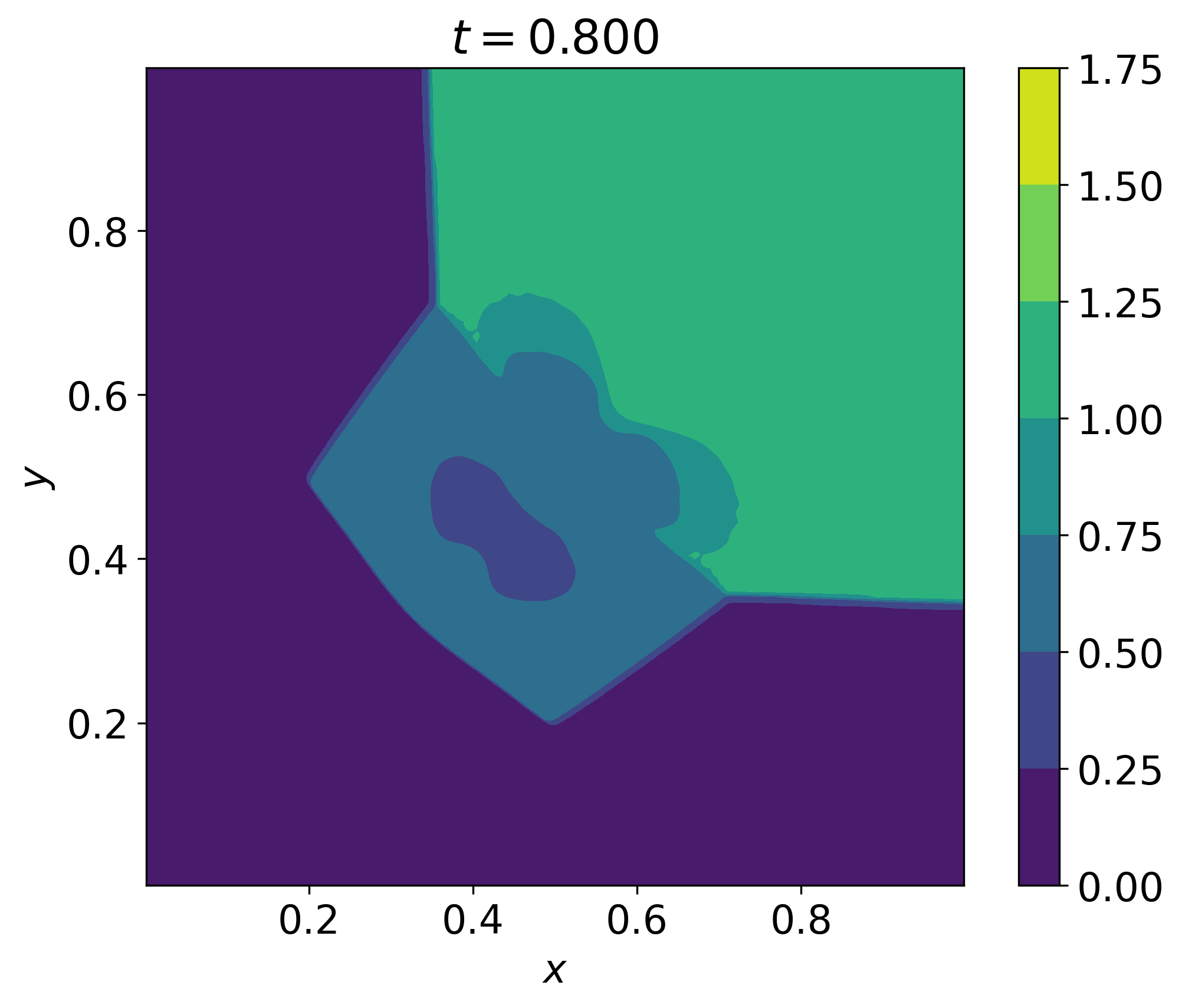}
        \caption{Testing, $p_{TR} = 1.125$}
    \end{subfigure}
    \caption{Pressure snapshots for the 2D Euler equations corresponding to different initial values of $p_{TR}$.}
    \label{fig:euler_snapshots}
\end{figure}

We collect FOM snapshots using the open-source Python library \texttt{pressio-demoapps}\footnote{\href{https://pressio.github.io/pressio-demoapps/index.html}{pressio.github.io/pressio-demoapps}}
to simulate \cref{eq:euler_pde}, which uses a cell-centered finite volume scheme.
For this example, we use a $256\times256$ uniform Cartesian mesh, resulting in a FOM with state dimension $n_q = 256\times256\times 4 = 262,144$, and a Weno5 scheme for inviscid flux reconstruction.
The FOM time stepping is done using \texttt{pressio-demoapps}' SSP3 scheme for times $t\in (0, 0.8)$ with time step $\Delta t=0.001$, while the ROM is integrated with BDF time stepping.
The first $2000$ normalized POD singular values are plotted in \Cref{fig:euler_singular_values}; the slow decay indicates the high difficulty of the problem for POD-based methods.

\begin{figure}[t]
    \centering
    \includegraphics[width=0.5\textwidth]{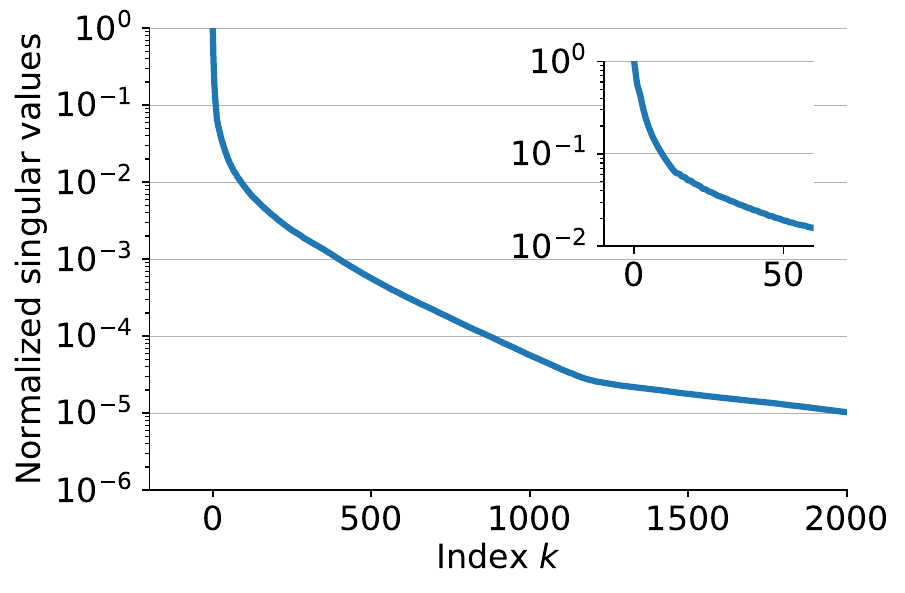}
    \caption{First $2000$ normalized singular values $\sigma_k/\sigma_1$ for the 2D Euler example. The inset figure shows the first $60$ normalized singular values to show where the singular value decay begins to slow.}
    \label{fig:euler_singular_values}
\end{figure}

Before computing ROMs, the FOM state variables are first transformed via the map
\begin{align}\label{eq:specific_volume_transformation}
    \begin{bmatrix}
        \rho \\ \rho u \\ \rho v \\ \rho E
    \end{bmatrix}
    \mapsto
    \begin{bmatrix}
        u \\ v \\ p \\ \zeta
    \end{bmatrix},
\end{align}
where $\zeta = 1/\rho$ is the specific volume.
A discretized FOM using the specific volume formulation is purely quadratic,
\begin{align}\label{eq:euler_quadratic_fom}
    \dt \bq(t) = \bH [\bq(t)\otimes\bq(t)],
    \qquad
    \bq(0) = \bq_0(p_{TR}).
\end{align}
For further details, see, e.g., \cite{EQian_BKramer_BPeherstorfer_KWillcox_2020a}.
This FOM is not formed explicitly, but it motivates an appropriate structure for feature map Kernel ROMs using POD or QM.
In both cases, we set $\obq$ to the average training snapshot and apply the kernel input normalization from \Cref{remark:kernel_normalization}, leading to a POD ROM structure
\begin{align}\label{eq:euler_galerkin_rhs_pod}
    \dt\hbq(t)
= \hbc + \hbA \hbq(t) + \hbH [\hbq(t) \otimes \hbq(t)],
\end{align}
whereas the QM ROMs have the quartic form
\begin{align}\label{eq:euler_galerkin_rhs_qm}
    \dt\hbq(t)
= \hbc + \hbA \hbq(t) + \hbH_2 [\hbq(t) \otimes \hbq(t)]
    + \hbH_{3}[\hbq(t)\otimes\hbq(t)\otimes\hbq(t)]
    + \hbH_{4}[\hbq(t)\otimes\hbq(t)\otimes\hbq(t)\otimes\hbq(t)].
\end{align}
Since we use \texttt{pressio-demoapps} to collect FOM data, this example only considers the purely non-intrusive cases.
That is, we do not compute intrusive ROMs for this problem and do not evaluate the \emph{a posteriori} error bound as in the previous examples.

The POD and QM OpInf ROMs are constructed to have the same structure as \cref{eq:euler_galerkin_rhs_pod} and \cref{eq:euler_galerkin_rhs_qm}, respectively.
Notice that this is the same structure as for Burgers' equation.
Consequently, the feature map Kernel ROMs use the same feature maps $\bphi$ and weighting matrices $\bG$ as in  \Cref{tbl:burgers_kernel_table}.
As in both previous examples, the RBF Kernel ROMs use a Gaussian RBF kernel with fixed shape parameter $\epsilon=10^{-1}$.
The hybrid Kernel ROMs use the sum of the kernel induced by the POD feature map from \Cref{tbl:burgers_kernel_table} with weighting coefficient $c_\phi=1$ and the same Gaussian RBF kernel with $\epsilon=0.1$ and weighting coefficient $c_\psi=10^{-3}$, resulting in a right-hand side of the form \cref{eq:burgers_hybrid_rhs}.
The error metric that we consider is the relative $L^\infty$-$L^1$ norm
\begin{align}\label{eq:relative_linf_l1_error}
    \be(\bq, \hbq) = \frac{\max_{k} \; \norm{\bq(t_k) - \bg(\hbq(t_k))}_1}{\max_{k} \; \norm{\bq(t_k)}_1},
\end{align}
The $L^1$ norm is more appropriate than $L^2$ for this problem due to the discontinuities in the solution.

We plot the error \cref{eq:relative_linf_l1_error} versus the reduced dimension $r$ for the POD OpInf, feature map Kernel, RBF Kernel, and hybrid Kernel ROMs in \Cref{fig:euler_pod_romsize_vs_error}.
For $r=5, 10, 15$, each of the ROMs obtain nearly identical performance.
For $r>15$, the projection error and the Kernel ROM errors plateau, with the Kernel ROMs yielding a $<2\%$ difference in error compared to the projection error.
The Hybrid and FM Kernel ROMs have nearly identical errors, while the RBF yields slightly different but very similar errors.
The OpInf ROM increases slightly in error for $r>15$, yet still obtains errors within a few percent of the projection error.
We note that the plateauing of the ROM and projection errors for the tested ROM sizes is expected since the singular value decay is slow, as shown in \Cref{fig:euler_singular_values}.

\begin{figure}[t]
    \centering
    \includegraphics[width=0.9\textwidth]{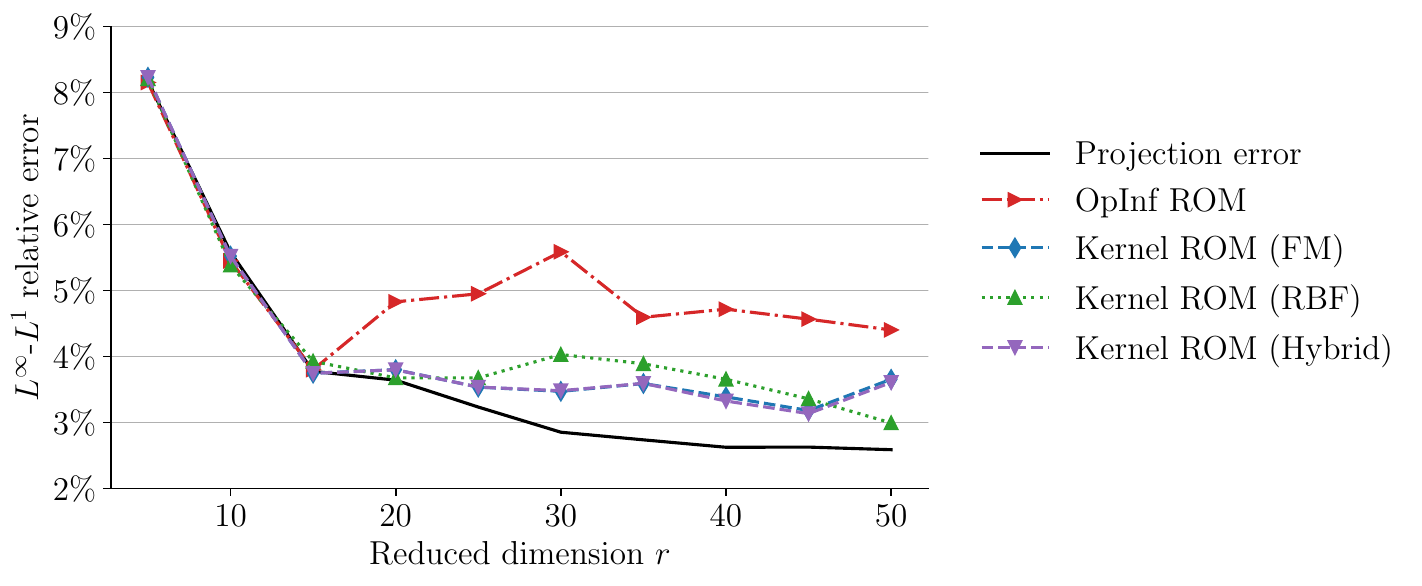}
    \caption{Relative errors for several non-intrusive POD ROMs for the Euler--Riemann problem \cref{eq:euler_pde}.}
    \label{fig:euler_pod_romsize_vs_error}
\end{figure}

We omit a similar comparison for the QM ROMs for this problem because the resulting ROMs are highly dependent on the QM regularization $\rho$, and require very large values of $\rho$ to obtain a stable ROM.
To illustrate this, we compute QM Kernel FM ROMs for $r=10, 20$ for QM regularizations $\rho \in \set{10^0, 10^1, \dots, 10^{12}}$ and plot the resulting errors, see \Cref{fig:euler_qm_error_vs_reg}.
For $r=10$, we observe that the QM Kernel ROM errors are very large for $\rho<10^{10}$, whereas the corresponding QM projection errors are relatively small.
The QM ROM errors do not approach the QM projection errors until $\rho = 10^{11}$, where a slightly better error compared to POD is achieved.
For $r=20$, the QM ROMs for $\rho < 10^8$ are unstable and do not finish the time integration, while for $\rho =10^8, 10^9, 10^{10}$, the QM ROM errors still exceed the POD errors.
The QM ROMs for $\rho=10^{11}, 10^{12}$ obtain yield the best errors, but because the QM regularization $\rho$ is so large, the resulting ROM errors are no better than POD.

\begin{figure}[t]
    \centering
    \includegraphics[width=\textwidth]{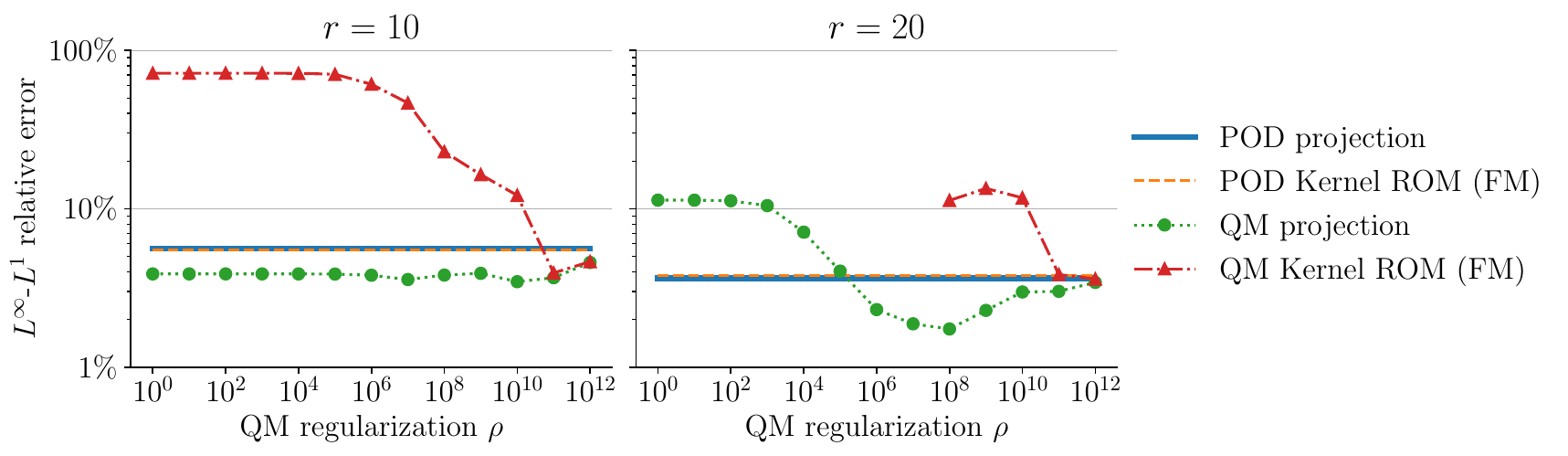}
    \caption{QM regularization $\rho$ versus relative $L^\infty$-$L^1$ error.}
    \label{fig:euler_qm_error_vs_reg}
\end{figure}

\section{Conclusion}\label{sec:conclusion}

This paper develops a novel non-intrusive model reduction technique grounded in regularized kernel interpolation.
While previous approaches approximate the ROM dynamics by solving a data-driven polynomial regression problem,
our approach yields an optimal approximant to the ROM dynamics from an RKHS, which is determined by the choice of kernel.
In particular, using kernels induced by feature maps allows one to imbue interpretable structure into the resulting ROM.
Furthermore, using an RBF kernel or a hybrid approach using the sum of a feature map and an RBF kernel allows one to compute effective non-intrusive ROMs that incorporate no structure or partial structure.
The hybrid approach also provides a natural way of incorporating closure terms into our ROM formulation, and this approach was demonstrated to be effective in each of the numerical examples.
Since the approximant lives in an RKHS, we can leverage the pointwise error bound from \Cref{thm:power_function_bound}, a standard result from RKHS theory, as well as standard intrusive ROM error estimates to derive an \emph{a posteriori} error estimate for our Kernel ROMs in \Cref{thm:a_posteriori_error}.
This error estimate, as well as the added flexibility afforded by arbitrary choices of kernel, are key innovations of our approach.

Future work will focus on expanding the applicability and efficiency of Kernel ROMs.
In particular, we will extend our approach to problems where the FOM right-hand side $\bff$ is parametrized, which is the case in many engineering applications of interest.
Second, we will implement a greedy sampling procedure to build a minimal training set for the kernel interpolants.
This is particularly relevant when using an RBF interpolant, since the computation cost of evaluating the RBF interpolant is proportional to the amount of training data whenever the kernel is not entirely prescribed by feature maps.
Third, we will develop a method for non-intrusively approximating the \emph{a posteriori} error bound in \Cref{thm:a_posteriori_error}.
As mentioned in \Cref{sec:numerics}, evaluating the bound \cref{eq:error_estimate} requires access to the FOM right-hand side $\bff$, which we assume that we cannot access in the fully non-intrusive setting.
Therefore, in future work, it will be necessary to develop an accurate estimator for the quantities in \cref{eq:error_functions}.

\section*{Acknowledgements}

S.A.M.~was supported in part by the John von Neumann postdoctoral fellowship, a position at Sandia National Laboratories sponsored by the Applied Mathematics Program of the U.S.~Department of Energy Office of Advanced Scientific Computing Research.
Sandia National Laboratories is a multi-mission laboratory managed and operated by National Technology \& Engineering Solutions of Sandia, LLC (NTESS), a wholly owned subsidiary of Honeywell International Inc., for the U.S.~Department of Energy's National Nuclear Security Administration~(DOE/NNSA) under contract DE-NA0003525. This written work is authored by an employee of NTESS. The employee, not NTESS, owns the right, title and interest in and to the written work and is responsible for its contents. Any subjective views or opinions that might be expressed in the written work do not necessarily represent the views of the U.S.~Government.

\begin{appendices}
\section{Quadratic systems with QM approximations}
\label{appendix:fullquadstructure}

This appendix considers a linear-quadratic FOM,
\begin{align}
    \tag{\ref{eq:linearquadratic_fom}}
    \dt\bq(t)
    = \bff(\bq(t))
    \coloneqq \bA\bq(t) + \bH[\bq(t) \otimes \bq(t)],
\end{align}
and derives the structure of the corresponding intrusive projection-based ROM with a QM approximation,
\begin{align}
    \tag{\ref{eq:quadratic_decoder}}
    \bg(\tbq) = \obq + \bV\tbq + \bW[\tbq \otimes \tbq],
\end{align}
for nonzero $\obq\in\real^{n_q}$ and $\bW\in\real^{n_q\times r^2}$.
Specifically, we show that a nonzero reference vector $\obq$ causes a constant term to appear in the ROM dynamics.

Using the product rule $(\bX\otimes\bY)(\bZ\otimes\bU) = (\bX\bZ)\otimes(\bY\bU)$, we have
\begin{align*}
    \bg(\tbq) \otimes \bg(\tbq)
    &= \big(\obq + \bV\tbq + \bW[\tbq \otimes \tbq]\big) \otimes \big(\obq + \bV\tbq + \bW[\tbq \otimes \tbq]\big)
    \\
    &= \obq\otimes\obq + \obq\otimes(\bV\tbq) + \obq\otimes(\bW[\tbq \otimes \tbq])
    \\&\qquad
    + (\bV\tbq)\otimes\obq + (\bV\tbq)\otimes(\bV\tbq) + (\bV\tbq)\otimes(\bW[\tbq \otimes \tbq])
    \\&\qquad
    + (\bW[\tbq \otimes \tbq]) \otimes \obq + (\bW[\tbq \otimes \tbq]) \otimes (\bV\tbq) + (\bW[\tbq \otimes \tbq]) \otimes (\bW[\tbq \otimes \tbq])
    \\
    &= \obq\otimes\obq + (\obq\otimes\bV)\tbq + (\obq\otimes\bW)[\tbq \otimes \tbq]
    \\&\qquad
    + (\bV \otimes\obq)\tbq + (\bV \otimes \bV)[\tbq\otimes\tbq] + (\bV\otimes\bW)[\tbq \otimes \tbq \otimes \tbq]
    \\&\qquad
    + (\bW\otimes\obq)[\tbq \otimes \tbq] + (\bW\otimes \bV)[\tbq \otimes \tbq \otimes \tbq] + (\bW \otimes \bW)[\tbq \otimes \tbq \otimes \tbq \otimes \tbq]
    \\
    &= \obq\otimes\obq + (\obq\otimes\bV + \bV\otimes\obq)\tbq + (\obq\otimes\bW + \bV \otimes \bV + \bW\otimes\obq)[\tbq \otimes \tbq]
    \\&\qquad
    + (\bV\otimes\bW + \bW\otimes \bV)[\tbq \otimes \tbq \otimes \tbq]
    + (\bW \otimes \bW)[\tbq \otimes \tbq \otimes \tbq \otimes \tbq].
\end{align*}
Therefore,
\begin{align*}
    \bff(\bg(\tbq))
    &= \bA\big(\obq + \bV\tbq + \bW[\tbq \otimes \tbq]\big)
    + \bH[\bg(\tbq) \otimes \bg(\tbq)]
    \\
    &= \bA\obq + \bH[\obq\otimes\obq] + \big(\bA\bV + \bH(\obq\otimes\bV + \bV\otimes\obq)\big)\tbq
    \\&\qquad
    + \big(\bA\bW + \bH(\obq\otimes\bW + \bV \otimes \bV + \bW\otimes\obq)\big)[\tbq \otimes \tbq]
    \\&\qquad
    + \bH(\bV\otimes\bW + \bW\otimes \bV)[\tbq \otimes \tbq \otimes \tbq]
    + \bH(\bW \otimes \bW)[\tbq \otimes \tbq \otimes \tbq \otimes \tbq],
\end{align*}
so that the intrusive projection-based ROM \cref{eq:generic_rom-withgandh} can be written as
\begin{subequations}\label{eq:intrusive-quartic-apdx}
\begin{align}
    \begin{aligned}
    \dt\tbq(t)
    &= \bV\trp\bff(\bg(\tbq(t)))
    \\
    &= \tbc + \tbA \tbq(t) + \tbH_2 [\tbq(t) \otimes \tbq(t)]
    + \tbH_{3}[\tbq(t)\otimes\tbq(t)\otimes\tbq(t)]
    + \tbH_{4}[\tbq(t)\otimes\tbq(t)\otimes\tbq(t)\otimes\tbq(t)],
    \end{aligned}
\end{align}
where
\begin{align}
    \begin{aligned}
    \tbc &= \bV\trp(\bA\obq + \bH[\obq\otimes\obq]) \in \real^{r},
    \\
    \tbA &= \bV\trp\bA\bV + \bV\trp\bH(\obq\otimes\bV + \bV\otimes\obq) \in \real^{r\times r},
    \\
    \tbH_{2} &= \bV\trp\bA\bW + \bV\trp\bH(\obq\otimes\bW + \bV \otimes \bV + \bW\otimes\obq)\in\real^{r\times r^2},
    \\
    \tbH_{3} &= \bV\trp\bH(\bV\otimes\bW + \bW\otimes \bV)\in\real^{r\times r^{3}},
    \\
    \tbH_{4} &= \bV\trp\bH(\bW \otimes \bW) \in\real^{r\times r^4}.
    \end{aligned}
\end{align}
\end{subequations}

The quartic polynomial structure of \cref{eq:intrusive-quartic-apdx} also arises when $\obq = \bf0$ but a Kernel ROM is constructed with the input scaling preprocessing step of \Cref{remark:kernel_normalization}.
In that case, the matrices in \cref{eq:intrusive-quartic-apdx} reduce to
\begin{align*}
\tbA &= \bV\trp\bA\bV,
    &
    \tbH_{2} &= \bV\trp\bA\bW + \bV\trp\bH(\bV \otimes \bV),
    \\
    \tbH_{3} &= \bV\trp\bH(\bV\otimes\bW + \bW\otimes \bV),
    &
    \tbH_{4} &= \bV\trp\bH(\bW \otimes \bW),
\end{align*}
with $\tbc = \bf0$.
However, the Kernel ROM targets a shifted and scaled reduced state $\hbq(t) = \bSigma^{-1}(\tbq(t) - \obx)$ for some $\bSigma\in\real^{r\times r}$ and $\obx\in\real^{r}$, which evolves according to
\begin{align*}
    \dt\hbq(t)
    &= \dt\left[\bSigma^{-1}\tbq(t) - \bSigma^{-1}\obx\right]
    \\
    &= \bSigma^{-1}\big(\tbA \tbq(t) + \tbH_2 [\tbq(t) \otimes \tbq(t)]
    + \tbH_{3}[\tbq(t)\otimes\tbq(t)\otimes\tbq(t)]
    + \tbH_{4}[\tbq(t)\otimes\tbq(t)\otimes\tbq(t)\otimes\tbq(t)]\big)
    \\
    &= \bSigma^{-1}\big(\tbA (\bSigma\hbq(t) + \obx) + \tbH_2 [(\bSigma\hbq(t) + \obx) \otimes (\bSigma\hbq(t) + \obx)]
    \\&\qquad\qquad
    + \tbH_{3}[(\bSigma\hbq(t) + \obx)\otimes(\bSigma\hbq(t) + \obx)\otimes(\bSigma\hbq(t) + \obx)]
    \\&\qquad\qquad
    + \tbH_{4}[(\bSigma\hbq(t) + \obx)\otimes(\bSigma\hbq(t) + \obx)\otimes(\bSigma\hbq(t) + \obx)\otimes(\bSigma\hbq(t) + \obx)]\big)
    \\
    &= \hbc + \hbA \hbq(t) + \hbH_2 [\hbq(t) \otimes \hbq(t)]
    + \hbH_{3}[\hbq(t)\otimes\hbq(t)\otimes\hbq(t)]
    + \hbH_{4}[\hbq(t)\otimes\hbq(t)\otimes\hbq(t)\otimes\hbq(t)],
\end{align*}
where
\begin{align*}
    \hbc
    &= \bSigma^{-1}\big(\tbA\obx
    + \tbH_2 [\obx \otimes \obx]
    + \tbH_{3}[\obx \otimes \obx \otimes \obx]
    + \tbH_{4}[\obx \otimes \obx \otimes \obx \otimes \obx]\big)
    \in\real^{r},
    \\
    \hbA
    &= \bSigma^{-1}\big(
        \tbA\bSigma
        + \tbH_2(\bSigma\otimes\obx + \obx\otimes\bSigma)
        + \tbH_3(
            \bSigma\otimes\obx\otimes\obx
            + \obx\otimes\bSigma\otimes\obx
            + \obx\otimes\obx\otimes\bSigma
        )
    \\&\qquad\qquad
    + \tbH_4(
        \bSigma\otimes\obx\otimes\obx\otimes\obx
        + \obx\otimes\bSigma\otimes\obx\otimes\obx
        + \obx\otimes\obx\otimes\bSigma\otimes\obx
        + \obx\otimes\obx\otimes\obx\otimes\bSigma
    )
    \big)\in\real^{r\times r},
    \\
    \hbH_{2}
    &= \bSigma^{-1}\big(
        \tbH_2(\bSigma \otimes \bSigma)
        + \tbH_3(
            \bSigma\otimes\bSigma\otimes\obx
            + \bSigma\otimes\obx\otimes\bSigma
            + \obx\otimes\bSigma\otimes\bSigma
        )
    \\&\qquad\qquad
    + \tbH_4(
        \bSigma\otimes\bSigma\otimes\obx\otimes\obx
        + \bSigma\otimes\obx\otimes\bSigma\otimes\obx
        + \bSigma\otimes\obx\otimes\obx\otimes\bSigma
        \\&\qquad\qquad\qquad
        + \obx\otimes\bSigma\otimes\bSigma\otimes\obx
        + \obx\otimes\bSigma\otimes\obx\otimes\bSigma
        + \obx\otimes\obx\otimes\bSigma\otimes\bSigma
    )
    \big)\in\real^{r\times r^2},
    \\
    \hbH_{3}
    &= \bSigma^{-1}\big(
        \tbH_{3}(\bSigma\otimes\bSigma\otimes\bSigma)
    \\&\qquad\qquad
        + \tbH_{4}(
            \bSigma\otimes\bSigma\otimes\bSigma\otimes\obx
            + \bSigma\otimes\bSigma\otimes\obx\otimes\bSigma
            + \bSigma\otimes\obx\otimes\bSigma\otimes\bSigma
            + \obx\otimes\bSigma\otimes\bSigma\otimes\bSigma
        )
    \big)\in\real^{r\times r^3},
    \\
    \hbH_{4} &= \bSigma^{-1}\tbH_{4}(\bSigma\otimes\bSigma\otimes\bSigma\otimes\bSigma)
    \in\real^{r\times r^{4}}.
\end{align*}
The salient point is that none of these matrices need to be constructed explicitly when using a non-intrusive model reduction method: only the desired structure is needed to design the non-intrusive ROM.

\section{Stability for linear systems}
\label{appendix:lti_stability_error}

The following stability result illustrates the importance of the regularization hyperparameter $\rho \ge 0$ when solving the minimization problem \cref{eq:qm_optimization_problem} for computing $\bW$.
Applying the QM approach with reference state $\obq=\bzero$ to a linear FOM
\begin{align}
    \dt\bq(t) = \bA\bq(t),
    \qquad
    \bq(0) = \bq_0(\bmu),
\end{align}
results in a ROM with quadratic dynamics
\begin{align}\label{eq:linear_qm_rom_appendix}
    \dt \tbq(t) = \tbA\tbq(t) + \tbH[\tbq(t)\otimes\tbq(t)],
    \qquad \tbq(0) &= \bV\trp\bq_0(\bmu),
\end{align}
where
$\tbA = \bV\trp \bA \bV$ and $\tbH = \bV\trp \bA \bW$.
We then have a stability estimate for the ROM solution.

\begin{proposition}\label{thm:lti_stability}
Let $\lambda$ denote the maximum eigenvalue of $\bA_{sym}=\frac{1}{2}(\bA+\bA\trp)$ (the symmetric part of $\bA$).
Then the following stability estimate for the QM ROM \cref{eq:linear_qm_rom_appendix} holds for all $t\in [0, T]$:
    \begin{align}\label{eq:stablity_result}
        \norm{\tbq(t)}_2\leq \norm{\bA}_2\norm{\bW}_2\int_0^t \|\tbq(s)\otimes\tbq(s)\|_2 e^{\lambda(t-s)}ds + e^{\lambda t} \norm{\bV\trp\bq_0}_2.
    \end{align}
\end{proposition}
\begin{proof}
    Observe that
    \begin{align}
        \label{eq:lti_stability_bound1}
        \begin{aligned}
        \tbq(t)\trp\dt\tbq(t)
        &= \tbq(t)\trp \bV\trp\bA \bV \tbq(t) + \tbq(t)\trp\bV\trp\bA\bW [\tbq(t)\otimes\tbq(t)]\\
        &= \tbq(t)\trp \bV\trp\bA_{sym} \bV \tbq(t) + \tbq(t)\trp\bV\trp\bA\bW [\tbq(t)\otimes\tbq(t)]\\
        &\leq \lambda\norm{\bV\tbq(t)}_2^2 + \norm{\bA}_2\norm{\bW}_2\norm{\bV\tbq(t)}_2\|\tbq(t)\otimes\tbq(t)\|_2\\
        &= \lambda\norm{\tbq(t)}_2^2 + \norm{\bA}_2\norm{\bW}_2\norm{\tbq(t)}_2\|\tbq(t)\otimes\tbq(t)\|_2,
        \end{aligned}
    \end{align}
    where the last line follows from the orthonormality of $\bV$.
    The bound \cref{eq:lti_stability_bound1} implies that
    \begin{align*}
        \dt\norm{\tbq(t)}_2 =\frac{\tbq(t)\trp\dt\tbq(t)}{\norm{\tbq(t)}_2} \leq
        \lambda\norm{\tbq(t)}_2 + \norm{\bA}_2\norm{\bW}_2\|\tbq(t)\otimes\tbq(t)\|_2
    \end{align*}
    Applying \Cref{lem:gronwall} with $u(t)=\norm{\tbq(t)}_2$,
    $\alpha(t)=\norm{\bA}_2\norm{\bW}_2\|\tbq(t)\otimes\tbq(t)\|_2$,
    and $\beta(t)=\lambda$ yields the result.
\end{proof}

\Cref{thm:lti_stability} indicates that the magnitude of $\norm{\bW}_2$ has a crucial impact on the stability of the resulting QM ROM.
Consequently, it is important to apply sufficient regularization (i.e., choose $\rho$ large enough) when computing $\bW$ to ensure that $\norm{\bW}_2$ remains small.

\section{Error estimate between intrusive ROM and Kernel ROM}
\label{appendix:intrusive_vs_kernel_error_estimate}

We derive an error result comparing the intrusive projection-based ROM solution $\tbq(t)$ and the Kernel ROM solution $\hbq(t)$.
Let $\hbe(t) = \tbq(t)-\hbq(t)$, which satisfies the ODE
\begin{align}
    \label{eq:rom_error_ode}
    \dt\hbe(t) = \tbf(\tbq(t)) - \hbf(\hbq(t)),
    \qquad
    \hbe(0) = \bzero.
\end{align}
We then have the following.

\begin{proposition}\label{prop:non-intrusive_vs_intrusive_error}
    Let $\hbM\in \real^{r\times r}$ be a symmetric positive definite weighting matrix with Cholesky factorization $\hbM=\hbL\hbL\trp$.
    If $\hbf$ is an unregularized kernel interpolant of $\tbf+\bdelta\in \cH_K^r$ where $\norm{\bdelta(\hbq(s))}_{\hbM}<\delta(s)$,
    then
    \begin{align}\label{eq:non-intrusive_vs_intrusive_error}
        \norm{\hbe(t)}_{\hbM}
        &\leq \int_0^t \left(P_{K, \tbQ}(\hbq(s)) \|\hbL\|_2 \|\tbf+\bdelta\|_{\cH_K^r}+\delta(s)\right)e^{\int_s\trp \Lambda_{\hbM}[\hbf](\hbq(\tau))d\tau} ds,
        & \forall \; t & \in (0, T).
    \end{align}
\end{proposition}
\begin{proof}
The dynamics in \cref{eq:rom_error_ode} can be rewritten as
\begin{align*}
    \dt\hbe(t) &= \tbf(\tbq(t))
    -\tbf(\hbq(t)) + \tbf(\hbq(t))
    + \bdelta(\hbq(t))
    - \hbf(\hbq(t)) - \bdelta(\hbq(t)).
\end{align*}
Taking the $\hbM$-weighted inner product with $\hbe(t)$ yields
\begin{align*}
    &\innerprod{\hbe(t), \dt\hbe(t)}_{\hbM}\\
    &=
    \innerprod{\hbe(t),\tbf(\tbq(t)) -\tbf(\hbq(t))}_{\hbM}
    + \innerprod{\hbe(t), \tbf(\hbq(t)) + \bdelta(\hbq(t)) - \hbf(\hbq(t)) }_{\hbM}
    - \innerprod{\hbe(t), \bdelta(\hbq(t))}_{\hbM} \\
    &\leq \Lambda_{\hbM}[\tbf](\hbq(t))\norm{\hbe(t)}_{\hbM}^2
    + \norm{\tbf(\hbq(t)) + \bdelta(\hbq(t)) - \hbf(\hbq(t))}_{\hbM}\norm{\hbe(t)}_{\hbM}
    +\delta(t)\norm{\hbe(t)}_{\hbM} \\
    &\leq \Lambda_{\hbM}[\tbf](\hbq(t))\norm{\hbe(t)}_{\hbM}^2
    + P_{K, \hbQ}(\hbq(t))\|\hbL\|_2\|\tbf+\bdelta\|_{\cH_K^r}\norm{\hbe(t)}_{\hbM}
    +\delta(t)\norm{\hbe(t)}_{\hbM}.
\end{align*}
Therefore
\begin{align*}
    \dt \norm{\hbe(t)}_{\hbM} &= \frac{\innerprod{\hbe(t), \dt\hbe(t)}_{\hbM}}{\norm{\hbe(t)}_{\hbM}}
    \leq \Lambda_{\hbM}[\tbf](\hbq(t))\norm{\hbe(t)}_{\hbM}
    + P_{K, \hbQ}(\hbq_N(t))\|\hbL\|_2\|\tbf+\bdelta\|_{\cH_K^r}
    +\delta(t).
\end{align*}
Applying \Cref{lem:gronwall} yields the result.
\end{proof}

\clearpage
\section{Main nomenclature}
\label{appendix:nomenclature}

\phantom{Notation}

\begin{tabular}{rl}
    \textbf{Kernel interpolation} \\
    $K:\real^{n_x}\times\real^{n_x}\to\real$ & symmetric kernel function \\
    $\cH_K$ & reproducing kernel Hilbert space \\
    $\bx_j\in\real^{n_x}$ & inputs for kernel interpolation \\
    $\by_j\in\real^{n_y}$ & outputs for kernel interpolation \\
    $\bfun:\real^{n_x}\to\real^{n_y}$ & function to interpolate: $\by_j = \bfun(\bx_j)$ \\
    $\gamma \ge 0$ & kernel regularization parameter \\
    $\bOmega\in \real^{m \times n_y}$ & coefficient matrix for kernel interpolation \\
    $\bs_\bfun^\gamma\in \cH_K^{n_y}$ & kernel interpolant of $\bfun$ with regularization $\gamma$ \\
    $\bpsi_{\!\epsilon}:\real^{n_x}\to\real^{m}$ & RBF kernel evaluation function \\
    $\bphi:\real^{n_x}\to\real^{n_\phi}$ & feature map \\
    $\bG\in\real^{n_\phi \times n_\phi}$ & weighting matrix for feature map kernels \\
    $\bC\in\real^{n_y\times n_\phi}$ & post-feature map kernel coefficients
\\[.0625in]

    \textbf{Full-order models} \\
    $\bq(t)\in\real^{n_q}$ & full-order model state \\
    $\bff:\real^{n_q}\to\real^{n_q}$ & full-order model dynamics function \\
    $\bmu\in\real^{n_\mu}$ & parameters for the initial condition \\
    $\bQ\in\real^{n_q\times M(n_t + 1)}$ & shifted state snapshot matrix (all trajectories)
\\[.0625in]

    \textbf{Reduced-order models} \\
    $\tbq(t)\in\real^{r}$ & intrusive reduced-order model state \\
    $\hbq(t)\in\real^{r}$ & non-intrusive reduced-order model state \\
    $\bV\in\real^{n_q\times r}$ & proper orthogonal decomposition (POD) basis matrix \\
    $\bW\in\real^{r\times r(r+1)/2}$ & quadratic manifold (QM) weight matrix \\
    $\bg:\real^{r}\to\real^{n_q}$ & decompression map \\
    $\bh:\real^{n_q}\to\real^{r}$ & compression map \\
$\be,\hbe$ & error quantities
\end{tabular}

\end{appendices}

\bibliographystyle{siamplain}
\bibliography{references}

\end{document}